%% file: main.tex
\documentclass[a4paper,UKenglish, numberwithinsect, cleveref, autoref, thm-restate, nolineno]{socg-lipics-v2021}
\hideLIPIcs  %uncomment to remove references to LIPIcs series (logo, DOI, ...), e.g. when preparing a pre-final version to be uploaded to arXiv or another public repository

\usepackage{mathproblem,xspace}
\usepackage{graphicx}

\usepackage{complexity}
\usepackage{amsmath,amssymb,amsthm}
\usepackage{hyperref,cleveref}
\usepackage{todonotes}
\usepackage{comment}

\usepackage{enumitem}
\setlist[enumerate,1]{label=\arabic*.,ref=\arabic*}
\setlist[enumerate,2]{label=\alph*.,ref=\arabic*.\alph*}
\setlist[enumerate,3]{label=\roman*.,ref=\arabic*.\alph*.\roman*}

\usepackage{tikz}
\usetikzlibrary{positioning,shapes,patterns,patterns.meta,calc,decorations.pathreplacing,
  decorations.pathmorphing}
\tikzstyle{vertex}=[draw, shape=circle, minimum size=2pt,inner sep=1.4pt, fill=black]
\tikzstyle{terminal}=[draw, shape=rectangle, minimum size=2pt,inner sep=1.6pt, fill=black]

\input{NewCommands}

\makeatletter
\g@addto@macro\bfseries{\boldmath}
\makeatother

\title{Parameterized Complexity of Power Network Design: Coordinating Cable Placement is Hard} %TODO Please add

\titlerunning{Parameterized Complexity of Power Network Design} %TODO optional, please use if title is longer than one line

\author{Thekla Hamm}{Eindhoven University of Technology, Eindhoven, the Netherlands}{t.l.s.hamm@tue.nl}{
https://orcid.org/0000-0002-4595-9982}{}

\author{Bart M. P. Jansen}{Eindhoven University of Technology, Eindhoven, the Netherlands}{b.m.p.jansen@tue.nl}{
https://orcid.org/0000-0001-8204-1268}{}

\author{Faezeh Motiei}{Eindhoven University of Technology, Eindhoven, the Netherlands}{f.motiei@tue.nl}{
https://orcid.org/0009-0006-8805-9307}{}

\authorrunning{T.~Hamm, B.M.P.~Jansen and F.~Motiei} %TODO mandatory. First: Use abbreviated first/middle names. Second (only in severe cases): Use first author plus 'et al.'

\Copyright{Thekla Hamm and Bart M.\ P.\ Jansen and Faezeh Motiei} %TODO mandatory, please use full first names. LIPIcs license is "CC-BY";  http://creativecommons.org/licenses/by/3.0/


 %TODO mandatory; please add comma-separated list of keywords

\category{} %optional, e.g. invited paper

\relatedversion{} %optional, e.g. full version hosted on arXiv, HAL, or other respository/website
\nolinenumbers %uncomment to disable line numbering

\EventEditors{}
\EventNoEds{2}
\EventLongTitle{42nd Conference on Very Important Topics (CVIT 2016)}
\EventShortTitle{CVIT 2016}
\EventAcronym{CVIT}
\EventYear{2025}
\EventDate{December 24--27, 2016}
\EventLocation{Little Whinging, United Kingdom}
\EventLogo{}
\SeriesVolume{42}
\ArticleNo{23}
\begin{document}

\maketitle

% For updates to abstract: point out cost-sharing for multiple trees with individual capacity constraints, and optionally depth bounds.

\begin{abstract}
We study several generalizations of the \textsc{Steiner Tree} problem that are motivated by the design of power networks. While \textsc{Steiner Tree} asks for a single minimum-cost tree that connects a given set of terminal vertices, a power network typically consists of \emph{multiple trees}. Each tree connects to a subset of the terminals, to avoid electrical overloads. The cost of installing a power network is therefore determined by two factors: the total length of the cables in the network and the cost of digging underground trenches into which the cables are placed. Since the digging costs can be substantial, to minimize the total cost of the network it might be necessary to place multiple cables into the same trench. These characteristics lead to variations of \textsc{Steiner Tree} in which the goal is to compute a minimum-cost set of Steiner trees, all with a common root, that together connect a given terminal set while balancing the power demand of the terminals in each tree. Two important variations arise depending on whether the network is intended for low-voltage or high-voltage power. In the low-voltage setting, there is substantial power loss across the cables which effectively means that the maximum depth of any tree in the solution has to be bounded. No such depth bound applies to the high-voltage setting.

We investigate the parameterized complexity of several power network design problems, using the number of terminals as the parameter. While this parameterization of the standard \textsc{Steiner Tree} problem is fixed-parameter tractable, many of our variants are W[1]-hard. For low-voltage networks (bounded-depth trees), we present an \XP-algorithm for planar inputs, which exploits a nontrivial bound on the treewidth of solution subgraphs. We provide an intricate reduction from \textsc{Grid Tiling} to establish that the resulting algorithm is tight under the Exponential Time Hypothesis. The \XP-algorithm extends to the high-voltage setting and to general graphs, albeit at a cost in the running time. For high-voltage networks, we prove that the problem remains W[1]-hard on planar graphs. Finally, we explore a variation of the cost model for sharing digging costs in which both problems become fixed-parameter tractable.
\end{abstract}

\clearpage

\section{Introduction}
\input{Introduction}

\section{Preliminaries}
\input{Preliminaries}

\section{Algorithms}
\input{Algorithms_Section}

\section{Hardness results}
\input{Hardness_Results}

\section{Conclusion}
\input{Conclusion}

% %%
% %% Bibliography
% %%

\bibliography{low-voltage}

\end{document}

%% file: NewCommands.tex
\newcommand{\LV}{\textsc{Low-voltage Network Design}\xspace}
\newcommand{\HV}{\textsc{High-voltage Network Design}\xspace}
\newcommand{\BLV}{\textsc{Low-voltage Network Design}\xspace}
\newcommand{\PreLV}{\textsc{Prefix-Sharing} \LV}
\newcommand{\BLVshort}{\textsc{LVND}\xspace}
\newcommand{\LVshort}{\textsc{LVND}\xspace}
\newcommand{\HVshort}{\textsc{HVND}\xspace}
\newcommand{\ULV}{\HV}
\newcommand{\ULVshort}{\textsc{HVND}\xspace}
\newcommand{\PreLVshort}{\textsc{P}\LVshort}
\newcommand{\GT}{\textsc{Grid Tiling}\xspace}
\newcommand{\GTRegular}{\textsc{Regular Grid Tiling}\xspace}
\newcommand{\dig}{\ensuremath{\operatorname{dig}}}
\newcommand{\Oh}{\ensuremath{\mathcal{O}}}
\newcommand{\connector}{\ensuremath{\mathcal{C}}}
\newcommand{\hconnector}{\ensuremath{\connector_R}}
\newcommand{\vconnector}{\ensuremath{\connector_C}}
\newcommand{\regpar}{\ensuremath{\delta}}

\newcommand{\PLVcost}[3]{\ensuremath{\text{cost}_{#3}^{\text{PLV}}}(#1,#2)}
\newcommand{\PLVcostWord}{PLV-cost\xspace}

\newcommand{\CAbs}[1]{\ensuremath{\mathcal{A}(#1)}}
\newcommand{\PartialCAbs}[2]{\ensuremath{\mathcal{A}_{#1}(#2)}}

\newcommand{\pfx}[3]{\ensuremath{\text{pfx}(#1,#2,#3)}}
\newcommand{\pfxset}[3]{\ensuremath{\text{pfx}(#1,#2,#3)}}

\input{reduction_commands}

%% file: reduction_commands.tex
\newcommand{\rcons}{\ensuremath{r_{\text{cons}}}}
\newcommand{\rselect}{\ensuremath{r_{\text{select}}}}
\newcommand{\blocklength}{\ensuremath{k+1}}
\newcommand{\hdigcost}{\ensuremath{nk^2(\blocklength)}}
\newcommand{\vertlength}{\ensuremath{2k^2n(\blocklength)+ 1}}
    \newcommand{\termLength}{\ensuremath{X}}
    \newcommand{\connectLength}{\ensuremath{L}}
    \newcommand{\rootLength}{\ensuremath{L'}}
    \newcommand{\srootLength}{\ensuremath{L''}}
    \newcommand{\rowNum}{\ensuremath{N}}
    \newcommand{\colNum}{\ensuremath{M}}
    \newcommand{\RstartIndex}{\ensuremath{\triangleright}}
    \newcommand{\RendIndex}{\ensuremath{\triangleleft}}
    \newcommand{\CstartIndex}{\ensuremath{\triangleright}}
    \newcommand{\CendIndex}{\ensuremath{\triangleleft}}

    \newcommand{\hedgesetCell}[3][]{\ensuremath{E^{#1,#2}_H(#3)}}
    \newcommand{\vedgesetCell}[3][]{\ensuremath{E^{#1,#2}_V(#3)}}
    \newcommand{\hedgeset}[1][]{\ensuremath{E_H(#1)}}
    \newcommand{\vedgeset}[1][]{\ensuremath{E_V(#1)}}

%% file: Introduction.tex
\label{sec:Introduction}
\textsc{Steiner Tree} is a classic NP-complete problem~\cite{GareyJ77,Karp1972} in network design with numerous applications~\cite{Cho01,DIMACS,Ivana21,PACE}. Given a graph~$G$ and a subset of its vertices called \emph{terminals}, the goal is to compute a minimum-weight subtree of~$G$ that contains all terminals. The study of \textsc{Steiner Tree} has led to a wide range of algorithmic techniques in parameterized complexity~\cite{BodlaenderCKN15, CyganNPPRW22, DreyfusW71, FominKLPS19, Fucs07, GroenlandNK24, JansenS24, LokshtanovN10, MarxPP18, Nederlof13, PilipczukPSL18}. In this work, we investigate the parameterized complexity of several generalizations of \textsc{Steiner Tree} concerning the design of power networks. They arose from a cooperation with an industrial partner that develops underground infrastructure in the form of low-voltage electrical power networks.

\subparagraph*{Background and Motivation}
A low-voltage power network distributes electricity from a transformer station to an urban neighborhood. A transformer station has a fixed number (typically less than a dozen) of power outputs that each feed one aluminum or copper power cable. Such a cable supplies power to multiple households and can fork off into a tree layout to do so. The maximum amount of power that can be drawn from any single output on a transformer station is limited, which is why the households have to be balanced over the power cables based on their demand. Due to power losses during transport, the \emph{depth} (maximum length of a root-to-leaf path) of the tree-structure formed by a power cable determines its maximum load. Hence, the design of a low-voltage power network involves deciding which households are connected to a common cable, as well as finding a suitable low-depth tree structure for each cable. A high-voltage network has a similar structure, but does not require a depth bound on the trees since losses on high-voltage connections are negligible.

The \emph{cost} of installing a network depends on two factors. Naturally, there is a cost for the material of the cables that scales linearly with their total length. A less obvious factor, which can be even more important in practice, is the cost of \emph{installing} the network, i.e., digging trenches in the ground in which the cables are placed. The cost of digging a unit-length trench depends heavily on the terrain: it is cheap when the surface is covered by grass or small pavement tiles that can temporarily be set aside, but it is expensive when the surface is covered by an asphalt road that has to be sliced open and repaired. When multiple cables are placed in the \emph{same} trench, the costs for digging the trench can be shared among all the cables placed inside it. Hence, it may be advantageous to route cables over slightly longer trajectories if this avoids digging through asphalt or allows cables to share a trench.

As the energy transition necessitates the development of better power networks that support a society that increasingly relies on electricity to operate cars, heat pumps, and cooking appliances~\cite{BertoldiLL16energy,europeanparliament2025grid}, designing power networks has important applications.

\subparagraph*{Problem Statement} We model the task of designing a power network as a graph problem. Each edge~$e$ of the input graph~$G$ represents a trench that can be dug to place cables underground. Each edge has a length~$\ell(e) \in \mathbb{N}$ which determines the material costs for placing a cable in the corresponding trench, and a separate cost~$\dig(e) \in \mathbb{N}$ for digging the trench. The transformer station is represented by a vertex~$r$ in~$G$, which serves as the root of the tree-structured cables. The power consumers to be connected are represented by \emph{terminal} vertices~$T$, which are dead ends in the network and have degree one in~$G$. Since power consumption differs per consumer, each terminal vertex~$u \in T$ has an associated power \emph{demand}~$d(u) \in \mathbb{N}$. Finally, the input specifies the number of power outputs~$t$ on the transformer station, the depth constraint~$\lambda$ of each tree-structured cable (in the low-voltage setting), and the maximum power demand~$\alpha$ that can be supplied by a single cable.

To turn the corresponding algorithmic task into a decision problem, there is a threshold value~$\beta$ for the total cost. An input of the problem then asks whether there is a feasible network of cost at most~$\beta$. In the cost function, we pay the digging cost $\dig(e)$ once for every edge~$e$ in the \emph{union} of the solution trees, and for each occurrence of an edge~$e$ in a solution tree, we additionally pay~$\ell(e)$ in cable costs. This leads to the following formalization, in which we employ the number of trees~$t$ and the number of terminals~$|T|$ as the parameters since they are typically much smaller than the total size of the network of potential trenches.

\begin{center}
    \begin{mathproblem}{\BLV (\BLVshort)}[\(t + |T|\)]
    \textit{\textbf{Instance:}} An undirected graph \(G\) with \emph{edge lengths} \(\ell \colon E(G) \to \mathbb{N}_0\) and \emph{digging costs} \(\dig \colon E(G) \to \mathbb{N}_0\), a \emph{root vertex} \(r \in V(G)\), a set of degree-\(1\) \emph{terminal} vertices \(T \subseteq V(G)\) with \emph{demands} \(d \colon T \to \mathbb{N}\), and integers \(\alpha\), \(\beta\), \(\lambda\), and \(t\), encoded in unary.
    
    \vspace{0.3cm}
    
    \textit{\textbf{Question:}} Does \(G\) contain \(t\) tree subgraphs \(H_1,\ldots,H_t\), each rooted at \(r\), such that:
        \begin{itemize}
            \item \(T\subseteq \bigcup_{i=1}^tV(H_i)\) [\emph{covering constraint}: each terminal is covered],
            \item each root-to-leaf path~$P$ in each tree~$H_i$ satisfies~$\sum_{e \in E(P)} \ell(e) \leq \lambda$ [\emph{depth} constraint],
            \item \(\sum\limits_{u\in V(H_i) \cap T}d(u) \leq \alpha\) for each \(1 \leq i \leq t\), [\emph{demand} constraint: the total demand of terminals in each tree~$H_i$ is at most \(\alpha\)], and
            \item \(\text{cost}(\{H_1, \dotsc, H_t\}) := \left (\sum\limits_{i=1}^t \ell(H_i) \right)+\dig(\bigcup_{i=1}^t H_i) \leq \beta\)? [\emph{cost constraint}]
    \end{itemize}
    \end{mathproblem}
\end{center}

The high-voltage variation of the problem arises from the above by setting~$\lambda = \infty$, so that the depth constraint plays no role. We refer to it as \HV (\HVshort). The assumption that integers are encoded in unary means we focus our analysis on how the size of the graph and the number of terminals affect the complexity of solving the problem. With input integers encoded in binary, \BLVshort remains NP-complete for~$t=1$ and~$|T|=1$ since it generalizes \textsc{Shortest Weight-Constrained Path}~\cite[ND~30]{GareyJ77}.

\subparagraph*{Our Contribution} We initiate the parameterized-complexity analysis of power network design, focusing on the parameterization by $t + |T|$. While \textsc{Steiner Tree}~\cite{DreyfusW71} is fixed-parameter tractable (\FPT) for the parameter~$|T|$, which easily extends to generalizations such as \textsc{(Prize-Collecting) Steiner Forest}, our problems turns out to be significantly harder. We prove that \BLVshort and \HVshort are W[1]-hard parameterized by~$t + |T|$, even when restricted to \emph{planar} graphs. Planar inputs are particularly relevant since the graph models potential locations for digging trenches on the surface of the earth. As in previous W[1]-hardness proofs for (planar) network design problems~\cite{BentertNRZ21,Chitnis23,ChitnisTW24}, our reductions start from \textsc{Grid Tiling}\footnote{\GT receives as input a \(k \times k\) matrix, each cell of which contains a set of pairs of integers.
The task is to decide whether a pair can be selected from each cell such that the first entries of selected pairs are equal in cells of the same row, while the second entries are equal for cells of the same column.} (\cite{Marx07a}, cf.~\cite[Sec.~14.4.1]{ParAlg}) and produce a grid-like graph in which the solution has to route a \emph{horizontal part} and a \emph{vertical part}, which have to interact in a prescribed way to meet the cost requirement. 

For \BLVshort (\Cref{thm:bddepth-ETHlowerbound}), we only utilize two trees ($t=2)$. By a suitable setting of the demand values for the terminals, we can ensure that the first solution tree~$H_1$ must contain a prescribed subset~$T_1$ of terminals, which are attached to the top row of the grid, whereas the other tree~$H_2$ contains the remainder, connected to the rightmost column. By a careful choice of depth constraint, we then ensure that the main part of tree~$H_1$ consists of horizontal paths through the grid, and the tree~$H_2$ consists of paths that are vertical apart from using some horizontal edges. To obtain a solution whose cost meets the desired threshold, each horizontal edge used by~$H_2$ must lie on a horizontal path in~$H_1$ so that some digging cost is shared among~$H_1$ and~$H_2$. By setting the cost of such edges based on feasible combinations of values in the input of \GT, we can ensure the correctness of the reduction. It shows that coordinating the cable placement for the two trees~$H_1, H_2$ is hard.

Our hardness reduction for \HVshort (\Cref{thm:lowerbound:unbounded}) is significantly more involved. Since solutions to the high-voltage problem are not bounded in depth, we have to insert additional terminals inside the grid to enforce the solution to consist of separate horizontal and vertical parts. Embedding these terminals in a planar way while still allowing a solution sufficient degrees of freedom to model \textsc{Grid Tiling} is difficult. We overcome this obstacle by encoding the vertical and horizontal choices using \emph{two} trees each, leading to~$t=4$.

By known lower bounds~\cite[Theorem 14.28]{ParAlg} on the running time for solving \GT, our first reduction implies that under the Exponential Time Hypothesis (ETH) \BLVshort on planar input graphs cannot be solved in time~$f(|T|,t)\cdot (|V(G)| \cdot \lambda^{f'(t)})^{o(|T|)}$ for any computable functions \(f\) and \(f'\). We give an algorithm for planar inputs whose running time matches this bound, thereby showing that the \emph{square-root phenomenon}~\cite{Marx13} does not manifest itself for this problem. The algorithm works by enumerating a bounded-size set of choices for a certain \emph{abstract structure}~$\CAbs{\mathcal{S}}$ of the desired solution~$\mathcal{S} = \{H_1, \ldots, H_t\}$ and encoding the search for the best solution matching that structure as a \textsc{Valued Binary Constraint Satisfaction Problem} (CSP).  
We give an exchange argument to prove that if there is a solution, there is one whose abstract structure has~$\Oh(|T|^2)$ vertices. When the input graph~$G$ is planar, then the abstract structure (which is a minor of~$G$) must be as well, and therefore its treewidth is~$\Oh(|T|)$. The treewidth of the abstract structure governs the treewidth of the primal graph of the CSP, which yields the claimed running time using known CSP algorithms. A similar approach also yields an \XP-algorithm if the input graph~$G$ is not planar, and for the high-voltage problem, albeit at a cost in the running time. 

Our lower bounds show that even when the parameterization of number of terminals and solution trees is bounded, power network design is intractable due to the difficulties of coordinating cable placement. These difficulties effectively arise when two solution trees leave the root node in different directions, but later converge to place some of their cables in a common trench. As our last contribution, we present an alternative model for the cost computation in which digging costs can only be shared along \emph{prefixes} of the solution trees. We develop an \FPT-algorithm for the resulting \PreLV problem parameterized by~$|T| + t$ by adapting Dreyfus-Wagner~\cite{DreyfusW71}, thereby shedding further light on which aspect of the cost sharing model leads to intractability.

\subparagraph*{Related Work} In \textsc{Strongly Connected Steiner Subgraph}, the goal is to find a minimum-weight subgraph of a directed graph that is \emph{strongly connected} and contains all terminals. This problem, and several variations, are known to be W[1]-hard parameterized by the number of terminals~\cite{FeldmannM23,Suchy16}. However, the direction of the edges plays a crucial role in this hardness. Common variations of \textsc{Steiner Tree} on \emph{undirected} graphs tend to be \FPT\ parameterized by~$|T|$. We are not aware of earlier work on the parameterized complexity of network design problems in which costs can be shared. But, in the operations research community, there has been research~\cite{Alvarez-Miranda17} on the \textsc{Minimum-Cost Shared Arborescence} problem, which also incorporates this cost-sharing aspect. It lacks other features of \LV though: most importantly, the depth constraint.  
When it comes to network design problems involving costs and depth-constraints, there has been some work on the \textsc{Shallow-Light Steiner Tree} problem, which aims to compute a low-diameter Steiner tree of minimum cost~\cite{ChimaniS15,HajiaghayiKS09,KortsarzP97}. This problem formulation, therefore, incorporates the depth constraint, but not the cost-sharing feature.

\subparagraph*{Organization} After presenting basic preliminaries in \cref{sec:Preliminaries}, we provide our algorithmic results in \cref{sec:Algorithms}. The complementary hardness proofs are given in \cref{sec:Hardness}. We conclude in \cref{sec:Conclusion}.

%% file: Preliminaries.tex
 \label{sec:Preliminaries}
	For \(i \in \mathbb{N}\), we denote by \([i]\) the set \(\{1, \dotsc, i\}\). For sets \(A,B\) we denote by \(2^A\) the powerset of \(A\), and by \(A \times B\) the set \(\{(a,b) \mid a \in A \land b \in B\}\). Given a total function \(f\colon A \to B\) and a subset \(A' \subseteq A\), we define \(f(A')=\{f(a)\mid a\in A'\}\). 
We assume familiarity with graph theory~\cite{Diestel} and parameterized complexity theory~\cite{ParAlg} and use standard terminology and notation from both. Unless stated otherwise, we consider undirected simple graphs and denote an edge between vertices \(u\) and \(v\) by \(uv = vu\). The \emph{interior} of a path \(P\) is the subgraph induced by its non-endpoint vertices. 
For a rooted tree \((T,r)\) with edge weights \(w \colon E(T) \to \mathbb{N}\), its \emph{\(w\)-depth} is \(\max_{t \in V(T)} \sum_{e \in E(P_{r,t})} w(e)\), where~$P_{r,t}$ is the unique path from~$r$ to~$t$ in~$T$. When $w$ is clear from context, we simply refer to this quantity as depth. 
For brevity, for \(w \colon E(G) \to \mathbb{N}\) and \(H \subseteq G\), we denote \(w(H) = \sum_{e \in E(H)} w(e)\). For a set of subgraphs \(S=\{H_1,\ldots,H_t\}\) and a subgraph \(G'\), the cost of \(S\) in \(G'\) is defined as \(\text{cost}_{G'}(H_1,\ldots,H_t)=\text{cost}(H_1\cap G',\ldots,H_t\cap G')\).
We denote the treewidth of a graph \(G\) by \(\operatorname{tw}(G)\).
	
	For our algorithms, we use two results from the literature as subroutines; one for certain path computations, and one for solving a certain \emph{constraint satisfaction problem (CSP)}.
	
    \begin{theorem}[{\cite{JOKSCH1966191}}]
    \label{thm:cheapestshortpath}
    Given a directed graph \(G\), two functions \(f,g\colon E(G)\to \mathbb{N}\), two vertices \(s,t\in V(G)\), and a threshold \(\tau \in \mathbb{N}\), there is an algorithm that computes a path~$P$ from \(s\) to \(t\) that satisfies \(\sum_{e\in E(P)}f(e) \leq \tau\) for which \(\sum_{e\in E(P)}g(e)\) is minimized, if one exists, in time \(\mathcal{O}(\tau\cdot|V(G)|\cdot |E(G)|)\). If no such path exists, the algorithm reports failure.
\end{theorem}
	
	A \emph{valued binary CSP} is a tuple \((\mathcal{V}, \mathcal{D}, \mathcal{C})\) where \(\mathcal{V}\) and \(\mathcal{D}\) are two finite sets called the set of \emph{variables} and the \emph{domain} respectively.
	The third component \(\mathcal{C}\) is a finite set of \emph{value constraints} which are tuples \((x,y,c)\) with \(x,y \in \mathcal{V}\) and \(c \colon \mathcal{D} \times \mathcal{D} \to \mathbb{N}\).
	Given an \emph{assignment} \(\phi \colon \mathcal{V} \to \mathcal{D}\), its \emph{value} is \(\sum_{(x,y,c) \in \mathcal{C}} c(\phi(x),\phi(y))\).
	The \emph{primal graph} \(G_P(\mathcal{V}, \mathcal{D}, \mathcal{C})\) of a valued binary CSP \((\mathcal{V}, \mathcal{D}, \mathcal{C})\) is the graph \((\mathcal{V}, \{xy \mid (x,y,c) \in \mathcal{C} \text{ for some } c \colon \mathcal{D} \times \mathcal{D} \to \mathbb{N} \}) \).
	\begin{theorem}[folklore, e.g.\ \cite{FreuderKtree}]
		\label{thm:vbcsp}
		Given a valued binary CSP \((\mathcal{V}, \mathcal{D}, \mathcal{C})\), an assignment of minimum value can be computed in time \(|\mathcal{V}|\cdot|\mathcal{D}|^{\mathcal{O}(\operatorname{tw}(G_P(\mathcal{V}, \mathcal{D}, \mathcal{C}))}\).
	\end{theorem}

    The following lemma can be used to show that for a given partition \(\mathcal{P}\) of terminals, one can define a demand function and a value \(\alpha\) such that \(\mathcal{P}\) is the unique partition into sets of total demand at most~$\alpha$. This will be used later, in Section \ref{subsec:W1HardULV}, to prove that both \BLV and \ULV are W[1]-hard.
    
    \input{Terminal_Partitioning_Lemma}

%% file: Terminal_Partitioning_Lemma.tex
\begin{lemma}        \label{lemma:UniqueTerminalPartitioning}
        There exists a polynomial-time algorithm that, given a set of items \(U\) and a partition \(\mathcal{P}\) of \(U\) into \(t\) sets \(\mathcal{P}_1,\ldots,\mathcal{P}_t\), computes a function \(f\colon U\to \left[|U|^{t+2}\right]\) such that \(\mathcal{P}\) is the unique partition of \(U\) into \(t\) sets that satisfies \(\sum\limits_{x\in \mathcal{P}_i}f(x)\leq \frac{1}{t} \sum\limits_{x\in U}f(x)\) for each \(i \in [t]\).
    \end{lemma}
    \begin{proof}
        For each \(i\in[t]\), let \(\epsilon_i=|U|^{i-1}\) and \(\Omega_i=|U|^{t+2}-(|\mathcal{P}_i|-1)\cdot \epsilon_i\). Set \(f\) as \(\Omega_i\) for one arbitrarily chosen item of \(\mathcal{P}_i\) and set \(f\) as \(\epsilon_i\) for the rest of \(\mathcal{P}_i\). Note that~$f(\mathcal{P}_i) = |U|^{t+2}$ for each~$i \in [t]$, so that~$\frac{1}{t} \sum_{x \in U}f(x) = |U|^{t+2}$. We now prove that this choice of~$f$ has the claimed properties.
        
        Let \(\mathcal{P}'\) be a partition of \(U\) into \(t\) sets \(\mathcal{P}'_1,\ldots,\mathcal{P}'_t\) such that \(\sum_{x\in \mathcal{P}'_i}f(x)\leq \frac{1}{t} \sum_{x\in U}f(x)\) for each \(i \in [t]\), which can only happen if~$f(\mathcal{P}'_i) = |U|^{t+2}$ for each~$i \in [t]$. If two items with \(f\)-values \(\Omega_i\) and \(\Omega_j\) are in the same part \(\mathcal{P}'_z\) of \(\mathcal{P}'\), then 
        \begin{align*}
            \sum\limits_{x\in \mathcal{P}'_z}f(x) \geq \Omega_i + \Omega_j &= |U|^{t+2}-(|\mathcal{P}_i|-1)\cdot \epsilon_i+|U|^{t+2}-(|\mathcal{P}_j|-1)\cdot \epsilon_j\\
            &\geq 2|U|^{t+2}-|U|^{t}-|U|^{t-1} > |U|^{t+2}.
        \end{align*}
        Thus, no part contains two such items, and w.l.o.g., we can assume that the item with \(f\)-value \(\Omega_i\) is in \(\mathcal{P}'_i\). Now, we prove the rest of the items in \(\mathcal{P}'_i\) each have \(f\)-value \(\epsilon_i\) by induction on \(i\). 

        For~$i=1$, we know that~$\mathcal{P}'_1$ contains the item of value~$\Omega_1$ and satisfies~$f(\mathcal{P}'_1) = |U|^{t+2}$. Hence the other items in~$\mathcal{P}'_1$ have total value~$|U|^{t+2} - \Omega_1 = (|\mathcal{P}_1| - 1) \cdot \epsilon_1 < |U| \cdot |U|^{0} = |U|$. Since~$\epsilon_i > |U|$ for all~$i > 1$, the only other items in~$\mathcal{P}_1$ have value less than~$\epsilon_2$, and are therefore the~$|\mathcal{P}_1|-1$ items of value~$\epsilon_1$. Hence~$\mathcal{P}'_1 = \mathcal{P}_1$.

        Now, assume that each part \(\mathcal{P}'_{i'}\in\{\mathcal{P}'_1,\ldots,\mathcal{P}'_{i-1}\}\) consists of one item with \(f\)-value \(\Omega_{i'}\) and \(|\mathcal{P}'_i|-1\) items with \(f\)-value \(\epsilon_{i'}\). We show that \(\mathcal{P}'_i\) consists of an item with \(f\)-value \(\Omega_i\) and \(|\mathcal{P}_i|-1\) items with \(f\)-value \(\epsilon_i\), so that~$\mathcal{P}'_i = \mathcal{P}_i$. We proved that an item in \(\mathcal{P}'_i\) has \(f\)-value \(\Omega_i\). The total \(f\)-value of the rest of the items in \(\mathcal{P}'_i\) must be at most \((|\mathcal{P}_i|-1)\cdot \epsilon_i\), which is strictly less than~$\epsilon_j$ for~$j > i$. Hence the rest of the items in \(\mathcal{P}_i\) must have \(f\)-value exactly equal to \(\epsilon_i\). Therefore, each part \(\mathcal{P}'_i\) must consist of an item with \(f\)-value \(\Omega_i\) and \(|\mathcal{P}_i|-1\) items with \(f\)-value \(\epsilon_i\), and so, \(\mathcal{P}=\mathcal{P}'\). Hence, this partition is unique and \(f\) is the desired function of this lemma. 
    \end{proof}

%% file: Algorithms_Section.tex
\label{sec:Algorithms}
\subsection{Algorithms for \BLVshort and \ULVshort}
In this section, we present our algorithmic results for \BLVshort and \ULVshort.
These are based on first analyzing the structure of solutions and then exploiting this structure to be able to compute their cost using \Cref{thm:cheapestshortpath} and \Cref{thm:vbcsp}.
In particular, we show that a minimum-cost solution $\{H_1, \ldots, H_t\}$ can be decomposed into pairwise disjoint paths between the vertices of degree three and higher in $H := \bigcup _{i=1}^t H_i$, with the guarantee that the digging cost of each path is minimum among all paths whose sum of $\ell$-values is the same. (If a path does not have minimum cost, then it can be replaced by a cheaper one to obtain a cheaper overall solution.)
Such paths can be computed with \Cref{thm:cheapestshortpath} as soon as we know their endpoints and length-thresholds.
As these are interdependent among each other, we formulate this problem as a valued binary CSP with a variable for each endpoint of such a path.
Hence, the structure of this CSP can be captured by a graph with a vertex for each endpoint of such a path and an edge for each such path.
This is the graph we investigate in \Cref{subsec:SolutionStructure}; ultimately arriving at a bound on the number of its vertices.
This will actually allow us to solve a CSP with the rough structure described above using \Cref{thm:vbcsp}.
We carry out the details of setting up the CSP appropriately and show that this indeed yields algorithms for \BLVshort and \ULVshort with the claimed running times in \Cref{subsec:XPAlgorithmBLV}.

\subsubsection{Solution structure}
\label{subsec:SolutionStructure}
We will now analyze the structure of a solution viewed as a set of minimum-cost bounded-length paths between certain endpoints.
In particular, we want to bound the number of such endpoints that we need to consider.
For this, we introduce the \emph{abstract structure} of a solution whose vertices correspond to such endpoints and whose edges correspond to minimum-cost bounded-length paths between them.

\begin{definition} \label{def:abstract:structure}
    For a solution \(\mathcal{S}=\{H_1,\ldots,H_t\}\) of \BLVshort or \HVshort, the \emph{abstract structure} \(\CAbs{\mathcal{S}}\) is a graph \(H\) with a color from \(\{c_S \mid \emptyset \neq S \subseteq 2^{[t]}\}\) associated to each of its edges, constructed as follows.
\begin{itemize}
    \item Initially, let \(H=\bigcup_{i=1}^t H_i\) and color each \(e \in E(H)\) by \(c_S\), where \(S = \{i \in [t] \mid e \in E(H_i)\}\).
    \item While a vertex \(v \neq r\) has degree \(2\) in \(H\) and its incident edges have the same color \(c\),  we remove \(v\) and connect the neighbors of \(v\) with an edge of color \(c\).
\end{itemize}
 The subgraph of \(\CAbs{\mathcal{S}}\) corresponding to \(H_i\), denoted by \(\PartialCAbs{i}{\mathcal{S}}\), is the subgraph of \(\CAbs{\mathcal{S}}\) consisting of all edges with a color \(c_S\) for which \(i \in S\).
\end{definition}
Considering the color of the edges of \(\CAbs{\{H_1,\ldots,H_t\}}\), we can find \(\PartialCAbs{i}{\mathcal{S}}\) in polynomial time for each~$i \in [t]$. Clearly, \(\PartialCAbs{i}{\mathcal{S}}\) is obtained from \(H_i\) by dissolving some degree-\(2\) vertices.
 
 The following lemma guarantees the existence of a solution for which the number of vertices in its abstract structure is bounded quadratically in \(|T|\). This bound enables us to enumerate all possible such structures and, accordingly, to search for a solution.

\begin{lemma}  
\label{lemma:NumberOfVerticesAbstractStructure}
    For every yes-instance \((G,\ell,\dig,r,T,d,\alpha,\beta,\lambda,t)\) to \BLVshort and \ULVshort, there exists a solution \(\mathcal{S}=\{H_1,\ldots,H_t\}\) such that \(|V(\CAbs{\mathcal{S}})|\leq 1 + 2|T|+8|T|^2\).
\end{lemma}
\begin{proof}
    To prove this lemma, we define an ordering on solutions to the instance, and then prove that a maximal element in this ordering has an abstract structure whose number of vertices is suitably small.
    For an arbitrary solution \(\mathcal{S}=\{H_1,\ldots,H_t\}\), for each edge \(e \in E(G)\), we set \(x_e(\mathcal{S})\) to be the number of trees in \(\mathcal{S}\) that contain \(e\). Then, we define \emph{the coloring sequence} of \(\mathcal{S}\) as a tuple \(\mathcal{T}(\mathcal{S}) := (y_1, \dotsc, y_{|E(G)|})\)  where \(y_1,\ldots,y_{|E(G)|}\) is the decreasingly sorted sequence of values \(x_e(\mathcal{S})\) for \(e \in E(G)\).
    Let \(\mathcal{S}=\{H_1,\ldots,H_t\}\) be a solution where the value of \(\text{cost}(\mathcal{S})\) is minimized and which, among all solutions that minimize \(\text{cost}(\mathcal{S})\), has a maximum tuple \(\mathcal{T}(\mathcal{S})\) with respect to the lexicographic ordering. 
    
    To bound the number of vertices in the abstract structure~$\CAbs{\mathcal{S}}$, we exploit the conditions under which \cref{def:abstract:structure} removes a degree-2 vertex. These conditions can be expressed using the following concept. For a tree~$H_i$, we define its \emph{special vertices} to be the root~$r$ of the instance together with the vertices of~$H_i$ whose degree in~$H_i$ is unequal to~2. We define a \emph{plain path} of~$H_i$ as a path~$P$ in~$H_i$ whose endpoints are special but whose internal vertices are not. Since degree-2 non-root vertices of~$\bigcup_i H_i$ are dissolved when both incident edges have the same color, while the edge colors encode which trees contain the edge, it follows that the vertices of~$\bigcup_i H_i$ that remain in~$\CAbs{\mathcal{S}}$ are (1)~the root vertex~$r$, (2)~each non-root vertex that has degree unequal to~$2$ in at least one tree~$H_i$, and (3)~vertices~$v$ that lie on plain paths~$P_i, P_j$ of two distinct trees~$H_i, H_j$, such that~$v$ has degree more than~$2$ in the union~$H_i \cup H_j$, i.e., places where plain paths of distinct trees converge or diverge (see \Cref{fig:AbstractStructureTypeOfVertices}). The root vertex \(r\) is the only vertex that is of type \(1\). To bound the number of vertices in the abstract structure, we bound the number of vertices of each of the two other types.
    
     First, we bound the the number of vertices of type \(2\). Let \(A_i\) be the set of non-root vertices of \(H_i\) whose degree is not \(2\) and \(f_i\) be the number of non-root leaves of \(H_i\). By cost-minimality, we can assume that the non-root leaves of each tree are terminals. Additionally, since each terminal has degree-\(1\) in the input graph $G$ by definition of \BLVshort, it also follows from cost-minimality that each terminal appears in exactly one tree. So, we have \(\sum_{i=1}^t f_i=|T|\). Recalling that the number of vertices of degree unequal to $2$ in a tree is at most twice the number of its leaves minus $1$, we can provide an upper bound for \(|A_i|\) as follows. If root \(r\) is a leaf of \(H_i\), then the number of leaves is \(f_i+1\) and so,  \(|A_i\cup \{v\}| \leq 2(f_i+1)-1=2f_i+1\). Otherwise, if the root \(r\) is not a leaf, we have \(|A_i| \leq 2f_i-1\leq 2f_i\). Since in both cases we have \(|A_i|\leq 2f_i\), we can conclude that the total number of vertices of type \(2\) is at most \(\sum\limits_{i=1}^{t}|A_i| \leq \sum\limits_{i=1}^t 2f_i\leq 2|T|\).

    Now, it remains to analyze vertices of type~(3). The following claim will be instrumental to do so.

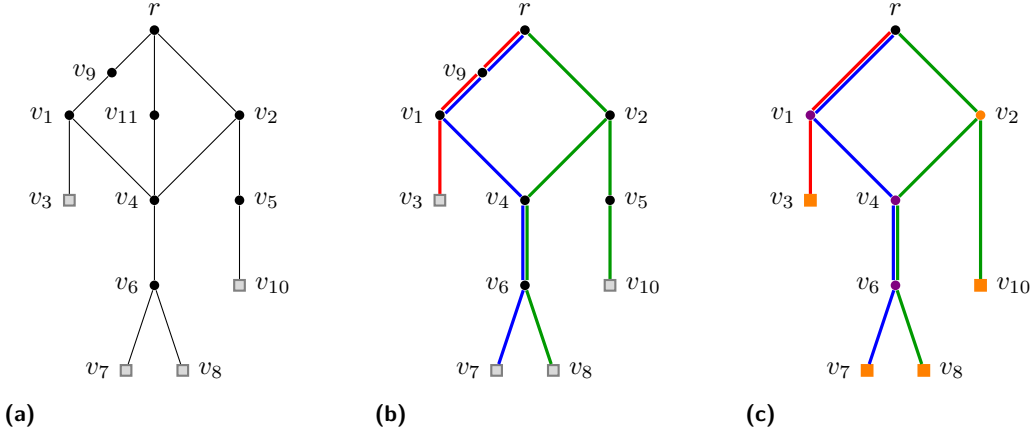
\begin{figure}[!t]
        \centering
        \begin{subfigure}[b]{0.3\textwidth}
            \centering
            \input{figs/Abstract_Structure_instance}
            \subcaption{}
            \label{fig:Abs_Instance}
        \end{subfigure}
        \hfill
        \begin{subfigure}[b]{0.3\textwidth}
            \centering
            \input{figs/Abstract_Structure_trees}
            \subcaption{}
            \label{fig:Abs_T1}
        \end{subfigure}
        \hfill
        \begin{subfigure}[b]{0.3\textwidth}
            \centering
            \input{figs/Abstact_Structure_Type_Of_Vertices}
            \subcaption{}
            \label{fig:Abs_Type_Vertices}
        \end{subfigure}
        \caption{(a) An instance of \BLVshort/\ULVshort. Each edge has length \(1\) and digging cost \(100\). Squares represent terminals. (b) A solution \(\mathcal{S}\). (c) \(\CAbs{\mathcal{S}}\), with vertices of the same type in the same color. 
         Type 2 vertices are shown in orange and type 3 vertices are shown in purple. The vertex \(v_{11}\) is not in any solution trees and so not in \(\CAbs{\mathcal{S}}\). The vertices \(v_5\) and \(v_9\) are in the solution trees but not in the abstract structure as they have degree \(2\) in the union of the trees. One can see that the vertices \(v_4\) and \(v_6\) are shared between the path from \(r\) to \(v_7\) in (c) and the path from \(v_2\) to \(v_8\) in (d) causing a degree more than \(2\) for them in the abstract structure. The path from \(r\) to \(v_3\) in (b) and the path from \(r\) to \(v_7\) in (c) also intersect in \(v_1\) which makes \(v_1\) to be of degree more than \(2\) in the abstract structure. So, all of these vertices remain in the abstract structure though they have degree \(2\) in the mentioned trees.}
        \label{fig:AbstractStructureTypeOfVertices}
    \end{figure}

    \begin{claim} \label{claim:union:2paths}
        For plain paths~$P_i, P_j$ of two distinct trees~$H_i,H_j$, the graph~$P_i \cup P_j$ has at most two vertices of degree more than~$2$.
    \end{claim}
    \begin{claimproof}
        Since both~$P_i$ and~$P_j$ are paths, a vertex of~$P_i \cup P_j$ of degree more than~$2$ is a place where the two paths either converge or diverge. Hence if there are three or more such vertices, the graph~$P_i \cup P_j$ has a cycle. A minimal such cycle can be decomposed into two parts where one part is a subpath of \(P_i\) and the other is a subpath of \(P_j\). We call these sub-paths \(P'_i\) and \(P'_j\), respectively. Using \(P'_i\) and \(P'_j\), we now derive a contradiction by building a minimum-cost solution \(S'\) whose coloring sequence is larger than that of \(\mathcal{S}\).

        Since all internal vertices of~$P'_i$ and~$P'_j$ are plain, they have degree~$2$ and cannot be a terminal or the root. Hence we can replace~$P'_i$ in~$H_i$ by~$P'_j$ to obtain a tree~$H'_i$ containing the root that covers the same terminals as~$H_i$. Similarly, we can replace~$P'_j$ in~$H_j$ by~$P'_i$ to obtain a tree~$H'_j$. If~$\ell(P'_i) \neq \ell(P'_j)$, then replacing the longer by the shorter segment yields a solution whose cost is strictly smaller: the digging costs do not increase (as all edges involved in the replacement were dug on account of the other tree anyway), the total depth does not increase, and the total length of the edges in the replaced tree strictly decreases. Since we started from a cost-minimal solution, we must therefore have~$\ell(P'_i) = \ell(P'_j)$. We now derive a contradiction by defining a minimum-cost solution that is lexicographically larger.  
         
        Let \(x'_i\) be the maximum value of \(x_e(\mathcal{S})\) over all \(e\in E(P'_i)\) and \(x'_j\) be the maximum value of \(x_e(\mathcal{S})\) over all \(e \in E(P'_j)\). Without loss of generality, assume that \(x'_i \geq x'_j\). Let~$\mathcal{S}'$ be the solution obtained from~$\mathcal{S}$ by replacing~$H_j$ with~$H'_j$. We argue the coloring sequence of~$\mathcal{S}'$ is larger than that of~$\mathcal{S}$ in the lexicographical ordering.  The edge of~$P'_i \cup P'_j$ that was used most in solution~$\mathcal{S}$ (which appears on the most relevant position in the sorted sequence) appears at least one additional time in solution~$\mathcal{S}'$ on account of replacing~$H_j$ (which does not contain that edge) by~$H'_j$ (which does). This contradiction proves the claim.
    \end{claimproof}

     \cref{claim:union:2paths} leads to an upper bound of~$8|T|^2$ on the number of vertices of~$\CAbs{\mathcal{S}}$ of type~(3) by bounding the total number of plain paths in the solution, as follows. Note that the plain paths of a tree~$H_i$ are in 1-to-1 correspondence with the edges of the tree~$H'_i$ that results by dissolving all non-special vertices of~$H_i$. Since any tree has strictly fewer edges than vertices, and noting that each tree has at most \(|A_i|+1\) special vertices, \(r\) together with all non-degree~\(2\) vertices, there are \(|A_i|\) edges in \(H'_i\). Thus, the total number of edges over all trees~$H'_1, \ldots, H'_t$  (and therefore the total number of plain paths) is at most \(\sum\limits_{i\in[t]}|A_i|\leq 2|T|\). \cref{claim:union:2paths} shows that each pair of plain paths adds at most \(2\) vertices of type~(3) to~$\CAbs{\mathcal{S}}$. Hence, the total number of vertices in~$\CAbs{\mathcal{S}}$ is~$1 + 2|T| + 2\cdot (2|T|)^2 = 1 + 2|T|+8|T|^2$.
\end{proof}

\subsubsection{\XP-algorithms for Power Network Design}
\label{subsec:XPAlgorithmBLV}
Using bounds on the abstract structure we present \XP-algorithms for power network design.

\begin{theorem} 
\label{thm:XP-algo-tw}
The \BLV problem can be solved in time \(f(t,|T|)\cdot(|V(G)|\cdot\lambda^t)^{\mathcal{O}(|T|^2)}\) for any arbitrary input instance \((G,\ell,\dig,r,T,d,\alpha,\beta,\lambda,t)\).
\end{theorem}

\begin{proof}
    First, we provide an algorithm that solves this problem. We then discuss its correctness and running time analysis. Given an input instance \((G,\ell,\dig,r,T,d,\alpha,\beta,\lambda,t)\), we proceed in the following steps.
\begin{enumerate}
    \item For each graph \(H\) with at most \(1+2|T|+8|T|^2\) vertices whose edges are colored with \(2^t-1\) colors, each color corresponding to a non-empty subset of \([t]\), we fix an arbitrary orientation of the edges of~$H$ and do the following. (We consider~$H$ as a guess for the abstract structure of a solution, up to isomorphism.)
    \label{algostep:absstructure}
    \begin{enumerate}
        \item Let \(H'_i\) be the subgraph of \(H\) induced by all edges with a color \(c_S\) where \(i \in S\). Let \(V_i\) be the set of vertices of \(H'_i\). Check if each subgraph \(H'_i\) for \(i \in [t]\) is a tree; if not, continue with the next guess for~$H$.
        \label{algostep:abscheck}
        \item For each injective assignment \(\varphi \colon T \cup \{r\} \to V(H)\) that satisfies the following conditions: 
        \begin{itemize}
            \item the demand constraint holds for all trees, i.e., the total demand of the preimages of \(\varphi\) of the terminals in each tree \(H_i\) is at most \(\alpha\), and
            \item vertex \(\varphi(r)\) is contained in each tree~$H'_i$,
        \end{itemize}
         we do the following:
        \label{algostep:VertexAssignment}
        \begin{enumerate}
        \item Root tree~$H'_i$ at~$\varphi(r)$ and direct its edges away from the root.\label{algostep:absdirection}
        \item Define a valued binary CSP-instance \((\mathcal{V},\mathcal{D},\mathcal{C})\) as follows.
        \label{algostep:CSP}
        \begin{itemize}
            \item \(\mathcal{V} = V(H)\), so each vertex of~$H$ has an associated variable.
            \item \(\mathcal{D} = \{(v,x_{1},\ldots,x_{t}) \mid v \in V(G) \land x_1, \dotsc, x_t \in \{0\} \cup [\lambda]\}\).
        \label{algostep:CSPVariables}
            \item \(\mathcal{C} = \mathcal{C}' \cup \mathcal{C}''\), so the set of constraints consists of two parts that are defined below. (The first part,~$\mathcal{C}'$, effectively acts on individual variables and forces those vertices of~$H$ that correspond to a terminal or the root under~$\varphi$, to be mapped to the corresponding terminal or root. The second part,~$\mathcal{C}''$, tracks the cost of realizing an edge of the abstract structure by a bounded-depth path in~$G$.)
            \item \(\mathcal{C'} = \{(u,u,c_{u})\mid u \in \mathcal{V}\}\) where \[c_u((v,x_{1},\ldots,x_{t}),(v',y_{1},\ldots,y_{t})) = \begin{cases}
                \beta + 1 & \mbox{ if } v \neq v' \lor \exists i \in [t] \ x_i \neq y_i \\
                &\lor (v \in T \cup \{r\} \land u \neq \varphi(v))\\
                &\lor (u \in \varphi(T \cup \{r\}) \land \varphi^{-1}(u) \neq v)\\
                0 & \mbox{otherwise.}
            \end{cases}\]
            Note that a feasible solution to the problem has cost at most~$\beta$, so that the value~$\beta+1$ in the first case effectively forbids such behavior from happening.
        \label{algostep:domainconstraints}
            \item \(\mathcal{C}'' = \{(u,v,c_{uv}) \mid uv \in E(H)\}\)
            where \(c_{uv}((u',x_1,\ldots,x_t),(v',y_1,\ldots,y_t))\) is defined as follows for each \(((u',x_1,\ldots,x_t),(v',y_1,\ldots,y_t)) \in \mathcal{D} \times \mathcal{D}\). Let~$S \subseteq [t]$ be the set encoded by the color of edge~$uv$. For all \(i\in [t]\) for which \(uv \in E(H'_i)\), let \(z_i\) be \(1\) if the direction of edge \(uv\) is from \(u\) to \(v\), and \(z_i\) be \(-1\) otherwise. Let \(\lambda'\) be the minimum value of \(z_i\cdot (y_i-x_i)\) taken over all \(i\in [t]\) for which \(uv \in E(H'_i)\).
            Set \(c_{uv}((u',x_1,\ldots,x_t),(v',y_1,\ldots,y_t))\) to the minimum value of \(|S|\cdot \ell(P')+\dig(P')\) over all paths \(P'\) between \(u'\) and \(v'\) in \(G\) with \(\ell(P') \leq \lambda' \). This can be computed using \Cref{thm:cheapestshortpath} considering the input graph as a copy \(G'\) of \(G\) where all edges are duplicated in two different directions, defining \(f(e')=\ell(e)\) and \(g(e')=|S|\cdot \ell(e)+ \dig(e)\) for all \(e' \in E(G')\) where \(e'\) is a copy of \(e \in E(G)\), and clearly considering \(s=u'\), \(t=v'\) and \(\lambda'\) as the threshold. If no such path exists, set the cost to \(\beta + 1\). 
            \label{algostep:CSPConstraints}
            \end{itemize}
            \item Solve CSP-instance \((\mathcal{V},\mathcal{D},\mathcal{C})\) using \Cref{thm:vbcsp}.\label{algostep:CSPsolve}
            \item Check if the computed cost is at most \(\beta\) and return `yes' if so.
            \label{algostep:yesSolution}
        \end{enumerate}
    \end{enumerate}
    \item If the algorithm reaches this step, then return `no'.
    \label{algostep:noSolution}
\end{enumerate}

\subparagraph*{Correctness} We show that the given algorithm outputs `yes' if and only if its input \((G,\ell,\dig,r,T,d,\alpha,\beta,\lambda,t)\) is a `yes'-instance to \BLVshort. To prove that the algorithm outputs `yes' if there exists some solution to the given input instance, we show that the properties of at least one solution lie in all of our branching steps leading up to solving the CSP-instance, and so, we show that at most the cost of that solution should be computed during the CSP subroutine. On the other hand, for the other direction, we show that if the algorithm returns `yes', we can construct a solution from the implied solution returned by the CSP subroutine. The validity of the properties required to be a solution for the \BLVshort-instance solution can be proved based on the properties considered in the branching steps. We prove these directions, respectively, through the following claims.

\begin{claim}
    If there exists a solution for the \BLVshort-instance \((G,\ell,\dig,r,T,d,\alpha,\beta,\lambda,t)\), the algorithm returns `yes'.
\end{claim}
\begin{proof}
    We discuss the algorithm step by step to show that, considering a proper solution, each step of the algorithm is executed successfully, resulting in a return of `yes'.

    By \Cref{lemma:NumberOfVerticesAbstractStructure}, there exists at least one solution \(\mathcal{S} = \{H_1, \dotsc, H_t\}\) such that \(\CAbs{\mathcal{S}}\) has at most \(1 + 2|T|+8|T|^2\) vertices. This means that in at least one of the cases of step~\ref{algostep:absstructure}, we are considering \(H\simeq\CAbs{\mathcal{S}}\) with this isomorphism witnessed by \(\iota \colon V(H) \to V(\CAbs{\mathcal{S}})\). As the algorithm tries all ways to assign colors from the range~$[2^t-1]$ to the edges of~$H$, in some iteration it achieves the same coloring as the abstract structure~$\CAbs{\mathcal{S}}$. This means that each~$H'_i$ as defined in the algorithm is equal to \(\PartialCAbs{i}{\mathcal{S}}\), which is obtained by dissolving degree-2 vertices of the tree~$H_i$ and therefore a tree. Therefore, the algorithm proceeds to step~\ref{algostep:abscheck}. 
    
    Next we argue that, in some branch of step~\ref{algostep:VertexAssignment}, we consider~$\varphi$ with \(\varphi(v) = \iota^{-1}(v)\) for each \(v \in T \cup \{r\}\).
    Without loss of generality, we can assume that \(\mathcal{S}\) is a solution that has the minimum cost among all solutions for which \(\CAbs{\mathcal{S}}\) has at most \(1+2|T|+8|T|^2\) vertices. This minimality assumption guarantees that all terminals are leaves of trees, each terminal is covered by exactly one tree, and each leaf is either a terminal or the root. Thus, the assignment of terminals implied by \(\varphi\) satisfies the demand constraint, and \(r\) is assigned to a vertex contained in each \(H'_i\).
    Hence, indeed \(\varphi = \iota^{-1}\) is considered in step~\ref{algostep:VertexAssignment}.
    
    Then, as we discussed, we know that \(H'_i\) in \(H\) is isomorphic to \(\PartialCAbs{i}{\mathcal{S}}\). Hence directions are defined in the same way on both subgraphs, noting that both \(\varphi(r)\) and \(r\) take the role of the root, respectively. We define a solution for CSP as follows. For each variable in~$\mathcal{V}$, corresponding to some vertex \(v\) in \(H\), we set the vertex value of the corresponding tuple to \(\iota(v)\) and set \(x_i\) to be the depth of this vertex in \(H_i\) if~$H_i$ contains it; if not, we set~$x_i$ to~$0$. Then, the cost of each constraint in~$\mathcal{C}'$, of the form \((u,u,c_u)\) for \(u\in \mathcal{V}\), is~\(0\) since the vertices are assigned correctly. Moreover, each constraint in~$\mathcal{C}''$, of the form \((u,v,c_{uv})\), has a cost \(c_{uv}=|S|\cdot \ell(e)+\dig(e)\) by definition. This is precisely the cost of the path in \(\mathcal{S}\) corresponding to \(uv\) in \(H\): since the path appears in exactly the trees of \(\mathcal{S}\) indexed by \(S\), its length is counted \(|S|\) times, while its digging cost, being shared among all of them, is counted once. Thus, the cost of this CSP solution is equal to the cost of \(\mathcal{S}\). Since we know that the CSP has a solution of cost at most~$\beta$, \cref{thm:vbcsp} returns a solution of cost at most \(\beta\), and thus the algorithm returns `yes'.
\end{proof}

\begin{claim}
    If the algorithm returns `yes' for the input instance \((G,\ell,\dig,r,T,d,\alpha,\beta,\lambda,t)\), \BLVshort has a solution for this input instance.
\end{claim}
\begin{proof}
    Assume that \(H\) is the guess for \(\CAbs{\mathcal{S}}\) that ends in a `yes' output. Since the algorithm returns `yes', in some iteration the CSP-subroutine in step~\ref{algostep:CSP} finds an assignment of the variables (which by construction correspond to vertices of \(H\)) whose cost is at most \(\beta\).
    All further references to steps in this claim's proof refer to the iteration in which step~\ref{algostep:CSP} returns `yes'.
    
    We construct a subgraph \(H''_i\) of \(G\) for all \(i \in [t]\) as follows. In \(H'_i\) from step~\ref{algostep:abscheck} replace the vertices with the vertices of \(G\) in the first entry of their corresponding variable's CSP-assignment, and replace each edge of \(H'_i\) with the path in \(G\) between the vertices assigned to its endpoints in this way, which can be found by the algorithm of \Cref{thm:cheapestshortpath}.
    By the definition of the cost of an assignment for the CSP and the fact that no constraint realizes a value of \(\beta + 1\), \(\text{cost}(\{H''_1, \dotsc, H''_t\}) \leq \beta\).
    Since we do not have control over finding pairwise disjoint paths for each pair, there is a possibility that \(H''_i\) is not a tree for some \(i \in [t]\).
    We will resolve this later and first show that apart from this \(\{H''_1, \dotsc, H''_t\}\) satisfies the covering, demand, and depth constraints of being a \BLVshort-solution.
    
    Because constraints \(\mathcal{C}'\) are included, the assignment ensures that \(T \cup \{r\}\) appears in \(V(H)\) according to \(\varphi\) from step~\ref{algostep:abscheck}.
    Hence, by the conditions on \(\varphi\) from step~\ref{algostep:abscheck}, the covering and demand constraints hold for all \(H''_i\) where \(i \in [t]\).

    For the depth constraint, consider an arbitrary \(i \in [t]\) and \(z \in V(H''_i) \cap T\).
    By the construction of \(\mathcal{C}''\) in step~\ref{algostep:absdirection}, and since the cost of the CSP-assignment is at most \(\beta\), there is some value \(x^v_i \in \{0\} \cup [\lambda]\) for each vertex \(v\) on the \(\varphi(r)\)-\(\varphi(z)\) path \(P\) in \(H_i\) such that the path we include into \(H''_i\) for the edge between two consecutive vertices \(u,v\) on \(P\) has length at most \(x^v_i - x^u_i\).
    This means that overall \(\ell(P) \leq \sum_{uv \in E(P)} x^v_i - x^u_i\) which is a telescoping sum equal to \(x^{\varphi(z)}_i - x^{\varphi(u)}_i\), which is at most~$\lambda - 0\) since the domain~$\mathcal{D}$ only allows integers from~$\{0\} \cup [\lambda]$ on the integer-valued coordinates.
    So apart from being trees, the subgraphs \(H''_1, \dotsc, H''_t\) satisfy all the conditions of being a solution for \BLVshort.

    For each \(i \in [t]\), we get rid of cycles in \(H''_i\) and construct \(\hat{H}_i \subseteq H''_i\), which is a tree in the following way. Initialize \(\hat{H}_i\) as a copy of \(H''_i\). As long as \(\hat{H}_i\) contains some cycle \(C\) proceed as follows. Let \(v\) be a vertex in \(C\) that has the maximum \(\ell\)-distance to \(r\). Let \(u\) and \(u'\) be neighbors of \(v\) in \(C\). If the removal of the edge \(vu\) increases the distance \(\ell\) of a vertex from \(r\), the edge \(vu\) must be on the unique shortest path from \(r\) to \(v\). In this case, since the shortest path \(p\) from \(r\) to \(v\) must be unique, the removal of \(u'v\) does not change the distance of any vertices from \(r\); otherwise, it belongs to a shortest path from \(r\) to \(v\) that is not \(p\), contradicting that \(p\) is unique. Therefore, we can remove one of the edges of \(C\) that is incident on \(v\) without changing the distance of the vertices from \(r\).
    
    This procedure maintains in \(\hat{H}_i\) as in \(H''_i\) the same connectivity and hence connects the same terminals, which is considered over all \(i \in [t]\) ensures the covering and demand constraints.
    Moreover, by the arguments above, the depth constraint is also maintained.
    Finally, since each \(\hat{H}_i\) is a subgraph of \(H''_i\), the cost constraint is also trivially maintained.
    \end{proof}

Putting these two directions together, we can conclude that \BLVshort has a solution for the input instance \((G,\ell,\dig,r,T,d,\alpha,\beta,\lambda,t)\) if and only if the algorithm returns `yes'.

\subparagraph*{Running Time Analysis} Step~\ref{algostep:absstructure} goes through all graphs with at most \(1+2|T|+8|T|^2\) vertices with edges of \(2^t-1\) colors. Considering all graphs with at most \(1+2|T|+8|T|^2\) vertices, \((1+2|T|+8|T|^2)^2\) edges, and \(2^t-1\) colors on its edges, the number of cases is at most \(2^{\mathcal{O}(t\cdot|T|^4)}\). Then, steps \ref{algostep:abscheck} and \ref{algostep:absdirection} can be processed in linear time using a simple graph traversal algorithm such as BFS on the edges belonging to each tree, and so each takes \(\mathcal{O}(|T|^2)\) and in total \(\mathcal{O}(t \cdot |T|^2)\). The number of ways to assign terminals and the root vertex to the vertices of \(H\) is at most \(\mathcal{O}(|T|^{|T|^2})\), which indicates an upper bound on the cases in step~\ref{algostep:VertexAssignment}. The demand constraint can be checked trivially in linear time \(\mathcal{O}(|T|)\). Finally, we discuss the time needed to solve the defined CSP-instance in step~\ref{algostep:CSP}. By \Cref{thm:cheapestshortpath}, we can compute the cost function for two values of variables in a constraint in time \(\mathcal{O}(\lambda \cdot |V(G)|\cdot |E(G)|)\) since \(\lambda' \leq \lambda\) and the number of vertices and edges of \(G\) and \(G'\) are asymptotically equal. To solve the described CSP, we use \Cref{thm:vbcsp}. So, the CSP-instance can be solved in time \(|\mathcal{V}|\cdot|\mathcal{D}|^{\mathcal{O}(\operatorname{tw}(G_P(\mathcal{V}, \mathcal{D}, \mathcal{C}))}\). Since there are at most \(8|T|^2+2|T|+1\) variables and so, at most \(8|T|^2+2|T|+1\) vertices in the primal graph, the treewidth is at most \(\mathcal{O}(|T|^2)\), this results in a running time of \((8|T|^2+2|T|+1)(|V(G)|\cdot \lambda^t)^{\mathcal{O}(|T|^2)}\) for solving the CPS-instance. Each of the remaining steps can be followed in constant time. Putting all of this together results in a running time of \(f(t,|T|)\cdot(|V(G)|\cdot\lambda^t)^{\mathcal{O}(|T|^2)}\).
\end{proof}

The main bottleneck of the algorithm presented in the proof of \Cref{thm:XP-algo-tw} lies in solving the CSP-instance. A straightforward alternative is to replace the call to the CSP algorithm from \Cref{thm:vbcsp} with an exhaustive branching approach: iterating over all possible assignments of domain values \(\mathcal{D}\) to variables \(\mathcal{V}\) and checking in polynomial time which of these satisfy all constraints to find a valid solution for the given input instance. In this approach, the running time remains \(f(t,|T|)\cdot(|V(G)|\cdot \lambda^t)^{\mathcal{O}(|T|^2)}\), but the algorithm requires only polynomial space and avoids the usage of the powerful CSP algorithm.
However, it is important to note that this asymptotic running-time equivalence arises from the analysis, not from the actual behavior of the CSP-based algorithm. In the CSP approach, the running time depends on the treewidth of the primal graph of the CSP-instance, which is conservatively replaced with a crude upper bound in the analysis. The following theorem refines this by considering the case where the input graph is planar, allowing for a tighter bound on the treewidth. In \Cref{sec:W1HardBLV}, we further show that in this planar setting, the algorithm achieves the best possible running time under standard complexity assumptions.

\begin{theorem}
\label{thm:XP-algo-planar}
\BLV can be solved in time \(f(t,|T|)\cdot(|V(G)|\cdot\lambda^t)^{\mathcal{O}(|T|)}\) for any input instance \((G,\ell,\dig,r,T,d,\alpha,\beta,\lambda,t)\) where \(G\) is planar.
\end{theorem}
\begin{proof}
    The algorithm is a small variation of that of \Cref{thm:XP-algo-tw}. Since the abstract structure~$\CAbs{\mathcal{S}}$ of a solution~$\mathcal{S} = \{H_1, \ldots, H_t\}$ is obtained by dissolving degree-2 vertices of a subgraph of~$G$, for a planar input graph~$G$ it follows that~$\CAbs{\mathcal{S}}$ is also planar. Hence  \Cref{lemma:NumberOfVerticesAbstractStructure} implies that there exists a solution~\(\mathcal{S}\) where \(\CAbs{\mathcal{S}}\) is a planar graph on \(\mathcal{O}(|T|^2)\) vertices. As an $n$-vertex planar graph has treewidth~$\Oh(\sqrt{n})$ (see, e.g.,~\cite[Corollary 7.24]{ParAlg}), the treewidth of such \(\CAbs{\mathcal{S}}\) is \(\mathcal{O}(\sqrt{|V(\CAbs{\mathcal{S}})|})=\mathcal{O}(|T|)\). 
    
    This means that in step~\ref{algostep:absstructure} of the algorithm it suffices to look only through branches where \(H\) is planar (which can be tested in linear time~\cite{HopcroftT74}) and therefore has treewidth \(\Oh(|T|)\). Furthermore, in these branches, since the treewidth of the primal graph is equal to the treewidth of \(H\), we can present a tighter upper bound on the running time of step~\ref{algostep:CSP}, solving CSP. The CSP subroutine invoked in this step takes time \(|\mathcal{D}|^{\mathcal{O}\operatorname{tw}(G_P(\mathcal{V}, \mathcal{D}, \mathcal{C}))}\). Thus, each call to the subroutine can be handled in time \((|V(G)|\cdot \lambda^t)^{\mathcal{O}(|T|)}\). The remaining overhead is the same as in \Cref{thm:XP-algo-tw}. Hence the overall running time for planar instances is \(f(t,|T|)\cdot(|V(G)|\cdot\lambda^t)^{\mathcal{O}(|T|)}\) for some computable function~$f$.
\end{proof}

Finally, we address the version of the problem without the depth constraint. The next theorem presents an algorithm for solving this variant.

\begin{theorem} \label{thm:XP:unbounded}
The \ULV problem can be solved in time \(f(t,|T|)\cdot|V(G)|^{\mathcal{O}(|T|^2)}\) for any arbitrary input instance \((G,\ell,\dig,r,T,d,\alpha,\beta,t)\).
\end{theorem}
\begin{proof}
    Consider the algorithm presented in the proof of \Cref{thm:XP-algo-tw}, where step~\ref{algostep:absdirection} is removed and step~\ref{algostep:CSP} is replaced by the following step.
    \begin{enumerate}[label=1.b.\roman{enumiii},leftmargin=*]
          \setcounter{enumiii}{2} 
          \item Define a valued binary CSP-instance \((\mathcal{V},\mathcal{D},\mathcal{C})\) as follows.
        \begin{itemize}
            \item \(\mathcal{V} = V(H)\). 
            \item \(\mathcal{D} = V(G)\), so an element of the domain no longer encodes a tuple of depth values.
            \item \(\mathcal{C} = \mathcal{C}' \cup \mathcal{C}''\), with the two parts of the constraint set as defined below.
    
            \item \(\mathcal{C'} = \{(u,u,c_{u})\mid u \in \mathcal{V}\}\) where \[c_u(v,v') = \begin{cases}
                \beta + 1 & \mbox{ if } v \neq v' \\
                &\lor (v \in T \cup \{r\} \land u \neq \varphi(v))\\
                &\lor (u \in \varphi(T \cup \{r\}) \land \varphi^{-1}(u) \neq v)\\
                0 & \mbox{otherwise.}
            \end{cases}\]
        
            \item \(\mathcal{C}'' = \{(u,v,c_{uv}) \mid uv \in E(H)\}\)
            where \(c_{uv}(u',v')\) is defined as follows for each \((u',v') \in \mathcal{D} \times \mathcal{D}\). Let~$S \subseteq [t]$ be the set encoded by the color of edge~$uv$. 
            Set \(c_{uv}(u',v')\) to the minimum value of \(|S|\cdot \ell(P')+\dig(P')\) over all paths \(P'\) between \(u'\) and \(v'\). This can computed using Dijkstra's algorithm for finding a shortest path between \(u'\) and \(v'\) in a weighting of graph \(G\) where for each edge \(e \in E(G)\), its weight is defined as \(|S|\cdot \ell(e)+\dig(e)\). If no such path exists, set the cost to \(\beta + 1\). 
            \end{itemize}
    \end{enumerate}
    
    Correctness of the algorithm follows analogously to the proof of \Cref{thm:XP-algo-tw}, with the key difference that the depth and its associated arguments are no longer relevant. 

    The running time analysis is also similar to \cref{thm:XP-algo-tw}. However, the running time of the CSP solving subroutine, which is \(|\mathcal{D}|^{\mathcal{O}(\operatorname{tw}(G_P(\mathcal{V}, \mathcal{D}, \mathcal{C}))}\), reduces due to the smaller domain~$\mathcal{D}$ of size~$|V(G)|$. Solving one instance of the newly defined CSP therefore takes time \(|V(G)|^{\mathcal{O}(|T|^2)}\). Thus, the total running time of this algorithm is \(f(t,|T|)\cdot |V(G)|^{\mathcal{O}(|T|^2)}\).
\end{proof}

\subsection{FPT-Algorithms for \PreLV}
\label{sec:PLVAlgorithm}
\input{PLV_algorithm}

%% file: figs/Abstract_Structure_instance.tex
\begin{tikzpicture}[scale=0.75]
    \tikzstyle{vertex}=[circle, fill=black, inner sep=1.5pt, scale=0.85]
    \tikzset{terminal/.style={fill=gray!30, draw=gray!100,  thick,minimum size=0pt, inner sep=2pt}}

    % bordering 
    % \draw[gray,draw opacity=.25,dashed] (-3,1) rectangle (3,-6.5);
    
    % vertices
    \node[vertex,label=above:\(r\)] (r) at (0,0) {};
    \node[vertex,label=left:\(v_1\)] (v1) at (-1.5,-1.5) {};
    \node[vertex,label=right:\(v_2\)] (v2) at (1.5,-1.5) {};
    \node[terminal,label=left:\(v_3\)] (v3) at (-1.5,-3) {};
    \node[vertex,label=left:\(v_4\)] (v4) at (0,-3) {};
    \node[vertex,label=right:\(v_5\)] (v5) at (1.5,-3) {};
    \node[vertex,label=left:\(v_6\)] (v6) at (0,-4.5) {};
    \node[terminal,label=left:\(v_7\)] (v7) at (-0.5,-6) {};
    \node[terminal,label=right:\(v_8\)] (v8) at (0.5,-6) {};
    \node[terminal,label=right:\(v_{10}\)] (v10) at (1.5,-4.5) {};
    \node[vertex,label=left:\(v_{11}\)] (v11) at (0,-1.5) {};
    \node[vertex,label=left:\(v_9\)] (v9) at (-0.75,-0.75) {};

    % edges
    \draw (v9) -- (v1);
    \draw (r) -- (v4);
    \draw (v1) -- (v3);
    \draw (r) -- (v9);
    \draw (v1) -- (v4);
    \draw (v6) -- (v7);
    \draw (r) -- (v2);
    \draw (v2) -- (v4);
    \draw (v2) -- (v5);
    \draw (v6) -- (v4);
    \draw (v6) -- (v8);
    \draw (v5) -- (v10);
\end{tikzpicture}

%% file: figs/Abstract_Structure_trees.tex
\begin{tikzpicture}[scale=0.75]
    \tikzstyle{vertex}=[circle, fill=black, inner sep=1.5pt, scale=0.85]
    \tikzset{terminal/.style={fill=gray!30, draw=gray!100,  thick,minimum size=0pt, inner sep=2pt}}

    % bordering 
    % \draw[gray,draw opacity=.25,dashed] (-3,1) rectangle (3,-6.5);
    
    % vertices
    \node[vertex,label=above:\(r\)] (r) at (0,0) {};
    \node[vertex,label=left:\(v_1\)] (v1) at (-1.5,-1.5) {};
    \node[vertex,label=right:\(v_2\)] (v2) at (1.5,-1.5) {};
    \node[terminal,label=left:\(v_3\)] (v3) at (-1.5,-3) {};
    \node[vertex,label=left:\(v_4\)] (v4) at (0,-3) {};
    \node[vertex,label=right:\(v_5\)] (v5) at (1.5,-3) {};
    \node[vertex,label=left:\(v_6\)] (v6) at (0,-4.5) {};
    \node[terminal,label=left:\(v_7\)] (v7) at (-0.5,-6) {};
    \node[terminal,label=right:\(v_8\)] (v8) at (0.5,-6) {};
    \node[terminal,label=right:\(v_{10}\)] (v10) at (1.5,-4.5) {};
    \node[vertex,label=left:\(v_9\)] (v9) at (-0.75,-0.75) {};

    % edges T1
    \draw[red, transform canvas={xshift=-0.7pt,yshift=0.7},very thick] (r) -- (v9);
    \draw[red, transform canvas={xshift=-0.7pt,yshift=0.7},very thick] (v9) -- (v1);
    \draw[red,very thick] (v1) -- (v3);

    % edges T2
    \draw[blue,transform canvas={xshift=0.7pt,yshift=-0.7pt},very thick] (r) -- (v9);
    \draw[blue,transform canvas={xshift=0.7pt,yshift=-0.7pt},very thick] (v9) -- (v1);
    \draw[blue,very thick] (v1) -- (v4);
    \draw[blue,transform canvas={xshift=-0.8pt},very thick] (v6) -- (v4);
    \draw[blue,very thick] (v6) -- (v7);

    % edges T3
    \draw[green!60!black,very thick] (r) -- (v2);
    \draw[green!60!black,very thick] (v2) -- (v4);
    \draw[green!60!black,very thick] (v2) -- (v5);
    \draw[green!60!black,transform canvas={xshift=0.8pt},very thick] (v6) -- (v4);
    \draw[green!60!black,very thick] (v6) -- (v8);
    \draw[green!60!black,very thick] (v5) -- (v10);
\end{tikzpicture}

%% file: figs/Abstact_Structure_Type_Of_Vertices.tex
\begin{tikzpicture}[scale=0.75]
    \tikzstyle{vertex}=[circle, fill=black, inner sep=1.5pt, scale=0.85]
    \tikzset{terminal/.style={fill=gray!30, draw=gray!100,  thick,minimum size=0pt, inner sep=2pt}}
    
    % vertices
    \node[vertex,label=above:\(r\)] (r) at (0,0) {};
    \node[vertex,label=left:\(v_1\),fill=violet] (v1) at (-1.5,-1.5) {};
    \node[vertex,label=right:\(v_2\), fill=orange] (v2) at (1.5,-1.5) {};
    \node[terminal,label=left:\(v_3\),fill=orange,draw=orange] (v3) at (-1.5,-3) {};
    \node[vertex,label=left:\(v_4\),fill=violet] (v4) at (0,-3) {};
    % \node[vertex,label=right:\(v_5\)] (v5) at (1.5,-3) {};
    \node[vertex,label=left:\(v_6\),fill=violet] (v6) at (0,-4.5) {};
    \node[terminal,label=left:\(v_7\),fill=orange,draw=orange] (v7) at (-0.5,-6) {};
    \node[terminal,label=right:\(v_8\),fill=orange,draw=orange] (v8) at (0.5,-6) {};
    \node[terminal,label=right:\(v_{10}\),fill=orange,draw=orange] (v10) at (1.5,-4.5) {};

    % edges
    \draw[red, transform canvas={xshift=-0.6pt,yshift=0.6pt},very thick] (r) -- (v1);
    \draw[blue, transform canvas={xshift=0.6pt,yshift=-0.6pt},very thick] (r) -- (v1);
    \draw[green!60!black,very thick] (r) -- (v2);
    \draw[red,very thick] (v1) -- (v3);
    \draw[blue,very thick] (v1) -- (v4);
    \draw[green!60!black,very thick] (v2) -- (v4);
    \draw[green!60!black,very thick] (v2) -- (v10);
    \draw[green!60!black,transform canvas={xshift=0.8pt},very thick] (v6) -- (v4);
    \draw[blue,transform canvas={xshift=-0.8pt},very thick] (v6) -- (v4);
    \draw[blue,very thick] (v6) -- (v7);
    \draw[green!60!black,very thick] (v6) -- (v8);
\end{tikzpicture}

%% file: PLV_algorithm.tex
\input{PLVDefinitions}
Contrasting the W[1]-hardness of \BLV and \ULV, which we prove in Section~\ref{sec:Hardness}, the \PreLV problem is in FPT. \Cref{thm:FPT-algo-tw} discusses this algorithm. In fact we can employ the Dreyfus-Wagner dynamic-programming approach that solves the \textsc{Steiner Tree} problem to solve \PreLVshort. For this algorithm, instead of guessing the abstract structure and using CSP for the realization of this abstract structure, which are the main bottlenecks of the previous algorithms, we branch on and fix the partition of terminals and use a dynamic-programming approach to check the existence of a desired solution.

\begin{theorem}       
\label{thm:FPT-algo-tw}
\PreLV can be solved in time \(2^{|T| \cdot \Oh(\log t)} \cdot |V(G)|^2\cdot \lambda\) for any arbitrary input instance \((G,\ell,\dig,r,T,d,\alpha,\beta,\lambda,t)\).
\end{theorem}
\begin{proof}
We find the solution to \((G,\ell,\dig,r,T,d,\alpha,\beta,\lambda,t)\) using a dynamic-programming approach. As a preparatory step, we first branch out and fix the partition \(\mathcal{P}\) of \(T\) into sets \(T_1,\ldots,T_t\) such that the demand constraint is satisfied, i.e.\ \(\sum\limits_{u \in T_i}d(u)\leq \alpha\) for all \(i \in [t]\). Then in the next step, the goal is to check the existence of a function \(f\colon T \to 2^{E(G)}\) that maps each terminal \(u\in T\) to a path in \(G\) from \(r\) to \(u\) that has \(\ell\)-length at most \(\lambda\) and \(\PLVcost{\mathcal{P}}{f}{r} \leq \beta\). 
This existence is checked through a dynamic-programming approach. For any subset \(T'\subseteq T\), we denote \(\{\,T_i\cap T' \mid T_i \in \mathcal{P}\}\setminus\{\emptyset\}\) by \(\mathcal{P}\doublecap T'\). For each subset \(T'\subseteq T\), \(v\in V(G)\), and each integer value \(\lambda'\) such that \(0 \leq \lambda' \leq \lambda \), let \(\text{d}_\mathcal{P}[T',v,\lambda']\) be defined as the minimum value of \(\PLVcost{\mathcal{P}\doublecap T'}{f'}{v}\) among all possible functions \(f'\colon T'\to 2^{E(G)}\) that map each terminal \(u\in T'\) to a path in \(G\) from \(v\) to \(u\) with \(\ell\)-depth at most \(\lambda'\). 
If no such set exists, then \(\text{d}_\mathcal{P}[T',v,\lambda']=+\infty\). To avoid several case distinctions, we also define~\(\text{d}_\mathcal{P}[T',v,\lambda']\) as~$+\infty$ for every~$\lambda' < 0$. Before discussing how we compute \(\text{d}_\mathcal{P}[T',v,\lambda']\), we present a claim that provides a structural tool. This tool helps us to combine the solutions to subproblems to construct a solution for \(\text{d}_\mathcal{P}[T',v,\lambda']\).

\begin{claim}
\label{claim:PLVStructure}
    Let \(\lambda'\) be a positive integer. For a given vertex \(v \in V(G)\) and any arbitrary function \(f\colon T' \to 2^{E(G)}\) that maps each terminal \(u\in T'\) to a simple walk in \(G\) from \(v\) to \(u\) such that \(\ell(f(u))\leq \lambda'\), there exists  a function \(f'\colon T'  \to  2^{E(G)}\) that maps each terminal \(u\in T'\) to a path in \(G\) from \(v\) to \(u\) such that \(\ell(f(u))\leq \lambda'\) and \(\PLVcost{\mathcal{P}}{f'}{v}\leq \PLVcost{\mathcal{P}}{f}{v}\).
\end{claim}
\begin{proof}
    Let \(f'=f\), initially. While there exists a terminal \(u\in T'\) such that \(f'(u)\) contains a cycle \(C\), do the following. For all \(u' \in T'\) (including \(u\)), if \(f'(u')\) contains \(C\), remove \(C\) from  \(f'(u')\). Clearly, the remaining part of the walk is still a walk from \(v\) to \(u'\). 
    
    Now, we verify that the conditions of the claim are maintained after each such step. Let \(f'_{\text{prev}}\) be the function \(f'\) before one such step and \(f'_{\text{new}}\) be the function \(f'\) after that step. The depth constraint holds since for each \(u'\in T'\), we only removed some edges from \(f'_{\text{prev}}(u')\) to reach \(f'_{\text{new}}(u')\) and so \(\ell(f'_{\text{new}}(u')) \leq \ell(f'_{\text{prev}}(u'))\). To show that  \(\PLVcost{\mathcal{P}}{f'_{\text{new}}}{v}\leq \PLVcost{\mathcal{P}}{f'_{\text{prev}}}{v}\), we show that for each \(e \in E(G)\) and \(T''\in \mathcal{P}\doublecap T'\), we have \(|\pfx{f'_{\text{new}}(T'')}{v}{e}|\leq |\pfx{f'_{\text{prev}}(T'')}{v}{e}|\) and \(\left|\bigcup\limits_{T''\in \mathcal{P}\doublecap T'}\pfx{f'_{\text{new}}(T'')}{v}{e}\right| \leq \left|\bigcup\limits_{T''\in \mathcal{P}\doublecap T'}\pfx{f'_{\text{prev}}(T'')}{v}{e}\right|\). 
    
    For the latter part, we present an injective function \(g\) from \(\bigcup\limits_{T''\in \mathcal{P}\doublecap T'}\pfx{f'_{\text{new}}(T'')}{v}{e}\) to \(\bigcup\limits_{T''\in \mathcal{P}\doublecap T'}\pfx{f'_{\text{prev}}(T'')}{v}{e}\). For each sequence \(S \in \bigcup\limits_{T''\in \mathcal{P}}\pfx{f'_{\text{new}}(T'')}{v}{e}\), if \(S\) is also in \(\bigcup\limits_{T''\in \mathcal{P}\doublecap T'}\pfx{f'_{\text{prev}}(T'')}{v}{e}\), then \(g(S)=S\). Otherwise, consider a terminal \(u\in T'\) such that \(\pfx{f'_{\text{new}}(u)}{v}{e}=S\). Let \(g(S)=f'_{\text{prev}}(u)\). Since \(S \notin \bigcup\limits_{T''\in \mathcal{P}\doublecap T'}\pfx{f'_{\text{prev}}(T'')}{v}{e}\), we have \((\pfx{f'_{\text{prev}}(u)}{v}{e}=S' \neq S\). Thus, \(f'_{\text{prev}}(u)\) is a path that contains the cycle \(C\) and this cycle also belongs to \(\pfx{f'_{\text{prev}}(u)}{v}{e}\). On the other hand, since no sequence in \(\bigcup\limits_{T''\in \mathcal{P}\doublecap T'}\pfx{f'_{\text{prev}}(T'')}{v}{e}\) contains \(C\), for all sequences that \(g(S)=S\), we know that \(S\) does not contain \(C\) and so, \(g(S)\) is distinct from the value of \(g\) for sequences not in \(\bigcup\limits_{T''\in \mathcal{P}\doublecap T'}\pfx{f'_{\text{prev}}(T'')}{v}{e}\). Therefore, to show that \(g\) is injective, it is sufficient to show that \(g(S)\) is distinct for all sequences like \(S\) that are not in \(\bigcup\limits_{T''\in \mathcal{P}\doublecap T'}\pfx{f'_{\text{prev}}(T'')}{v}{e}\). This is true since by the construction, for each such sequence \(S\), the sequence \(S\) should be \(g(S)\) where \(C\) is removed from it. Since \(g(S)\) is a walk, removing \(C\) from \(g(S)\) is unique and cannot result in two different sequences. Therefore, \(g\) is injective and so, \(\left|\bigcup\limits_{T''\in \mathcal{P}\doublecap T'}\pfx{f'_{\text{new}}(T'')}{v}{e}\right| \leq \left|\bigcup\limits_{T''\in \mathcal{P}\doublecap T'}\pfx{f'_{\text{prev}}(T'')}{v}{e}\right|\). 
    
    We can define an injective function \(g_{T''}\) from \(\pfx{f'_{\text{new}}(T'')}{v}{e}\) to \(\pfx{f'_{\text{prev}}(T'')}{v}{e}\) for each \(T'' \in \mathcal{P}\doublecap T'\) and show that it is injective in the same way. This yields that \(|\pfx{f'_{\text{new}}(T'')}{v}{e}|\leq |\pfx{f'_{\text{prev}}(T'')}{v}{e}|\) for each \(T' \in \mathcal{P}\). Therefore, we have \(\PLVcost{\mathcal{P}}{f'_{\text{new}}}{v}\leq \PLVcost{\mathcal{P}}{f'_{\text{prev}}}{v}\). 
    
    Since in each step the total number of edges in the walks is decreased, after at most \(|E(G)|\cdot |T|\) steps, we get rid of all cycles and so repeated vertices in all walks and obtain a function that maps the terminals to paths and satisfies the conditions of the claim.
\end{proof}

We compute \(\text{d}_\mathcal{P}[T',v,\lambda']\) as follows.

\begin{claim}
\label{claim:dpCorrectness}
    For a fixed partition \(\mathcal{P}\) of \(T\) into sets \(T_1,\ldots,T_t\), the values of \(\text{d}_\mathcal{P}[T',v,\lambda']\) satisfy the following recurrence.
    \[
\text{d}_{\mathcal{P}}[T',v,\lambda'] =
\begin{cases}
  0, & \mbox{ if } T'\subseteq\{v\},\\[6pt]
  \displaystyle
  \min_{\substack{T''\subset T' \\ v'\in N_G(v)}}
  
    \begin{aligned}[t]
     \Bigl(&  \text{d}_{\mathcal{P}}[T'\setminus T'',v,\lambda']+ \text{d}_{\mathcal{P}}[T'',v',\lambda'-\ell(vv')]
      \\[3pt]
        \qquad
      &
    +\, \ell(vv')\cdot\bigl|\mathcal{P}\doublecap T''\bigr|
    + \dig(vv')
      
    \Bigr),
    \end{aligned}
  & \text{otherwise}.
\end{cases}
\]
\end{claim}
\begin{proof}
    We discuss this claim for each part of the recurrence separately.
    \begin{itemize}
        \item \textbf{\(T'=\emptyset\).} Since \(T'=0\), the domain of \(f'\) is empty and so, the \PLVcostWord is \(0\).
        \item \textbf{\(T'=\{v\}\).} In this case, if \(f': T' \to 2^{E(G)}\) is the corresponding function to a solution of minimum \PLVcostWord, \(f'(v)\) is the single vertex \(v\) which results in a \PLVcostWord of \(0\) which is best-possible.
        
        \item \textbf{Otherwise.} We show the correctness of this case by showing it in two directions: The left-hand side is greater than or equal to the right-hand side, and vice versa.
          
          \textbf{Left \(\leq\) Right:} Consider an arbitrary choice of \(T''\subset T'\) and \(v' \in N_G(v)\). Let \(g\colon T''\to 2^{E(G)}\) be a function such that \(\text{d}_{\mathcal{P}}[T'',v',\lambda'-\ell(vv')]=\PLVcost{\mathcal{P}\doublecap T''}{g}{v'}\) and \(g':T'\setminus T'' \to 2^{E(G)}\) be a function such that \(\text{d}_{\mathcal{P}}[T'\setminus T'',v,\lambda']=\PLVcost{\mathcal{P}\doublecap (T'\setminus T'')}{g'}{v}\). Let \(g''\colon T' \to 2^{E(G)}\) be defined as follows for each~$u \in T'$. 
          \begin{itemize}
              \item If \(u \in T''\) and the next vertex after \(v'\) in \(g(u)\) is not \(v\), let \(g''(u)=g(u)\cup \{vv'\}\).
              \item If \(u \in T''\) and the next vertex after \(v'\) in \(g(u)\) is \(v\), let \(g''(u)\) be the subpath of \(g(u)\) from \(v\) to \(u\).
              \item If \(u \in T'\setminus T''\), let \(g''(u)=g'(u)\).
          \end{itemize} 
          For the third case, clearly \(g''(u)\) is a path in \(G\) from \(v\) to \(u\) of length at most \(\lambda'\) as \(g'\) is a solution to \(\text{d}_{\mathcal{P}}[T'\setminus T'',v,\lambda']\). This is also true for the second case as we used a subpath of \(g(u)\). Furthermore, for the first case, \(g(u)\) is a path from \(v'\) to \(u\) of length at most \(\lambda'-\ell(vv')\) since \(g\) is a solution to \(\text{d}_{\mathcal{P}}[T'',v',\lambda'-\ell(vv')]\). Since \(v'\) is an endpoint of this path and \(v' \in N_G(v)\), by adding edge \(vv'\), we are extending this path to obtain a  walk from \(v\) to \(u\) of length at most \(\lambda'\) in \(G\). Note that this walk is simple as \(g(u)\) is a path and does not contain any repeated edges and does not contain the edge \(vv'\) as otherwise \(v\) is the neighbor of \(v'\) in \(g(u)\). Therefore, \(g''\) is a function that maps each terminal \(u\) of \(T'\) to a simple walk of length at most \(\lambda'\) in \(G\) that is from \(v\) to \(u\). 
          We can rewrite \(\PLVcost{\mathcal{P}\doublecap T'}{g''}{v}\) as
          \begin{align*}
    &\sum\limits_{\substack{S\in\mathcal{P}\doublecap T'\\e \in E(G)}} \left(\ell(e)\cdot |\pfx{g''(S)}{v}{e}|\right)
        &+\sum\limits_{e \in E(G)} \left(\dig(e) \cdot \left|\bigcup\limits_{S\in \mathcal{P}\doublecap T'}\pfx{g''(S)}{v}{e}\right|\right).
              \end{align*}
        This yields that \(\PLVcost{\mathcal{P}\doublecap T'}{g''}{v}\) is at most
        \begin{align} \label{prefix:cost} &\sum\limits_{\substack{S\in\mathcal{P}\doublecap T'\\e \in E(G)}} \ell(e)\cdot (|\pfx{g''(S\cap T'')}{v}{e}|+|\pfx{g''(S\cap (T'\setminus T''))}{v}{e}|)\\
        &+\sum\limits_{e \in E(G)} \left(\dig(e) \cdot \left(\left|\bigcup\limits_{S\in \mathcal{P}\doublecap T''}\pfx{g''(S)}{v}{e}\right|+\left|\bigcup\limits_{S\in \mathcal{P}\doublecap (T'\setminus T'')}\pfx{g''(S)}{v}{e})\right|\right)\right). \nonumber
        \end{align}
          Recalling that for each \(S\in\mathcal{P}\doublecap T'\) and each edge \(e \in E(G)\), if \(S \cap T'' \neq \emptyset\), the simple walks in \(\pfx{g''(S\cap T'')}{v}{e}\) consist of either the single edge \(vv'\) or a path in \(\pfx{g(S\cap T'')}{v'}{e}\) extended by \(vv'\), we can rewrite \(|\pfx{g''(S\cap T'')}{v}{e}|\) as \(|\pfx{g(S\cap T'')}{v'}{e}|\) if \(e \neq vv'\) and as \(|\pfx{g(S\cap T'')}{v'}{e}|+1\) if \(e=vv'\). Clearly, if \(S \cap T''=\emptyset\), then \(\pfx{g''(S\cap T'')}{v}{e}=\pfx{g(S\cap T'')}{v'}{e}=\emptyset\). Moreover,  for each \(S\in\mathcal{P}\doublecap T'\) and each edge \(e \in E(G)\), we have \(\pfx{g''(S\cap (T'\setminus T''))}{v}{e}=\pfx{g'(S\cap (T'\setminus T''))}{v}{e}\) since \(g\) and \(g'\) have the same value for all terminals in \(T'\setminus T''\). On the other hand, for each edge \(e \in E(G)\), the walks in \(\bigcup\limits_{S\in \mathcal{P}\doublecap T''}\pfx{g''(S)}{v}{e}\) are consisting of either the single edge \(vv'\) or a path in \(\bigcup\limits_{S\in \mathcal{P}\doublecap T''}\pfx{g(S)}{v'}{e}\) extended by \(vv'\) and hence 
          \[
          \left|\bigcup\limits_{S\in \mathcal{P}\doublecap T''}\pfx{g''(S)}{v}{e}\right|=
          \begin{cases}
  \left|\bigcup\limits_{S\in \mathcal{P}\doublecap T''}\pfx{g(S)}{v'}{e}\right|, & e \neq vv',\\[6pt]
  \left|\bigcup\limits_{S\in \mathcal{P}\doublecap T''}\pfx{g(S)}{v'}{e}\right|+1, & \text{otherwise}.\\[6pt]
  \end{cases}
  \] 
  Also, for each edge \(e \in E(G)\), we have \[\bigcup\limits_{S\in \mathcal{P}\doublecap (T'\setminus T'')}\pfx{g''(S)}{v}{e})=\bigcup\limits_{S\in \mathcal{P}\doublecap (T'\setminus T'')}\pfx{g'(S)}{v}{e}).\]
  Note that we need to pay once more for the digging cost of \(vv'\) as it is added to the beginning of some paths. Putting all the above together, we can conclude that \eqref{prefix:cost} is equal to
          \begin{align*}
              &\begin{aligned}[t]
    & \sum\limits_{\substack{S\in\mathcal{P}\doublecap T''\\e \in E(G)}} \left(\ell(e)\cdot \left|\pfx{g(S)}{v'}{e}\right|\right)+\sum\limits_{\substack{S\in\mathcal{P}\doublecap (T'\setminus T'')\\e \in E(G)}} \left(\ell(e)\cdot\left|\pfx{g'(S)}{v}{e}\right|\right)\\
    + & \ \ell(vv')\cdot|\mathcal{P}\doublecap T''|\\
        + & \ \sum\limits_{e \in E(G)} \left(\dig(e) \cdot \left(\left|\bigcup\limits_{S\in \mathcal{P}\doublecap T''}\pfx{g(S)}{v'}{e}\right|+\left|\bigcup\limits_{S\in \mathcal{P}\doublecap (T'\setminus T'')}\pfx{g'(S)}{v}{e})\right|\right)\right)\\
        + & \ \dig(vv')
              \end{aligned}\\
              = & \ \PLVcost{\mathcal{P}\doublecap T''}{g}{v'}+\PLVcost{\mathcal{P}\doublecap(T'\setminus T'')}{g'}{v}+\ell(vv')\cdot|\mathcal{P}\doublecap T''|+\dig(vv')\\
              = & \ \text{d}_{\mathcal{P}}[T'',v',\lambda'-\ell(vv')]+\text{d}_{\mathcal{P}}[T'\setminus T'',v,\lambda']+\ell(vv')\cdot|\mathcal{P}\doublecap T''|+\dig(vv').
          \end{align*}
             By \Cref{claim:PLVStructure}, we can find a function \(g'''\) that maps paths to the terminals and is a candidate for a solution to \(\text{d}_{\mathcal{P}}[T',v,\lambda']\) and hence we have \(\text{d}_{\mathcal{P}}[T',v,\lambda']\leq \PLVcost{\mathcal{P}\doublecap T'}{g'''}{v}\leq \PLVcost{\mathcal{P}\doublecap T'}{g''}{v}\). Thus, we have
            \begin{align*}
            \text{d}_{\mathcal{P}}[T',v,\lambda']&\leq \PLVcost{\mathcal{P}\doublecap T'}{g''}{v}\\&\leq \text{d}_{\mathcal{P}}[T'',v',\lambda'-\ell(vv')]+\text{d}_{\mathcal{P}}[T'\setminus T'',v,\lambda']+\ell(vv')\cdot|\mathcal{P}\doublecap T''|+\dig(vv'), \end{align*}
            and since \(T'' \subset T'\) and \(v' \in N_G(v)\) are chosen arbitrarily, we can conclude the first inequality.
            
            \textbf{Left \(\geq\) Right:} Let \(g\colon T' \to 2^{E(G)}\) be a function such that \(\text{d}_{\mathcal{P}}[T',v,\lambda']=\PLVcost{\mathcal{P}\doublecap T'}{g}{v}\). Consider an arbitrary terminal \(u \in T'\) and let \(v'\) be the neighbor of \(v\) in \(g(u)\). We define \(T_{v'}\) as the set of all terminals \(u' \in T'\) such that \(vv'\) appears as the first edge in \(g(u')\). We define \(g'\colon T_{v'} \to 2^{E(G)}\) and \(g''\colon (T' \setminus T_{v'}) \to 2^{E(G)}\) as follows. For each \(u' \in T_{v'}\), let \(g'(u')\) be the path \(g(u')\) where the edge \(vv'\) is removed; Thus, it is a path from \(v'\) to \(u'\) with length at most \(\lambda'-\ell(vv')\). Additionally, for all \(u'' \in T' \setminus T_{v'}\), let \(g''(u')=g(u'')\) which is a path from \(v\) to \(u'\) of length at most \(\lambda'\). 
            Since the paths of \(g(T_{v'})\) start with \(vv'\) and the other paths of \(g\) start with other edges, the prefixes arising from the ones in \(g(T_{v'})\) and the ones not in \(g(T_{v'})\) for all edges are distinct. Thus, we can rewrite \(\text{d}_{\mathcal{P}}[T',v,\lambda']=\PLVcost{\mathcal{P}\doublecap T'}{g}{v}\) as 
            \begin{align*}
            &\sum\limits_{\substack{T''\in\mathcal{P}\doublecap T'\\e \in E(G)}} \left(\ell(e)\cdot \left|\pfx{g(T'')}{v}{e}\right|\right)
        +\sum\limits_{e \in E(G)} \left(\dig(e) \cdot \left|\bigcup\limits_{T''\in \mathcal{P}}\pfx{g(T'')}{v}{e}\right|\right)\\
        =&
        \begin{aligned}[t]
        &\sum\limits_{\substack{T''\in\mathcal{P}\doublecap T_{v'}\\e \in E(G)}} \left(\ell(e)\cdot \left|\pfx{g(T'')}{v}{e}\right|\right)\\
        +&\sum\limits_{\substack{T''\in\mathcal{P}\doublecap (T'\setminus T_{v'})\\e \in E(G)}} \left(\ell(e)\cdot \left|\pfx{g(T'')}{v}{e}\right|\right)\\
        +&\sum\limits_{e \in E(G)} \left(\dig(e) \cdot \left|\bigcup\limits_{T''\in\mathcal{P}\doublecap T_{v'}}\pfx{g(T'')}{v}{e}\right|\right)\\
        +&\sum\limits_{e \in E(G)} \left(\dig(e) \cdot\left|\bigcup\limits_{T''\in\mathcal{P}\doublecap (T'\setminus T_{v'})}\pfx{g(T'')}{v}{e}\right|\right).
        \end{aligned}
            \end{align*}
        If \(e=vv'\), then \(\pfx{g(T_{v'}\cap T'')}{v}{e}\) consists of \(\pfx{g'(T_{v'}\cap T'')}{v'}{e}\) where each is extended by \(vv'\) and the single edge \(vv'\) and otherwise, \(\pfx{g(T_{v'}\cap T'')}{v}{e}\) is equal to \(\pfx{g''(T_{v'}\cap T'')}{v}{e}\) where each is extended by \(vv'\). Therefore, we have
        \begin{align*}
            &\sum\limits_{\substack{T''\in\mathcal{P}\doublecap T_{v'}\\e \in E(G)}} \left(\ell(e)\cdot \left|\pfx{g(T'')}{v}{e}\right|\right)
        +\sum\limits_{\substack{T''\in\mathcal{P}\doublecap (T'\setminus T_{v'})\\e \in E(G)}} \left(\ell(e)\cdot \left|\pfx{g(T'')}{v}{e}\right|\right)\\
        =&
        \begin{aligned}[t]&\sum\limits_{\substack{T''\in\mathcal{P}\doublecap T_{v'}\\e \in E(G)}} \left(\ell(e)\cdot \left|\pfx{g'(T'')}{v'}{e}\right|\right)+\ell(vv')\cdot |\mathcal{P}\doublecap T_{v'}|\\
        +&\sum\limits_{\substack{T''\in\mathcal{P}\doublecap (T'\setminus T_{v'})\\e \in E(G)}} \left(\ell(e)\cdot \left|\pfx{g''(T'')}{v}{e}\right|\right),
        \end{aligned}
        \end{align*}
        and
        \begin{align*}
            &\begin{aligned}[t]
            &\sum\limits_{e \in E(G)} \left(\dig(e) \cdot \left|\bigcup\limits_{T''\in\mathcal{P}\doublecap T_{v'}}\pfx{g(T'')}{v}{e}\right|\right)\\
        +&\sum\limits_{e \in E(G)} \left(\dig(e) \cdot\left|\bigcup\limits_{T''\in\mathcal{P}\doublecap (T'\setminus T_{v'})}\pfx{g(T'')}{v}{e}\right|\right)
        \end{aligned}\\
        =&
        \begin{aligned}[t]&\sum\limits_{e \in E(G)} \left(\dig(e) \cdot \left|\bigcup\limits_{T''\in\mathcal{P}\doublecap T_{v'}}\pfx{g'(T'')}{v'}{e}\right|\right)+\dig(vv')\\
        +&\sum\limits_{e \in E(G)} \left(\dig(e) \cdot\left|\bigcup\limits_{T''\in\mathcal{P}\doublecap (T'\setminus T_{v'})}\pfx{g''(T'')}{v}{e}\right|\right).
        \end{aligned}
        \end{align*}
        This implies that \(\PLVcost{\mathcal{P}\doublecap T'}{g}{v}\) can be written as
        \[
            \PLVcost{\mathcal{P
            }\doublecap T_{v'}}{g'}{v'}+\PLVcost{\mathcal{P}\doublecap(T'\setminus T_{v'})}{g''}{v}+\ell(vv')\cdot |\mathcal{P}\doublecap T_{v'}|+\dig(vv').
            \]

        Clearly, \(g'\) is a candidate for a solution to \( \text{d}_{\mathcal{P}}[T_{v'},v',\lambda'-\ell(vv')]\) and \(g''\) is a candidate for a solution to \(\text{d}_{\mathcal{P}}[T'\setminus T_{v'},v,\lambda']\). Therefore, we have
        \begin{align*}
            &\PLVcost{\mathcal{P
            }\doublecap T_{v'}}{g'}{v'}+\PLVcost{\mathcal{P}\doublecap(T'\setminus T_{v'})}{g''}{v}+\ell(vv')\cdot |\mathcal{P}\doublecap T_{v'}|+\dig(vv')\\
            &\geq \text{d}_{\mathcal{P}}[T_{v'},v',\lambda'-\ell(vv')]+\text{d}_{\mathcal{P}}[T'\setminus T_{v'},v,\lambda']+\ell(vv')\cdot |\mathcal{P}\doublecap T_{v'}|+\dig(vv').
        \end{align*}
        Considering that \(T_{v'}\) is a possible value of \(T''\), we can conclude that
        \begin{align*}
            \text{d}_{\mathcal{P}}[T',v,\lambda']&=\PLVcost{\mathcal{P}\doublecap T'}{g}{v}\\
            &\geq \text{d}_{\mathcal{P}}[T_{v'},v',\lambda'-\ell(vv')]+\text{d}_{\mathcal{P}}[T'\setminus T_{v'},v,\lambda']+\ell(vv')\cdot |\mathcal{P}\doublecap T_{v'}|+\dig(vv')\\
            &\geq \min_{\substack{T''\subset T' \\ v'\in N_G(v)}}
    \begin{aligned}[t]
     \Bigl(&  \text{d}_{\mathcal{P}}[T'\setminus T'',v,\lambda']+ \text{d}_{\mathcal{P}}[T'',v',\lambda'-\ell(vv')]
      \\[3pt]
        \qquad
      &
    +\, \ell(vv')\cdot\bigl|\mathcal{P}\doublecap T''\bigr|
    + \dig(vv') 
    \Bigr).
    \end{aligned}
        \end{align*}
        This implies the other direction and hence the claim.\qedhere
    \end{itemize}
\end{proof}

The algorithm can be described as follows.

\begin{enumerate}
    \item For each partition \(\mathcal{P}\) of \(T\) into \(t\) sets \(T_1,\ldots,T_t\), \label{PLValg:PartitionBranch}
    \begin{enumerate}
        \item If for some \(i\in [t]\), \(\sum\limits_{u \in T_i}d(u) > \alpha\), skip this branch.\label{PLValg:DemandCheck}
        \item Compute \(\text{d}_\mathcal{P}[T,r,\lambda]\) using \Cref{claim:dpCorrectness}. If \(\text{d}_\mathcal{P}[T,r,\lambda]\leq \beta\), then return `yes'.\label{PLValg:DPCheck}
    \end{enumerate}
    \item If none of the cases resulted in a `yes'-output, output `no'.\label{PLValg:ReturnNo}
\end{enumerate}

We show the correctness of this algorithm and afterward, we its running time analysis.

\paragraph*{Correctness.} Let \(\mathcal{P}'\) and \(f'\) be a solution for  \((G,\ell,\dig,r,T,d,\alpha,\beta,\lambda,t)\) that has the minimum \PLVcostWord. Clearly, if \(\PLVcost{\mathcal{P}'}{f'}{r}\leq \beta\), the answer to this input instance is `yes' and otherwise, the answer is `no'. 

First, assume that the answer is `yes'. Since we got through all possible partitions in line~\ref{PLValg:PartitionBranch}, at some point the algorithm takes \(\mathcal{P}=\mathcal{P'}\). Then, since \(\mathcal{P}'\) captures the demand constraint it passes the check of line~\ref{PLValg:DemandCheck}. Since this is a solution with minimum \PLVcostWord, the value of \(\text{d}_\mathcal{P'}[T,r,\lambda]\) should be \(\PLVcost{\mathcal{P}'}{f'}{r}\) by definition. Therefore, line~\ref{PLValg:DPCheck} returns `yes'.

Conversely, if the answer is `no', we know that for all partitions \(\mathcal{P}\) and all functions \(f\), we have \(\PLVcost{\mathcal{P}}{f}{r}> \beta\). Therefore, for each choice of \(\mathcal{P}\), the value of \(\text{d}_\mathcal{P'}[T,r,\lambda]\) is more than \(\beta\) by the definition of \(\text{d}_\mathcal{P'}[T,r,\lambda]\). Therefore, the algorithm does not return in line~\ref{PLValg:DPCheck} and proceeds with the rest, reaching line~\ref{PLValg:ReturnNo} which returns `no'.

\paragraph*{Running Time Analysis.} There are \(t^|T|\) possible options for \(\mathcal{P}\) in step~\ref{PLValg:PartitionBranch}. Then, checking the step~\ref{PLValg:DemandCheck} of the algorithm can be done in \(\mathcal{O}(|T|)\). The computation of \(\text{d}_{\mathcal{P}}[T',v,\lambda']\) in step~\ref{PLValg:DPCheck} takes \(|T|\cdot |V(G)|\cdot \lambda \cdot 2^{|T|}\cdot |V(G)|\cdot |T|^2\). Step~\ref{PLValg:ReturnNo} can be done in constant time. Therefore, the total running time is \(2^{|T| \cdot \Oh(\log t)} \cdot |V(G)|^2\cdot \lambda\).
\end{proof}

%% file: PLVDefinitions.tex
In an attempt to understand how sharing digging costs causes intractability, we define a variation of the power network design problem whose cost function is easier to optimize. While the cost of a solution still arises as a digging term plus the sum of the cable lengths, we change the definition of the digging term. In the definition of \cref{sec:Introduction}, whenever multiple trees contain the same edge the digging costs of the corresponding trench only have to be paid once. The variation we introduce effectively means that when multiple trees contain the same edge, the digging costs can only be shared when the path from the root to this edge is the same in both trees. This is formalized through a technical definition of the different \emph{prefixes} leading to an edge. The digging costs are then paid once for each different prefix leading to that trench. We show that under this cost model, the problem is FPT parameterized by~$|T|+t$ via dynamic programming.

We now formally define and analyze the resulting \PreLV problem. For the definition we need the concept of a \emph{walk} in a graph, which is an alternating sequence of vertices and edges~$(v_0, e_1, \ldots, v_k)$ such that edge~$e_i$ is incident on~$v_{i-1}$ and~$v_i$ for all~$i \in [k]$. A walk is \emph{simple} (sometimes called \emph{trail} in the literature) when no edge occurs more than once. Note that a path is a simple walk with no repeated vertex. If the direction of the walk is irrelevant, we can also consider a simple walk as a set of edges.

\begin{definition}
    Let \(G\) be a graph and let~$W=(v_0, e_1, \ldots, v_k)$ be a simple walk in~$G$ starting at~$r$. For each edge \(e_i \in E(W)\) we define the prefix of~\(e_i\) in~\(W\), denoted by \(\pfx{W}{r}{e_i}\), as the subwalk of \(W\) from \(r\) to \(v_i\); in other words, \(\pfx{W}{r}{e_i}=(v_0,e_1,\ldots,e_i,v_i)\).
\end{definition}

This definition can be used to define the \emph{prefix set} of an edge in a given set of simple walks.

\begin{definition}
    Let \(G\) be a graph and let \(\mathcal{F}=\{W_1,\ldots,W_x\}\) be a set of simple walks in~$G$ starting at~$r \in V(G)$. For each \(e \in E(G)\), we define the prefix set of \(e\) in \(\mathcal{F}\), denoted by \(\pfxset{\mathcal{F}}{r}{e}\), as the set of all prefixes for \(e\) among the walks in \(\mathcal{F}\) that contain \(e\). More formally: \[\pfxset{\mathcal{F}}{r}{e} := \{\pfx{P_i}{r}{e}\mid i \in [x] \land e \in E(P_i)\}.\] If there is no walk in \(\mathcal{F}\) that contains \(e\), then \(\pfxset{\mathcal{F}}{r}{e} =\emptyset\).
\end{definition}

The main idea behind our definition of the prefix-sharing version of the problem is as follows. We want to modify the cost function so that if~$H_1, \ldots, H_t$ is a set of trees that forms a solution to a power network design problem, the number of times we pay the cost of digging a trench at a particular edge~$e \in E(G)$ is equal to the number of different routes by which a tree~$H_i$ can arrive at edge~$e$. Intuitively, this captures the idea that if the power network is installed starting at the root vertex, then digging costs are shared among different trees that place a cable in the same trench, as long as the work crews arrive at this edge by the same path from the root. This can be formalized as saying that the number of times the digging cost~$\dig(e)$ has to be paid is equal to the size of the prefix set~$\pfxset{{F}}{r}{e}$, when taking~$\mathcal{F}$ to be the following family of paths: for each tree~$H_i$ that contains edge~$e$, the set~$\mathcal{F}$ contains the path in~$H_i$ connecting~$e$ to the root~$r$.

When incorporating this concept into the definition of a power network design problem described in \cref{sec:Introduction}, a surprising behavior arises. To minimize the resulting cost function it might be beneficial for a single power cable~$H_i$ to be placed in a network of trenches that does not form a tree, but admits cycles among its trenches: the (seemingly superfluous) additional edges that form cycles can lead to a reduction in the number of different prefixes formed in combination with other power cables~$H_j$ (see~\Cref{fig:PLVNonTreeInstanceSolution}). In our formulation of the problem, we therefore allow cycles in the connected subgraph~$H_i$ that represents the trenches into which a single power cable is placed. To avoid short-circuits, the route by which electricity is transferred from the root to a terminal is still a simple path, but different terminals fed by the same power cable can use different paths through the (possibly cyclic) network of trenches that is dug. Using this intuition and the definitions above, the \PreLV problem can now be defined as follows.

\begin{figure}
    \centering
    \begin{subfigure}[t]{0.9\textwidth}
            \centering
            \input{figs/PLVNonTreeInstanceBoundedDepth}
            \subcaption{An instance of \PreLVshort. The lengths of edges \(v_1v_3\), \(v_2v_3\) are \(10\), and the lengths of \(rv_1\),\(rv_2\) are \(20\). The digging cost of edges \(v_4v_5\) and \(v_4v_6\) is \(1000\). The remaining lengths and digging costs are \(1\). The depth restriction~$\lambda$ is set to \(33\). The demand values~$d$ are chosen via \cref{lemma:UniqueTerminalPartitioning} such that the only possible partition \(\mathcal{P}\) of terminals is \(\{t_1,t_3\},\{t_5,t_7\},\{t_2,t_4,t_6,t_{8}\}\).}
            \label{fig:PLVNonTreeInstance}
        \end{subfigure}
        \hfill
        \begin{subfigure}[t]{0.45\textwidth}
            \centering
            \input{figs/PLVTreeSolution1}
            \subcaption{An acyclic solution for (a). In this solution, we pay for the digging cost of each edge other than \(v_1v_3\) once, which gives a total digging cost of \(2013\). We pay for the length of \(rv_1\) twice, for \(rv_2\), \(v_2v_3\), \(v_3v_4\) three times (once for each color), for \(v_4v_5\) and \(v_4v_6\) twice, and for the edges incident on terminals once. So, the total length cost is \(2 \times 20 +3 \times (20+10+1) +2 \times 2+8=145\). Thus, the total cost is \(2013+145=2158\).}
            \label{fig:PLVTreeSolution1}
        \end{subfigure}
        \hfill
    \begin{subfigure}[t]{0.45\textwidth}
            \centering
            \input{figs/PLVNonTreeSolution}
            \subcaption{A cyclic solution to (a) of smaller cost. Note that the red subgraph contains the cycle \((r,v_1,v_3,v_2)\). In this solution, we pay the digging cost of all edges except \(v_3v_4\) once, and those of \(v_3v_4\) twice. So, the total digging cost of this solution is \(2014\). We pay for the length of the edges around the cycle \(rv_1v_3v_2\) twice, for \(v_3v_4\) four times consisting of 2 reds and once for blue and green, for \(v_4v_5\) and \(v_4v_6\) twice, and for the edges incident on terminals once. So, the total length is \(2\times 60+4\times 1+2\times 2+8\times 1=136\) which yields a total cost of \(2150\).}
            \label{fig:PLVNonTreeSolution}
        \end{subfigure}
    \caption{Illustration of the \PreLVshort problem and the possibility of cyclic solution structures. (a) An instance of \PreLVshort. (b),(c) an acyclic solution, and optimal and cyclic solution respectively. In these solutions, edges of the same color form solution paths to terminals belonging to the same part of \(\mathcal{P}\). The edge styles distinguish the different root-terminal paths. One can argue that for the instance in (a), all optimal solutions are such that the union of the paths to~$t_4$ and~$t_6$ contain a cycle, as follows. 
    Due to the depth restriction, the paths for \(t_1,t_2,t_7,t_8\) are unique. If the subgraph that covers \(t_2,t_4,t_6,t_8\) is acyclic, then at least one of \(v_1v_3\) and \(v_2v_3\) is excluded from the corresponding paths and at least one of them must be used for connectivity to \(t_4\) and \(t_6\). So, exactly one of these edges belongs to the paths to \(t_4\) and \(t_6\). If the paths to \(t_3,t_4,t_5,t_6\) do not pass through the same edge among \(\{v_1v_3, v_2v_3\}\), the solution will have a digging cost larger than \(3000\). Otherwise, all must pass through the same edge which results in a suboptimal solution similar to (b).}
    \label{fig:PLVNonTreeInstanceSolution}
\end{figure}
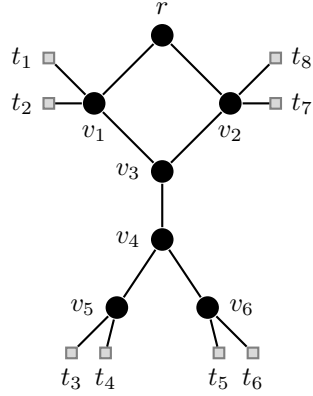
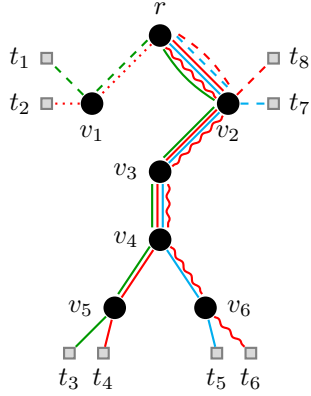
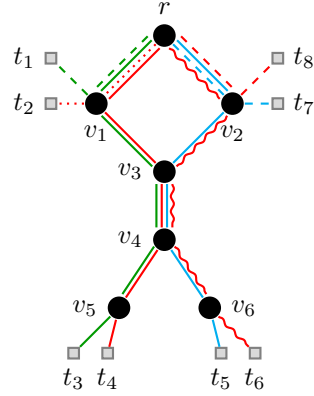

\begin{center}
    \begin{mathproblem}{\PreLV (\PreLVshort)}[\(t + |T|\)]
        \textit{\textbf{Instance:}} An undirected graph \(G\) with \emph{edge lengths} \(\ell \colon E(G) \to \mathbb{N}\) and \emph{digging costs} \(\dig \colon E(G) \to \mathbb{N}\), a \emph{root vertex} \(r \in V(G)\), a set of degree-\(1\) \emph{terminal} vertices \(T \subseteq V(G)\) with \emph{demands} \(d \colon T \to \mathbb{N}\), and integers \(\alpha\), \(\beta\), \(\lambda\), and \(t\), encoded in unary.
        
        \vspace{0.3cm}
        
        \textit{\textbf{Question:}} Does there exist a partition \(\mathcal{P}\) of \(T\) into \(t\) sets \(T_1,\ldots,T_t\) and a function \(f \colon T \to 2^{E(G)}\) that maps each terminal \(u\in T\) to a path in \(G\) from \(r\) to \(u\) such that:
        \begin{itemize}
            \item for all \(u\in T\), the path \(P = f(u)\) satisfies~$\sum_{e \in E(P)} \ell(e) \leq \lambda$ [\emph{depth} constraint],
            \item \(\sum\limits_{u\in T_i}d(u) \leq \alpha\) for all \(i \in [t]\) [\emph{demand} constraint], and
            \item \(\PLVcost{\mathcal{P}}{f}{r} \leq \beta \)?[\emph{cost} constraint]
        \end{itemize}
        The value \(\PLVcost{\mathcal{P}}{f}{r}\), referred to as \PLVcostWord, is defined as follows:
        \par
        \(
        \PLVcost{\mathcal{P}}{f}{r}:=   \sum\limits_{\substack{T'\in\mathcal{P}\\e \in E(G)}} \left(\ell(e)\cdot \left|\pfx{f(T')}{r}{e}\right|\right)
        +\sum\limits_{e \in E(G)} \left(\dig(e) \cdot \left|\bigcup\limits_{T'\in \mathcal{P}}\pfx{f(T')}{r}{e}\right|\right).
        \)
    \end{mathproblem}
\end{center}

%% file: figs/PLVNonTreeInstanceBoundedDepth.tex
\begin{tikzpicture}[scale=0.6]
    \tikzstyle{vertex}=[circle, fill=black, inner sep=1.5pt, scale=2]
    \tikzset{terminal/.style={fill=gray!30, draw=gray!100,  thick,minimum size=0pt, inner sep=2pt}}

    % bordering 
    % \draw[gray,draw opacity=.25,dashed] (-3,1) rectangle (3,-6.5);
    
    % vertices
    \node[vertex,label=above:\(r\)] (r) at (0,0) {};
    \node[vertex,label=below:\(v_1\)] (v1) at (-1.5,-1.5) {};
    \node[vertex,label=below:\(v_2\)] (v2) at (1.5,-1.5) {};
    \node[vertex,label=left:\(v_3\)] (v3) at (0,-3) {};
    \node[vertex,label=left:\(v_4\)] (v4) at (0,-4.5) {};
    \node[vertex,label=left:\(v_5\)] (v5) at (-1,-6) {};
    \node[vertex,label=right:\(v_6\)] (v6) at (1,-6) {};
    \node[terminal,label=left:\(t_1\)] (t1) at (-2.5,-0.5) {};
    \node[terminal,label=left:\(t_2\)] (t2) at (-2.5,-1.5) {};
    % \node[terminal,label=left:\(t_3\)] (t3) at (-2.5,-2.5) {};
    
    \node[terminal,label=below:\(t_3\)] (t4) at (-2,-7) {};
    \node[terminal,label=below:\(t_4\)] (t5) at (-1.25,-7) {};
    % \node[terminal,label=below:\(t_6\)] (t6) at (-0.5,-7) {};
    
    % \node[terminal,label=below:\(t_7\)] (t7) at (0.5,-7) {};
    \node[terminal,label=below:\(t_5\)] (t9) at (1.25,-7) {};
    \node[terminal,label=below:\(t_6\)] (t8) at (2,-7) {};
    
    \node[terminal,label=right:\(t_{8}\)] (t10) at (2.5,-0.5) {};
    \node[terminal,label=right:\(t_{7}\)] (t11) at (2.5,-1.5) {};
    % \node[terminal,label=right:\(t_{12}\)] (t12) at (2.5,-2.5) {};

    % edges
    \draw[thick] (r) -- (v1);
    \draw[thick] (r) -- (v2);
    \draw[thick] (v1) -- (v3);
    \draw[thick] (v2) -- (v3);
    \draw[thick] (v3) -- (v4);
    \draw[thick] (v4) -- (v5);
    \draw[thick] (v4) -- (v6);
    
    \draw[thick] (v1) -- (t1);
    \draw[thick] (v1) -- (t2);
    % \draw[thick] (v1) -- (t3);

    \draw[thick] (v5) -- (t4);
    \draw[thick] (v5) -- (t5);
    % \draw[thick] (v5) -- (t6);
    
    % \draw[thick] (v6) -- (t7);
    \draw[thick] (v6) -- (t8);
    \draw[thick] (v6) -- (t9);
    
    \draw[thick] (v2) -- (t10);
    \draw[thick] (v2) -- (t11);
    % \draw[thick] (v2) -- (t12);
\end{tikzpicture}

%% file: figs/PLVTreeSolution1.tex
\begin{tikzpicture}[scale=0.6]
    \tikzstyle{vertex}=[circle, fill=black, inner sep=1.5pt, scale=2]
    \tikzset{terminal/.style={fill=gray!30, draw=gray!100,  thick,minimum size=0pt, inner sep=2pt}}
    \tikzset{
      curly/.style={
        decorate,
        decoration={snake, amplitude=0.2mm, segment length=2mm}
      }
    }

    % bordering 
    % \draw[gray,draw opacity=.25,dashed] (-3,1) rectangle (3,-6.5);
    
    % vertices
    \node[vertex,label=above:\(r\)] (r) at (0,0) {};
    \node[vertex,label=below:\(v_1\)] (v1) at (-1.5,-1.5) {};
    \node[vertex,label=below:\(v_2\)] (v2) at (1.5,-1.5) {};
    \node[vertex,label=left:\(v_3\)] (v3) at (0,-3) {};
    \node[vertex,label=left:\(v_4\)] (v4) at (0,-4.5) {};
    \node[vertex,label=left:\(v_5\)] (v5) at (-1,-6) {};
    \node[vertex,label=right:\(v_6\)] (v6) at (1,-6) {};
    \node[terminal,label=left:\(t_1\)] (t1) at (-2.5,-0.5) {};
    \node[terminal,label=left:\(t_2\)] (t2) at (-2.5,-1.5) {};
    % \node[terminal,label=left:\(t_3\)] (t3) at (-2.5,-2.5) {};
    
    \node[terminal,label=below:\(t_3\)] (t3) at (-2,-7) {};
    \node[terminal,label=below:\(t_4\)] (t4) at (-1.25,-7) {};
    % \node[terminal,label=below:\(t_6\)] (t6) at (-0.5,-7) {};
    
    % \node[terminal,label=below:\(t_7\)] (t7) at (0.5,-7) {};
    \node[terminal,label=below:\(t_5\)] (t6) at (1.25,-7) {};
    \node[terminal,label=below:\(t_6\)] (t5) at (2,-7) {};
    
    \node[terminal,label=right:\(t_{8}\)] (t7) at (2.5,-0.5) {};
    \node[terminal,label=right:\(t_{7}\)] (t8) at (2.5,-1.5) {};
    % \node[terminal,label=right:\(t_{12}\)] (t12) at (2.5,-2.5) {};

    % edges t1
    \draw[green!60!black, thick, transform canvas={xshift=-1.5pt,yshift=1.5pt}, thick,dashed] (r) -- (v1);
    \draw[green!60!black, thick, thick,dashed] (v1) -- (t1);

    % edges t2
    \draw[red,transform canvas={xshift=0pt,yshift=0pt}, thick,dotted] (r) -- (v1);
    \draw[red, thick,dotted] (v1) -- (t2);

    % edges for t3
    \draw[green!60!black,transform canvas={xshift=-1.25pt,yshift=-1.25pt}, bend right=12, thick] (r) to (v2);
    \draw[green!60!black, transform canvas={xshift=-2pt,yshift=2pt}, thick] (v2) to (v3);
    \draw[green!60!black,transform canvas={xshift=-3pt}, thick] (v3) -- (v4);
    \draw[green!60!black,transform canvas={xshift=-1.3pt,yshift=1.3}, thick] (v4) -- (v5);
    \draw[green!60!black, thick] (v5) -- (t3);

    % edges for t4
    \draw[red,transform canvas={xshift=1.5pt,yshift=1.5pt}, thick] (r) to (v2);
    \draw[red, transform canvas={xshift=-1pt,yshift=1pt}, thick] (v2) to (v3);
    \draw[red,transform canvas={xshift=-1pt}, thick] (v3) -- (v4);
    \draw[red,transform canvas={xshift=0pt,yshift=0}, thick] (v4) -- (v5);
    \draw[red, thick] (v5) -- (t4);

    % edges for t5
    \draw[red, transform canvas={xshift=-1.25pt,yshift=-1.25pt}, thick,curly] (r) -- (v2);
    \draw[red,transform canvas={xshift=1.3pt,yshift=-1.3pt}, thick,curly] (v2) -- (v3);
     \draw[red,transform canvas={xshift=3}, thick,curly] (v3) -- (v4);
    \draw[red,transform canvas={xshift=1.3pt,yshift=1.3}, thick,curly] (v4) -- (v6);
    \draw[red, thick,curly] (v6) -- (t5);

    % edges for t6
    \draw[cyan,transform canvas={xshift=0.25pt,yshift=0.25pt},thick] (r) -- (v2);
    \draw[cyan,transform canvas={xshift=0pt,yshift=0pt}, thick] (v2) -- (v3);
    \draw[cyan,transform canvas={xshift=1}, thick] (v3) -- (v4);
    \draw[cyan,transform canvas={xshift=0pt,yshift=0pt}, thick] (v4) -- (v6);
    \draw[cyan, thick] (v6) -- (t6);

    % edges t8
    \draw[cyan, transform canvas={xshift=3pt,yshift=3pt}, thick,dashed] (r) -- (v2);
    \draw[cyan, thick,dashed] (v2) -- (t8);
    
    % edges t7
    \draw[red, transform canvas={xshift=3pt,yshift=3pt}, bend left=12, thick,dashed] (r) to (v2);
    \draw[red, thick,dashed] (v2) -- (t7);
\end{tikzpicture}

%% file: figs/PLVNonTreeSolution.tex
\begin{tikzpicture}[scale=0.6]
    \tikzstyle{vertex}=[circle, fill=black, inner sep=1.5pt, scale=2]
    \tikzset{terminal/.style={fill=gray!30, draw=gray!100,  thick,minimum size=0pt, inner sep=2pt}}
    \tikzset{
      curly/.style={
        decorate,
        decoration={snake, amplitude=0.2mm, segment length=2mm}
      }
    }

    % bordering 
    % \draw[gray,draw opacity=.25,dashed] (-3,1) rectangle (3,-6.5);
    
    % vertices
    \node[vertex,label=above:\(r\)] (r) at (0,0) {};
    \node[vertex,label=below:\(v_1\)] (v1) at (-1.5,-1.5) {};
    \node[vertex,label=below:\(v_2\)] (v2) at (1.5,-1.5) {};
    \node[vertex,label=left:\(v_3\)] (v3) at (0,-3) {};
    \node[vertex,label=left:\(v_4\)] (v4) at (0,-4.5) {};
    \node[vertex,label=left:\(v_5\)] (v5) at (-1,-6) {};
    \node[vertex,label=right:\(v_6\)] (v6) at (1,-6) {};
    \node[terminal,label=left:\(t_1\)] (t1) at (-2.5,-0.5) {};
    \node[terminal,label=left:\(t_2\)] (t2) at (-2.5,-1.5) {};
    % \node[terminal,label=left:\(t_3\)] (t3) at (-2.5,-2.5) {};
    
    \node[terminal,label=below:\(t_3\)] (t3) at (-2,-7) {};
    \node[terminal,label=below:\(t_4\)] (t4) at (-1.25,-7) {};
    % \node[terminal,label=below:\(t_6\)] (t6) at (-0.5,-7) {};
    
    % \node[terminal,label=below:\(t_7\)] (t7) at (0.5,-7) {};
    \node[terminal,label=below:\(t_5\)] (t6) at (1.25,-7) {};
    \node[terminal,label=below:\(t_6\)] (t5) at (2,-7) {};
    
    \node[terminal,label=right:\(t_{8}\)] (t7) at (2.5,-0.5) {};
    \node[terminal,label=right:\(t_{7}\)] (t8) at (2.5,-1.5) {};
    % \node[terminal,label=right:\(t_{12}\)] (t12) at (2.5,-2.5) {};

    % edges t1
    \draw[green!60!black, thick, transform canvas={xshift=-3pt,yshift=3pt}, thick,dashed] (r) -- (v1);
    \draw[green!60!black, thick, thick,dashed] (v1) -- (t1);

    % edges t2
    \draw[red,transform canvas={xshift=0pt,yshift=0pt}, thick,dotted] (r) -- (v1);
    \draw[red, thick,dotted] (v1) -- (t2);

    % edges for t3
    \draw[green!60!black,transform canvas={xshift=-1.5pt,yshift=1.5pt}, thick] (r) -- (v1);
    \draw[green!60!black, transform canvas={xshift=-1.3pt,yshift=-1.3pt}, thick] (v1) -- (v3);
    \draw[green!60!black,transform canvas={xshift=-3pt}, thick] (v3) -- (v4);
    \draw[green!60!black,transform canvas={xshift=-1.3pt,yshift=1.3}, thick] (v4) -- (v5);
    \draw[green!60!black, thick] (v5) -- (t3);

    % edges for t4
    \draw[red,transform canvas={xshift=1.5pt,yshift=-1.5pt}, thick] (r) -- (v1);
    \draw[red, transform canvas={xshift=0pt,yshift=0pt}, thick] (v1) -- (v3);
    \draw[red,transform canvas={xshift=-1pt}, thick] (v3) -- (v4);
    \draw[red,transform canvas={xshift=0pt,yshift=0}, thick] (v4) -- (v5);
    \draw[red, thick] (v5) -- (t4);
    
    % edges for t5
    \draw[red, transform canvas={xshift=-1.5pt,yshift=-1.5pt}, thick,curly] (r) -- (v2);
    \draw[red,transform canvas={xshift=1.3pt,yshift=-1.3pt}, thick,curly] (v2) -- (v3);
    \draw[red,transform canvas={xshift=3}, thick,curly] (v3) -- (v4);
    \draw[red,transform canvas={xshift=1.3pt,yshift=1.3}, thick,curly] (v4) -- (v6);
    \draw[red, thick,curly] (v6) -- (t5);

    % edges for t6
    \draw[cyan,transform canvas={xshift=1.5pt,yshift=1.5pt},thick] (r) -- (v2);
    \draw[cyan,transform canvas={xshift=0pt,yshift=0pt}, thick] (v2) -- (v3);
    \draw[cyan,transform canvas={xshift=1}, thick] (v3) -- (v4);
    \draw[cyan,transform canvas={xshift=0pt,yshift=0pt}, thick] (v4) -- (v6);
    \draw[cyan, thick] (v6) -- (t6);

    % edges t8
    \draw[cyan, transform canvas={xshift=0pt,yshift=0pt}, thick,dashed] (r) -- (v2);
    \draw[cyan, thick,dashed] (v2) -- (t8);
    
    % edges t7
    \draw[red, transform canvas={xshift=3pt,yshift=3pt}, thick,dashed] (r) -- (v2);
    \draw[red, thick,dashed] (v2) -- (t7);
\end{tikzpicture}

%% file: Hardness_Results.tex
\label{sec:Hardness}
For our hardness results, we provide reductions from \GT, which is defined as follows.
\begin{center}
\begin{mathproblem}{\GT}[\(k\)]
        \textit{\textbf{Instance:}} A \(k \times k\) matrix where each cell \((i,j) \in [k]^2\) contains a set of pairs \(S_{i,j} \subseteq [n] \times [n]\).

        \vspace{0.3cm}
        
        \textit{\textbf{Question:}} Can we choose a pair \(s_{i,j} \in S_{i,j}\) for each cell such that
        \begin{itemize}
            \item if \(s_{i,j}=(a,b)\) and \(s_{i,j'}=(a',b')\), then \(a=a'\); and
            \item if \(s_{i,j}=(a,b)\) and \(s_{i',j}=(a',b')\), then \(b=b'\)?
        \end{itemize}
    \end{mathproblem}
\end{center}
\GT cannot be solved in time \(f(k)\cdot n^{o(k)}\) under the ETH and is \W[1]-hard \cite{ParAlg}.

\subsection{\W[1]-hardness of \BLV on planar graphs}
\label{sec:W1HardBLV}
\input{BLV_reduction}

\subsection{\W[1]-hardness of \ULV}
\input{ULV_reduction}

%% file: BLV_reduction.tex
In our first reduction, we show that the algorithm from \Cref{thm:XP-algo-planar} is ETH-tight.

\begin{theorem}
    \label{thm:bddepth-ETHlowerbound}
    The \BLV problem is W[1]-hard for~$t=2$ and under the ETH admits no \(f(|T|,t)\cdot (|V(G)| \cdot \lambda)^{f'(t) \cdot o(|T|)}\)-time algorithm for any computable functions \(f\) and \(f'\), even when restricted to planar input graphs.
\end{theorem}
\begin{proof}
    From a \GT instance \(\left(S_{i,j}\right)_{(i,j) \in [k]^2}\), construct an \BLVshort-instance as follows.

    \paragraph*{\(G\), \(r\) and \(T\)}
    Let \(G^0\) be a grid graph with \(N := k \cdot n\) rows and \(M := k \cdot (k+1) \cdot n\) columns.
    For notational convenience, for \(i \in [N]\) and \(j \in [M]\), we denote the vertex in the \(i\)-th row and \(j\)-th column of \(G^0\) by \(v[i,j]\). For \(i \in [N]\) and \(j \in [M-1]\), we denote the edge \(v[i,j]v[i,j+1]\) by \(e[i,j]\).

    We define \(G\) by adding the following vertices to \(G^0\): 
    \begin{itemize}
        \item A vertex \(\rcons\) which we call \emph{row consistency root} and make adjacent to all \(v[i,1]\) with \(i \in [N]\) and to a new terminal \(p(\rcons)\) as pendant.
        \item A vertex \(\rselect\) which we call \emph{selection root} and make adjacent to all \(v[1,(j-1)(k+1) + 1]\) with \(j \in [k\cdot n]\) and to a new terminal \(p(\rselect)\) as pendant.
        \item The root vertex \(r\) of the \BLVshort-instance which we make adjacent to \(\rcons\) and \(\rselect\).
        \item For each \(i \in [k]\),
        \begin{itemize}
            \item a vertex \(t_i\) called a \emph{row consistency preterminal} (for value of row \(i\)) and connect it to \(v[i',M]\) for all \(i' \in \{(i-1) \cdot n + 1, \dotsc, i\cdot n\}\) and to a new terminal \(p(t_i)\) as pendant.
            \item a vertex \(t'_i\) called a \emph{selection preterminal} (for value of column \(i\)) and make adjacent to \(v[N,(i-1)n(k+1) + i'(k+1)]\) for all \(i' \in [n]\) and to a new terminal \(p(t'_i)\) as pendant.
        \end{itemize}
    \end{itemize}
    Up to the definition of lengths,  digging costs and demands this concludes the description of all parts of the input pertaining to \(G\). 
    By construction \(G\) is planar; an embedding is indicated in \Cref{fig:blvreduction-G}.

    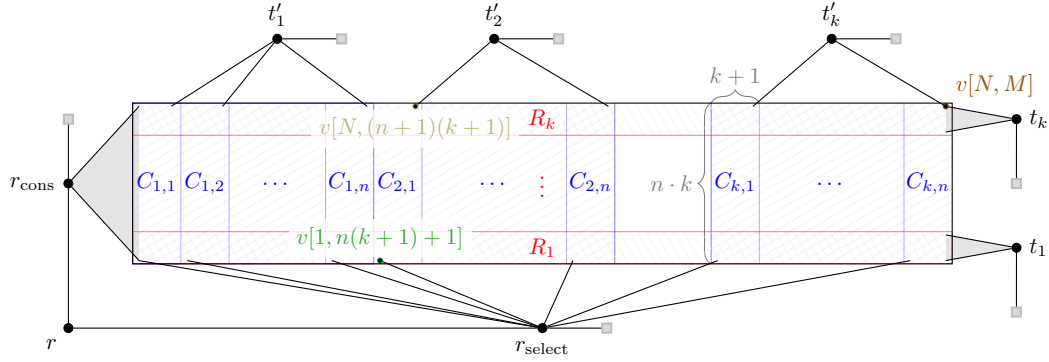
\begin{figure}[!ht]
        \centering
        \input{figs/planarhardnessconstr.tex}
        \caption{Illustration of the graph \(G\) in the proof of \Cref{thm:bddepth-ETHlowerbound}.
        Terminals are shown as squares; only selected vertices of \(G^0\) and their indices are indicated explicitly.
        The hatched rectangle represents \(G^0\), with \(v[1,1]\) at the bottom left and \(v[N,M]\) at the top right corner.
        A more detailed view of a part of the construction using an explicit example is given in \Cref{fig:planarhardness-zoomin}.}
        \label{fig:blvreduction-G}
    \end{figure}

    \subparagraph*{\(t\), \(d\) and \(\alpha\).}
    We fix \(t=2\), i.e.\ a hypothetical solution for the \BLVshort-instance under construction consists of two trees \(H_1\) and \(H_2\). Consider the following partition \(\mathcal{P}\) of terminals into two parts: 1) \(p(\rcons)\) along with all row consistency terminals, and 2) \(p(\rselect)\) along with all selection preterminals; more formally, \(\mathcal{P}=\{\{p(\rcons),p(t_1), \dotsc, p(t_k)\},\{p(\rselect),p(t'_1), \dotsc, p(t'_k)\}\}\).
    By \Cref{lemma:UniqueTerminalPartitioning}, there exists a function \(f: T \to [|U|^4]\) such that \(\mathcal{P}\) is the unique partition of terminals into at most two parts for which the summation of \(f\)-values in each part is at most \(\sum\limits_{u \in T}\frac{f(u)}{2}\). Since the demand values are integers, by setting the demand function \(d=f\) and \(\alpha = \lfloor \sum\limits_{u \in T}\frac{d(u)}{2} \rfloor\), the aim is to find two trees \(H_1,H_2\) where \(H_1\) contains the first part and \(H_2\) contains the second part of the terminals.

    \paragraph*{Reduction overview}
    Before fixing the parts of the \BLVshort-instance that still need to be specified, we provide some intuition about what we will try to achieve, along with some useful notation for specific parts of the construction.\\
    The lengths of the edges, in particular between \(r\) and \(\rcons\) and \(r\) and \(\rselect\), and \(\beta\) will be chosen in such a way that each terminal-root-path in \(H_1\) passes through \(\rcons\) and each terminal-root-path in \(H_2\) passes through \(\rselect\).\\
    \(\lambda\) will be set in such a way that these paths in \(H_2\) will have to be of the minimum possible length.
    This will constrain each \(t'_j\)-\(\rselect\)-path to be contained in exactly one of the \(n\) subgrids (see \Cref{fig:blvreduction-G}) of the form
    \[C_{j,j''}:= G^0[\{v[i,j'] \mid i \in [N] \land \exists z \in [k+1] \ j'= (j-1)(k+1)n + (j''-1)(k + 1) + z\}]\]
    for some \(j'' \in [n]\) and contain precisely \(k\) horizontal edges.
    This allows us to encode a choice of the second entry for the \(j\)-th column of the \GT instance through \(j''\) such that the \(\rselect\)-\(t'_j\)-path in \(H_2\) is contained in \(C_{j,j''}\).
    Similarly, we will be able to assume that each \(t_i\)-\(\rcons\)-path traverses exactly one of the \(n\) rows in the subgrid \(R_i := G^0[\{v_{i',j} \mid i' \in \{(i-1)n + 1, \dotsc, i \cdot n, j \in [M]\}]\).
    This allows us to encode a choice of the first entry of the pairs selected for the \(i\)-th row of the \GT instance through the index of the row whose vertices are on the \(\rcons\)-\(t_i\)-path in \(H_1\).\\
    We need to ensure that the implied encoding of a pair \((a,b)\) for each cell \((i,j) \in [k]^2\) satisfies that \((a,b) \in S_{i,j}\).
    For this we set the digging costs in a way that each \(t_i\)-\(\rcons\)-path in \(H_1\) shares a (by the assumed structure of paths in \(H_1\) necessarily horizontal) edge with each \(t'_j\)-\(\rselect\)-path in \(H_2\) and we set the lengths of a horizontal edge prohibitively large for the depth requirements of \(H_2\) whenever it is between the \(j'\) and \((j'+ 1)\)-th column in \(C_{j,j''}\), on the \(i'\)-th row in \(R_i\), and \((i',j'') \notin S_{i,j}\) -- thereby disallowing pairs that are not in the appropriate set \(S_{i,j}\) -- or \(j' \neq i\) -- thereby ensuring that we check the compatibility of the selected second entry with the selected row for each \(i \in [k]\).
    We now continue with the formal construction of the \BLVshort-instance to implement the described intended behavior.

    \subparagraph*{\(\ell\), \(\dig\), \(\lambda\) and \(\beta\)}
    To specify the edge lengths, we distinguish some horizontal edges.
    In particular, we will set the lengths of some horizontal edges to be higher (\(\blocklength\)) than that of others (\(1\)) and refer to the longer edges (see \cref{fig:BLVStructureInstance}) as \emph{blocking}:
    \begin{itemize}
        \item For \((j,j'') \in ([k] \times [n]) \setminus \{(k,n)\}\), we call horizontal edges that connect \(C_{j,j''}\) to \(C_{j,j''+1}\) if \(j''<n\) and to \(C_{j+1,1}\) otherwise \emph{column inconsistency blocking}.
        I.e.\ the set of column inconsistency blocking edges is given by \[\{e \in e[i,(j-1)n(k+1) + j''(k+1)] \mid i \in [N], j \in [k], j'' \in [n]\}.\]
        \item For \(j,j' \in [k]\), and \(j'' \in [n]\), we call edges between vertices in the \(j'\)-th and \((j'+1)\)-th column in \(C_{j,j''}\) \emph{invalid selection blocking} if they are not in \(R_{j'}\) or they are in the \(i'\)-th row of \(R_{j'}\) for some \(i' \in [n]\) for which \((i',j'') \notin S_{j',j}\).
        I.e.\ the set of invalid selection blocking edges is given by
    \begin{align*}\{e = e^{C_{j,j''}}[i^*,j'] \mid & j,j'\in [k], j'' \in [n], i^* \in [N] \land \\
    & \exists i \in [k]\exists i'\in [n] \exists j^* \in [M] \ e = e^{R_i}[i',j^*] \land ((i',j'') \notin S_{i,j} \lor j' \neq i)\}.\end{align*}
    \end{itemize}
    Horizontal edges of \(G^0\) that are not blocking, i.e.\ edges between vertices in the \(j'\)-th and \((j'+1)\)-th column in \(C_{j,j''}\) for some \(j,j' \in [k]\), and \(j'' \in [n]\) for which \(\exists i'\in [n] \exists j^* \in [M] \ e = e^{R_{j'}}[i',j^*] \land (i',j'') \in S_{j',j}\), receive length \(1\). There are at most~$kn$ such edges per row.
    \begin{figure}[!ht]
        \centering
        \begin{subfigure}[b]{0.95\textwidth}
            \centering
            \input{figs/BLV_Instance}
            \subcaption{The \BLVshort-instance corresponding to a \GT instance with \(n = 3\), \(k = 2\), \(S_{1,1}=\{(2,1),(3,2)\}\), \(S_{1,2}=\{(2,3),(3,1)\}\), \(S_{2,1}=\{(1,1)\}\), and \(S_{2,2}=\{(1,3),(2,2)\}\). Equal colors and edge styles indicate equal lengths. Dotted and dashed edges represent column inconsistency and invalid selections and are too long to appear in an \(\rselect\)-\(t'_1\)-path in \(H_2\). The purple edges are the only length-\(1\) edges. These are the only horizontal edges that can appear in such a path.}
            \label{fig:BLVStructureInstance}
        \end{subfigure}
        \hfill
        \begin{subfigure}[b]{0.95\textwidth}
            \centering
            \input{figs/BLVSolutionInstance}
            \subcaption{A solution for the instance in (a). Red and blue denote the two trees; purple edges are shared and indicate the selected cell pairs \(s_{1,1}=(2,1)\), \(s_{1,2}=(2,3)\), \(s_{2,1}=(1,1)\), and \(s_{2,2}=(1,3)\).}
            \label{fig:BLVSolutionInstance}
        \end{subfigure}
        \hfill
        \begin{subfigure}[b]{0.95\textwidth}
            \centering
            \input{figs/ULV_Solution_In_BLV_Structure}
            \subcaption{Illustrating why the construction fails for \ULVshort when \(\lambda=+\infty\). The red and blue trees share purple edges and achieve lower cost than the solution in~\cref{fig:BLVSolutionInstance}.}
            \label{fig:ULVinBLVStructure}
        \end{subfigure}
        \caption{}
    \end{figure}
    See \Cref{fig:planarhardness-zoomin} for a specific example of blocking edges and an indication of their intended role.
    \begin{figure}[!ht]
        \centering
        \input{figs/planarhardnesspartexample}
        \caption{Example of \(G[\{\rselect,t'_1\} \cup V(C_{1,1}) \cup V(C_{1,2}) \cup V(C_{1,3})]\) for \(n = 3\), \(k = 2\), \(S_{1,1} = \{(1,1),(2,3)\}\), \(S_{2,1} = \{(1,1),(2,1)\}\), i.e.\ the part of \(G\) encoding a selection of pairs in the first column of the \GT instance.
        Equal colors and drawing styles of edges indicate equal lengths.
        Column inconsistency and invalid selection blocking edges are dotted and dashed, respectively.
        Their lengths will be too high to be included in an \(\rselect\)-\(t'_1\)-path in \(H_2\).
        The bold purple edges are the only horizontal edges that can be included in such a path.
        We can see that in this example this leaves us with two possible \(\rselect\)-\(t'_1\)-paths; one corresponding to selecting \((1,1) \in S_{1,1}\) and \((1,1) \in S_{2,1}\) and the other corresponding to  \((1,1) \in S_{1,1}\) and \((2,1) \in S_{2,1}\).
        \label{fig:planarhardness-zoomin}}
    \end{figure}
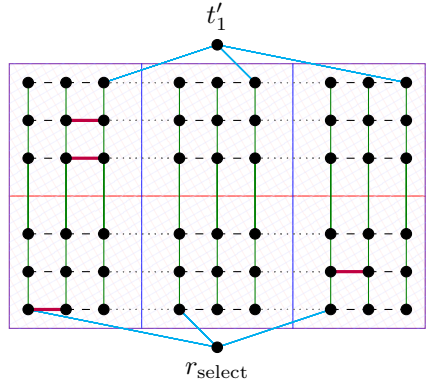

    The edges with precisely one endpoint in \(V(G^0)\) receive length \(X := \sum_{e \in E(G^0)} 2\ell(e) + \dig(e) + 1\) each. (While we have not specified all lengths and the digging costs of edges in \(E(G^0)\) yet, they will not depend on \(X\).)
    The edges between \(r\) and \(\rcons\) and between \(r\) and \(\rselect\) receive length \(X^*\) each, where \(X^*\) denotes the sum of the length and digging costs of all edges other than \(r\rselect\), \(r\rcons\), and edges incident to terminals. 
    The edges \(\rcons p(\rcons)\) and \(\rselect p(\rselect)\) receive length \(\lambda - X^*\), with~$\lambda$ as defined below. For \(i \in [k]\), the edge~\(t'_ip(t'_i)\) receives length~$2$ and the edge~\(t_ip(t_i)\) receives length~\(\lambda - X^* -2X - (M-1)(\blocklength)\). 
    The length of all vertical edges is set to \(\vertlength\).
    
    The digging cost of horizontal edges is \(\hdigcost\), and the digging cost of all other edges is~\(0\).
    Finally, we set \(\lambda = X^* + 2X + 2+(N - 1)(\vertlength) + k\). 
    This will constrain each root-terminal path in \(H_2\) to be of minimum length.
    The total cost threshold \(\beta\) is the sum of \(k\)-times this minimum length minus \(X^*\) (one for each path to a selection terminal), \(k\)-times \(\lambda\) minus \(X^*\) (one for each path to a row consistency terminal), two times \(\lambda\) to connect \(p(\rcons)\) via $\rcons$ and \(p(\rselect)\) via $\rselect$, and lastly the digging cost necessarily incurred by \(k\) horizontal rows of \(G^0\).
    An alternative way of expressing \(\beta\) is as the sum of:
    \begin{itemize}
        \item \(2\lambda\) for the cumulative length of \(r\rselect\),\(\rselect p(\rselect)\), \(r\rcons\) and \(\rcons p(\rcons)\).
        \item \(k(\lambda - X^*)\) for the paths to each row consistency terminal without their shared first edge.
        \item \(2kX + 2k\) corresponding to the cumulative length of the edges between each selection preterminal and its terminal pendant, one length-$X$ edge between each selection preterminal and \(G^0\), an one length-$X$ edge between \(\rselect\) and \(G^0\) for each selection preterminal.
        \item \(k \cdot ((N-1) (\vertlength) + k)\) corresponding to the minimum length of \(k\) edge-disjoint paths that each traverse all rows of \(G^0\) while each using exactly \(k\) horizontal edges of length one. (Note that summing this item with the previous gives~$k(\lambda-X^*)$.)
        \item \(k \cdot (M-1)(\hdigcost)\) corresponding to the digging cost of one horizontal edge per pair of adjacent columns in \(G^0\) per consistency preterminal.
    \end{itemize} 
    This completes the \BLVshort-instance.

    We will show that \(\left(S_{i,j}\right)_{(i,j) \in [k]^2}\) and \(\mathcal{I} := (G,\ell,\dig,r,T,d,\alpha,\beta,\lambda,t)\) are equivalent.
    
    \begin{claim}
    \label{claim:planarhardness-forward}
    If \(\left(S_{i,j}\right)_{(i,j) \in [k]^2}\) is a yes-instance then \(\mathcal{I}\) is a yes-instance.
    \end{claim}
    \begin{claimproof}
    Fix a solution \(((a_i,b_j))_{(i,j) \in [k]^2}\) for \(\left(S_{i,j}\right)_{(i,j) \in [k]^2}\).

    We will show that \(H_1\) and \(H_2\) defined as follows are a solution for \(\mathcal{I}\).
        Let \(H_1\) be the tree induced on the endpoints of edges \(r\rcons\), \(\rcons p(\rcons)\), for each \(i \in [k]\) the edges in the \(a_i\)-th row in \(R_i\), and \(t_i p(t_i)\).
        Let \(H_2\) be the tree induced on the endpoints of edges \(r\rselect\), \(\rselect p(\rselect)\), for each \(j \in [k]\) the edges on the unique path in \(C_{j,b_j}\) from a neighbor of~$\rselect$ to a neighbor of~$t'_j$ whose only horizontal edges are those that go from column \(i\) to \(i+1\) in \(C_{j,b_j}\) and are on the \(a_i\)-th row of \(R_i\), and \(t'_j p(t'_j)\).
    
        It is straightforward to verify that \(H_1\) and \(H_2\) are indeed trees that cover all of \(T\).
        Moreover, by construction each of \(H_1\) and \(H_2\) contains terminals with demand precisely \(\alpha\).
        
        When considering the depth of \(H_1\), the root-distance of \(p(\rcons)\) is precisely \(\lambda\), and for each \(i \in [k]\), the \(r\)-\(p(t_i)\)-path has length at most \(\lambda - (M-1)(\blocklength) +(M-1-k)(\blocklength)+k\), because the \(a_i\)-th row of \(R_i\) contains at least \(k\) edges which are of length \(1\), namely those that are between the \(i\) and \((i+1)\)-th column of \(C_{j,b_j}\) for each \(j \in [k]\).
        This is smaller than \(\lambda\).

        When considering the depth of \(H_2\), the root-distance of \(p(\rselect)\) is precisely \(\lambda\), and for each \(j \in [k]\) the \(r\)-\(p(t'_j)\)-path has length \(2 + 2X + X^* + (N-1)(\vertlength) +k\), because the selection of the horizontal edges on this path is such that they are each of length \(1\).
        This is what we chose \(\lambda\) as.

        Finally, one may verify that the cost of \(H_1\) and \(H_2\) does not exceed \(\beta\) by using the previous length bounds of the paths which only intersect in \(r\rcons\), \(\rcons p(\rcons)\), \(r\rselect\) and \(\rselect p(\rselect)\) and the fact that all horizontal edges are contained in \(H_1\) and we can hence bound the overall digging costs by \(k \cdot (N-1) (\hdigcost)\).
    \end{claimproof}
    
    For the backward direction, assume that \(\{H_1,H_2\}\) is a solution for \(\mathcal{I}\).
    To derive a solution for \GT, it is useful to analyze the behavior of \(H_1\) and \(H_2\).
    As announced, by our construction, paths to terminals all use \(\rcons\) and \(\rselect\) respectively.
    
    \begin{claim}
    \label{claim:treeroots}
        Every path from \(r\) to a terminal in \(H_1\) contains \(r\rcons\), and every path from \(r\) to a terminal in \(H_2\) contains \(r\rselect\).
    \end{claim}
    \begin{claimproof}
    Assume for contradiction that this is not the case.
    Then, because \(\{r\rcons,r\rselect\}\) separates \(r\) from all terminals, this means that \(r\rcons \in E(H_2)\) or \(r\rselect \in E(H_1)\). Note that both edges have length~$X^*$.
    
    By the choice of \(\lambda\) and \(\ell\), \(r\rcons \in E(H_1)\) because \(\lambda\) is equal to the length of the unique shortest path between \(r\) and \(p(\rcons)\) and this path contains \(r\rcons\).
    Similarly, by the choice of \(\lambda\) and \(\ell\), \(r\rselect \in E(H_2)\) because \(\lambda\) is equal to the length of the unique shortest path between \(r\) and \(p(\rselect)\) and this path contains \(r\rselect\). Hence from the edges incident to~$r$, at least three edges of length~$X^*$ contribute to the cost of solution~$\{H_1, H_2\}$. But the definition of \(X^*\) implies \(3X^* > \beta\), hence the cost of solution \(\{H_1,H_2\}\) exceeds~$\beta$; a contradiction.
    \end{claimproof}
    By construction, for each \(j \in [k]\), the path from \(\rselect\) to \(p(t'_j)\) in \(H_2\) contains \(t'_j\).
    The implied paths from \(\rselect\) to \(t'_j\) in \(H_2\) turn out to have a very specific structure.
         \begin{claim}
    \label{claim:planarhardness-selectionpaths}
        For each \(j \in [k]\) the interior of the path in \(H_2\) from \(\rselect\) to \(t'_j\) is contained in \(C_{j,j''}\) for some \(j'' \in [n]\) and consists of exactly \(N-1\) vertical and exactly \(k\) horizontal edges, of which the latter have length \(1\).
    \end{claim}
    \begin{claimproof}
        Consider an arbitrary \(j \in [k]\).
        By \cref{claim:treeroots}, the path from~$r$ to~$p(t'_j)$ in~$H_2$ contains the edge~$r \rselect$ of length~$X^*$, a length-$X$ edge into~$G^0$, and a length-$X$ edge from~$G^0$ to~$t'_j$. Our choice of~$X^*, X$ and the depth threshold~$\lambda$ then ensure that no further length-$X$ edge can be included in the path.
        Hence, the interior \(P\) of the \(\rselect\)-\(t'_j\)-path in \(H_2\) has to be contained in \(G^0\). 
        Such a path necessarily uses at least \(N-1\) vertical edges.
        Let us calculate an upper bound on the length of \(P\) implied by \(\lambda\).
        In particular, its length can be at most \(\lambda - \ell(p(t'_j)t'_j) - \ell(r\rselect) - 2X = (N-1) (\vertlength) + k\).
        This means the total length of all horizontal edges in \(P\) can be at most \(k\).
        In particular, no horizontal edge with length higher than \(k\) can be contained in \(P\) which restricts \(P\) to be entirely in \(C_{j,j''}\) for some \(j'' \in [n]\) by the length choice for column inconsistency blocking edges.
        Within the interior of any \(\rselect\)-\(t'_j\)-path in \(C_{j,j''}\) there must be at least \(k\) horizontal edges by the construction of \(G\): edges from~$\rselect$ to~$G^0$ arrive in~$G^0$ at offset~$+1$, while edges from~$G^0$ to~$t'_j$ exit~$G^0$ from offset~$+(k+1)$. Each of these at least~$k$ horizontal edges has length at least \(1\) (because all edges in our construction have length at least \(1\)), which means that the length~$(N-1) (\vertlength) + k$ of~$P$ can only be attained by a path~$P$ that has exactly \(k\) horizontal and exactly \(N-1\) vertical edges.
    \end{claimproof}

    Similarly, for each \(i \in [k]\), the path from \(\rcons\) to \(p(t_i)\) in \(H_1\) contains \(t_i\).
    The implied paths from \(\rcons\) to \(t_i\) in \(H_1\) can also be shown to have a specific structure.
    \begin{claim}
        \label{claim:planarhardness-consistencypaths}
        For each \(i \in [k]\), the interior of the \(\rcons\)-\(t_i\)-path in \(H_1\) is one row of \(R_i\).
    \end{claim}
    \begin{claimproof}
        We have argued in \cref{claim:treeroots} that \(r\rcons\) is necessarily included in \(H_1\).
        Further, the \(\rcons\)-\(t_i\)-path in \(H_1\) necessarily includes an edge of length~$X$ from \(t_i\) to~$V(G^0)$ and an edge of length~$X$ from \(\rcons\) to~\(V(G^0)\). 
        This means that any further length-\(X\) edge that could potentially be used to leave and reenter \(G^0\),
        is too long to be included in each \(r\)-\(p(t_i)\)-path in \(H_1\).
        Hence the interior \(P\) of each \(\rcons\)-\(t_i\)-path in \(H_1\) has to be contained in \(G^0\).
        
        Let us calculate an upper bound on the length of \(P\) implied by \(\lambda\).
        In particular, its length can be at most \(\lambda - \ell(p(t_j)t_j) - \ell(r\rcons) - 2X = (M-1) (\blocklength)\).
        
       \(P\) must contain \(kn-1\) column inconsistency blocking edges which means that the other edges of \(P\) have cumulative length at most \((M-1 - (kn-1)) (\blocklength) = k^2n(\blocklength)\).
       This is strictly smaller than \(\vertlength\) which means that \(P\) cannot contain a vertical edge and hence must be a row of \(G^0\).
    \end{claimproof}

    Our choice of \(\beta\) and \(\dig\) prohibits horizontal edges that are not in \(H_1\) to be dug for \(H_2\).
         \begin{claim}    
        \label{claim:planarhardness-selectionhedges}
        Every horizontal edge in \(H_2\) is also in \(H_1\).
    \end{claim}
    \begin{claimproof}
        We consider the contribution of~$\ell(H_1)$ and~$\ell(H_2)$ to the cost of solution~$\{H_1, H_2\}$, using the structure of~$H_1, H_2$ inferred in the previous claims. 

        By \Cref{claim:planarhardness-consistencypaths}, for each row consistency terminal~$j \in [k]$ the path from~$\rcons$ to~$p(t_j)$ in~$H_1$ contains an entire row of~$G^0$. As~$t_j$ is only adjacent to the rows of~$R_j$, these paths and the length-$X$ edges connecting them to~$\rcons$ and~$t_j$ are distinct for all~$j \in [k]$. Let us give a lower bound on the total length of any row of~$G^0$, as follows. There are~$M-1$ edges in a row. Of these edges, all but~$kn$ are \emph{blocking} edges of length~$\blocklength$ and the remaining edges have length at least~$1$. Hence the length of any row is at least~$(M-1-kn)(\blocklength)+kn$. We use this information to lower-bound~$\ell(H_1) + \ell(H_2)$. 
        
        By \Cref{claim:treeroots}, \Cref{claim:planarhardness-selectionpaths}, \Cref{claim:planarhardness-consistencypaths} and the definition of \(\ell\):
        \begin{itemize}
            \item the total length~$\ell(H_1)$ of edges in~$H_1$ is at least \(\lambda + k(2X+\lambda-X^*-2X - (M-1)(\blocklength) + (M - 1-nk)(\blocklength) + nk) = \lambda + k(\lambda - X^*-nk (\blocklength) + nk)=\lambda+k(\lambda-X^*-nk^2)\).
            \item the total length~$\ell(H_2)$ of edges in~$H_2$ is at least \(\lambda + k(\lambda - X^*)\).
        \end{itemize}
        Further, \(\dig(H_1)\) is already \(k \cdot (M - 1) (\hdigcost)\).
        Compared to the total cost threshold~$\beta$, this leaves at most \(nk^3\) of slack for \(\dig(E(H_2) \setminus E(H_1))\). This is smaller than the digging cost of any horizontal edge. Hence every horizontal edge in~$H_2$ is also in~$H_1$.
    \end{claimproof}

    We can use the above analysis of arbitrary solutions for \(\mathcal{I}\) to show the reverse implication to \Cref{claim:planarhardness-forward}.
    \begin{claim}
    If \(\mathcal{I}\) is a yes-instance then \(\left(S_{i,j}\right)_{(i,j) \in [k]^2}\) is a yes-instance.
    \end{claim}
    \begin{claimproof}
    Let \(\{H_1,H_2\}\) be a solution set for \(\mathcal{I}\).
    We show that selecting for each \((i,j) \in [k]^2\), the pair \((a,b)\) such that the \(i\)-th horizontal edge of the \(\rselect\)-\(t'_j\)-path in \(H_2\) is in the \(a\)-th row of \(R_i\) and in \(C_{j,b}\), yields a solution to the \GT instance from which we started the construction.
        
        By \Cref{claim:planarhardness-selectionpaths}, the horizontal edges in \(\rselect\)-\(t'_j\)-paths in \(H_2\) have length \(1\).
        Fix \(j \in [k]\) and let \(P\) be the \(\rselect\)-\(t'_j\)-path in \(H_2\).
        Because of which horizontal edges receive length \(1\), all horizontal edges of \(P\) are contained in pairwise distinct \(R_i\) meaning that we indeed select a pair for each \((i,j) \in [k]^2\).
        Also by the definition of \(\ell\), the \(i\)-th horizontal edge of \(P\) (i.e.\ the one between the \(i\)-th and \((i+1)\)-th column of \(C_{j,b}\)) having length one implies that it is in the \(b\)-th row in \(R_i\), and \((a,b) \in S_{i,j}\).
        
        Moreover, as by \Cref{claim:planarhardness-selectionpaths}, \(P\) is entirely contained in \(C_{j,b}\), this choice implies that all pairs selected for cells in column \(j\) have the same second entry.
        Lastly, horizontal edges of \(H_2\) have to be contained in \(H_1\) by \Cref{claim:planarhardness-selectionhedges}.
        By \Cref{claim:planarhardness-consistencypaths}, there is precisely one row of each \(R_i\) contained in \(H_1\), implying that all selected pairs in the same row share their first entry.
    \end{claimproof}

    Overall, we have shown that a \GT-solution for \((S_{i,j})_{(i,j) \in [k]^2}\) implies a \BLVshort-solution for \(\mathcal{I}\).
    This concludes the proof of correctness of our reduction.
    Notice that the number of terminals in the \BLVshort-instance is linear in \(k\) from the \GT instance, that \(t=2\) is constant, and that~$|V(G)|, \lambda \in n^{\Oh(1)}$. 
    Hence if there were \(f,f'\) such that \BLVshort can be solved in time \(f(|T|,t)\cdot (|V(G)| \cdot \lambda^{f'(t)})^{o(|T|)}\) for any instance \(\mathcal{I}\) then we could solve the original \GT instance in time \(\tilde{f}(k) \cdot n^{o(k)}\), contradicting the ETH.
\end{proof}

%% file: figs/planarhardnessconstr.tex
\begin{tikzpicture}[scale=0.85, every node/.style={scale=0.85}]
 \tikzset{terminal/.style={fill=gray!30, draw=gray!50,  thick,minimum size=0pt, inner sep=2pt}}
        \draw[blue,draw opacity=.5,,pattern={Lines[angle=30,yshift=-5pt],opacity=.2},pattern color = blue!10] (0,0) rectangle ++(3.75,2.5);
		\fill[blue,draw opacity=.5,pattern={Lines[angle=30,yshift=0],opacity=.2},pattern color = blue!10] (3.75,0) rectangle ++(3.75,2.5);
		\fill[blue,draw opacity=.5,pattern={Lines[angle=30,yshift=0],opacity=.2},pattern color = blue!10] (9,0) rectangle (12.75,2.5);
		\node at (.375,1.25) {\textcolor{blue}{\(C_{1,1}\)}};
		\draw[blue,draw opacity=.3] (0,0) rectangle ++(0.75,2.5);
		\node at (1.125,1.25) {\textcolor{blue}{\(C_{1,2}\)}};
		\draw[blue,draw opacity=.3] (1.5,0) -- (1.5,2.5);
		\node at (2.25,1.25) {\textcolor{blue}{\dots}};
		\draw[blue,draw opacity=.3] (3,0) -- (3,2.5);
		\node at (3.375,1.25) {\textcolor{blue}{\(C_{1,n}\)}};
		
		%\draw[blue,draw opacity=.3] (1.5,0) -- (1.5,2.5);
		%\draw[blue,draw opacity=.3] (1.75,0) -- (1.75,2.5);
		%\draw[blue,draw opacity=.3] (2.25,0) -- (2.25,2.5);
		\node at (4.125,1.25) {\textcolor{blue}{\(C_{2,1}\)}};
		\draw[blue,draw opacity=.3] (4.5,0) -- (4.5,2.5);
		\node at (5.625,1.25) {\textcolor{blue}{\dots}};
		\node at (7.125,1.25) {\textcolor{blue}{\(C_{2,n}\)}};
		\draw[blue,draw opacity=.3] (6.75,0) -- (6.75,2.5);
		\draw[blue,draw opacity=.5] (7.5,0) -- (7.5,2.5);
		
		\draw [gray,decorate,decoration={brace,amplitude=5pt}] (9,0) -- (9,2.5) node[midway,xshift=-1.75em]{\(n \cdot k\)};
		\draw [gray,decorate,decoration={brace,amplitude=5pt}] (9,2.5) -- (9.75,2.5) node[midway,yshift=1.2em]{\(k + 1\)};
		
		%\node at (8.25,1.25) {\textcolor{blue}{\dots}};
		\draw[blue,draw opacity=.5] (9,0) -- (9,2.5);
		\node at (9.375,1.25) {\textcolor{blue}{\(C_{k,1}\)}};
		\draw[blue,draw opacity=.3] (9.75,0) -- (9.75,2.5);
		
		\node at (10.875,1.25) {\textcolor{blue}{\dots}};
		\draw[blue,draw opacity=.3] (12,0) -- (12,2.5);
		\node at (12.375,1.25) {\textcolor{blue}{\(C_{k,n}\)}};
		
		\draw[red,draw opacity=.5,,pattern={Lines[angle=120,yshift=-5pt],opacity=.2},pattern color = red!10] (0,0) rectangle ++(12.75,.5);
		\node at (6.375,0.25) {\textcolor{red}{\(R_{1}\)}};
		\fill[red,draw opacity=.5,,pattern={Lines[angle=120,yshift=0],opacity=.2},pattern color = red!10] (0,2) rectangle ++(12.75,.5);
		\draw[red,draw opacity=.5] (0,2) -- (12.75,2);
		\node at (6.375,1.3) {\textcolor{red}{\vdots}};
		\node at (6.375,2.25) {\textcolor{red}{\(R_{k}\)}};
		
		\node[vertex,label=left:\(r_{\text{cons}}\)] (rcons) at (-1,1.25) {};
		\node[vertex,label=below:\(r_{\text{select}}\)] (rselect) at (6.375,-1) {};
		\node[vertex,label=below left:\(r\)] (r) at (-1,-1) {};
		
		\node[terminal,above of = rcons] (prcons) {};
		\node[terminal,right of = rselect] (prselect) {};
		
		\draw (r) -- (rcons);
		\draw (r) -- (rselect);
		\draw (rcons) -- (prcons);
		\draw (rselect) -- (prselect);
		
		\fill[fill=gray,fill opacity=.2] (rcons) -- (0.1,0.05) -- (0.1,2.45) -- (rcons);
		\draw (rcons) -- (0.1,0.05);
		\draw (rcons) --(0.1,2.45);
		
		\node[vertex,label=right:\(t_1\)] (t1) at (13.75,0.25) {};
		\fill[fill=gray,fill opacity=.2] (t1) -- (12.65,0.05) -- (12.65,.45) -- (t1);
		\draw (t1) -- (12.65,0.05);
		\draw (t1) -- (12.65,0.45);
		\node[vertex,label=right:\(t_k\)] (tk) at (13.75,2.25) {};
		\fill[fill=gray,fill opacity=.2] (tk) -- (12.65,2.05) -- (12.65,2.45) -- (tk);
		\draw (tk) -- (12.65,2.05);
		\draw (tk) -- (12.65,2.45);
		\node[orange!60!black,vertex,minimum size=0.8pt,inner sep=0.8pt,label=above right:\textcolor{orange!60!black}{\(v[N,M]\)}] at (12.65,2.45) {};
		
		\node[terminal,below of = t1] (pt1) {};
		\node[terminal,below of = tk] (ptk) {};
		
		\draw (t1) --(pt1);
		\draw (tk) --(ptk);
		
		\draw (rselect) -- (0.1,0.05);
		\draw (rselect) -- (0.85,0.05);
		\draw (rselect) -- (3.1,0.05);
		\draw (rselect) -- (3.85,0.05);
		\draw (rselect) -- (6.85,0.05);
		\draw (rselect) -- (9.1,0.05);
		\draw (rselect) -- (12.1,0.05);
		
		\node[green!60!black,vertex,minimum size=0.8pt,inner sep=0.8pt,label={[fill=white,fill opacity=.8]above:\textcolor{green!60!black}{\(v[1,n (k + 1) + 1]\)}}] at (3.85,0.05) {};
		
		\node[vertex,label=above:\(t'_1\)] (t'1) at (2.25,3.5) {};
		\draw (t'1) -- (0.6,2.45);
		\draw (t'1) -- (1.4,2.45);
		\draw (t'1) -- (3.65,2.45);
		\node[vertex,label=above:\(t'_2\)] (t'2) at (5.625,3.5) {};
		\draw (t'2) -- (4.4,2.45);
		\draw (t'2) -- (7.4,2.45);
		
		\node[yellow!60!black,vertex,minimum size=0.8pt,inner sep=0.8pt,label={[fill=white,fill opacity =.8]below:\textcolor{yellow!60!black}{\(v[N,(n + 1) (k + 1)]\)}}] at (4.4,2.45) {};
		
		\node[vertex,label=above:\(t'_k\)] (t'k) at (10.875,3.5) {};
		\draw (t'k) -- (9.65,2.45);
		\draw (t'k) -- (12.65,2.45);
		
		\node[terminal,right of = t'1] (pt'1) {};
		\node[terminal,right of = t'2] (pt'2) {};
		\node[terminal,right of = t'k] (pt'k) {};
		\draw (t'1) --(pt'1);
		\draw (t'2) --(pt'2);
		\draw (t'k) --(pt'k);
        
		\draw (0,0) rectangle ++(12.75,2.5);
	\end{tikzpicture}
	

%% file: figs/BLV_Instance.tex
\begin{tikzpicture}[scale=0.7]
\tikzstyle{vertex}=[circle, fill=black, inner sep=1.5pt, scale=0.85]
\tikzset{terminal/.style={fill=gray!30, draw=gray!100,  thick,minimum size=0pt, inner sep=2pt}}

% Parameters
\def\k{2}
\def\n{3}
\pgfmathsetmacro{\rowsVal}{int(\k*\n)}
\pgfmathsetmacro{\colsVal}{int(\k*(\k+1)*\n)}
\def\rows{\rowsVal}
\def\cols{\colsVal}
\def\xspacing{0.7}
\def\yspacing{0.6}

% Bordering
% rows
\pgfmathsetmacro{\hrecwidth}{(\cols-1)*\xspacing+0.5}
\pgfmathsetmacro{\hrecheight}{(\n-1)*\yspacing+0.25+\yspacing/2}
\pgfmathsetmacro{\startYCord}{-6*\yspacing-0.25}
\pgfmathsetmacro{\startXCord}{\xspacing-0.25}
\pgfmathsetmacro{\endYCord}{\startYCord+\hrecheight}
\pgfmathsetmacro{\endXCord}{\startXCord+\hrecwidth}
\pgfmathsetmacro{\SecondYCord}{\endYCord+\hrecheight}
\draw[red,draw opacity=.2] (\startXCord,\startYCord) rectangle (\endXCord,\endYCord);
\draw[red,draw opacity=.2,,pattern={Lines[angle=120,yshift=-5pt],opacity=.2},pattern color = violet!15] (\startXCord,\endYCord) rectangle (\endXCord,\SecondYCord);

% columns
\pgfmathsetmacro{\vrecwidth}{(\n-1)*\xspacing+\xspacing}
\pgfmathsetmacro{\vrecheight}{(\rows-1)*\yspacing+0.5}
\pgfmathsetmacro{\endVXCord}{\startXCord+\vrecwidth-\xspacing/2+0.25}
\pgfmathsetmacro{\endVYCord}{\startYCord+\vrecheight}
\pgfmathsetmacro{\secondVXCord}{\endVXCord+\vrecwidth}
\pgfmathsetmacro{\thirdVXCord}{\secondVXCord+\vrecwidth}
\pgfmathsetmacro{\fourthVXCord}{\thirdVXCord+\vrecwidth}
\pgfmathsetmacro{\fifthVXCord}{\fourthVXCord+\vrecwidth}
\pgfmathsetmacro{\sixthVXCord}{\fifthVXCord+\vrecwidth-\xspacing/2+0.25}
\draw[blue,draw opacity=.2,,pattern={Lines[angle=30,yshift=-5pt],opacity=.25},pattern color = red!15] (\startXCord,\startYCord) rectangle (\endVXCord,\endVYCord);
\draw[blue,draw opacity=.2,,pattern={Lines[angle=30,yshift=-5pt],opacity=.25},pattern color = red!15] (\endVXCord,\startYCord) rectangle (\secondVXCord,\endVYCord);
\draw[blue,draw opacity=.2,,pattern={Lines[angle=30,yshift=-5pt],opacity=.25},pattern color = red!15] (\secondVXCord,\startYCord) rectangle (\thirdVXCord,\endVYCord);
\draw[blue,draw opacity=.2,,pattern={Lines[angle=30,yshift=-5pt],opacity=.25},pattern color = blue!15] (\thirdVXCord,\startYCord) rectangle (\fourthVXCord,\endVYCord);
\draw[blue,draw opacity=.2,,pattern={Lines[angle=30,yshift=-5pt],opacity=.25},pattern color = blue!15] (\fourthVXCord,\startYCord) rectangle (\fifthVXCord,\endVYCord);
\draw[blue,draw opacity=.2,,pattern={Lines[angle=30,yshift=-5pt],opacity=.25},pattern color = blue!15] (\fifthVXCord,\startYCord) rectangle (\sixthVXCord,\endVYCord);

% ----------------- grid -----------------
\foreach \i in {1,...,\rows}{
    \foreach \j in {1,...,\cols}{
        \pgfmathsetmacro{\x}{\j*\xspacing}
        \pgfmathsetmacro{\y}{-\i*\yspacing}
        \node[vertex] (T\i\j) at (\x,\y) {};
    }
}

% rcons
\pgfmathsetmacro{\rconsx}{-\xspacing/2}
\pgfmathsetmacro{\rconsy}{-(\rows/2)*\yspacing-\yspacing/2}
\pgfmathsetmacro{\rconstermy}{\rconsy+2*\yspacing}
\node[vertex,label=left:\(r_{\text{cons}}\)] (rcons) at (\rconsx,\rconsy) {};
\node[terminal] (rconsterm) at (\rconsx,\rconstermy) {};
\draw[orange] (rcons) -- (rconsterm);

% rselect
\pgfmathsetmacro{\rselectx}{(\cols/2)*\xspacing+\xspacing/2}
\pgfmathsetmacro{\rselecttermx}{\rselectx+2*\xspacing}
\pgfmathsetmacro{\rselecty}{-\rows*\yspacing-1.5*\yspacing}
\node[vertex,label=below:\(r_{\text{select}}\)] (rselect) at (\rselectx,\rselecty) {};
\node[terminal] (rselectterm) at (\rselecttermx,\rselecty) {};
\draw[orange] (rselect) -- (rselectterm);

% r
\node[vertex,label=below left:\(r\)] (r) at (\rconsx,\rselecty) {};
\draw[green!50!black] (r) -- (rselect);
\draw[green!50!black] (r) -- (rcons);

% right terminals
\foreach \i in {1,...,\rows}{
    \pgfmathsetmacro{\r}{int(mod(\rows-\i,\n))}
    \pgfmathsetmacro{\index}{int((\rows-\i)/\n+1)}
    \ifnum\r=0
        \pgfmathsetmacro{\ptermy}{-(\rows-\i+((\n-1)/2)+1)*\yspacing}
        \pgfmathsetmacro{\termy}{\ptermy-\yspacing}
        \pgfmathsetmacro{\termx}{(\cols+1)*\xspacing}
        \node[vertex,label=right:\(t_{\index}\)] (prter\index) at (\termx,\ptermy) {};
        \node[terminal] (rter\index) at (\termx,\termy) {};
        \draw[blue] (prter\index) -- (rter\index);
        \foreach \addi in {1,...,\n}{
            \pgfmathsetmacro{\vindex}{int(\rows-\i+\addi)}
            \ifnum\vindex=3
                \draw[cyan] (prter\index)-- (T\vindex\cols);
            \else
                \ifnum\vindex=5
                    \draw[cyan] (prter\index)-- (T\vindex\cols);
                \else
                    \draw[cyan] (prter\index)-- (T\vindex\cols);
                \fi
            \fi
        }
    \fi
}

%  top terminals
\foreach \j in {1,...,\cols}{
    \pgfmathsetmacro{\r}{int(mod(\j,(\k+1)*\n))}
    \pgfmathsetmacro{\index}{int(\j/((\k+1)*\n)+1)}
    \ifnum\r=1
        \pgfmathsetmacro{\ptermx}{(\j+((((\k+1)*\n)-1)/2)+1)*\xspacing}
        \pgfmathsetmacro{\termx}{\ptermx+2*\xspacing}
        \pgfmathsetmacro{\termy}{0}
        \node[vertex,label=above:\(t'_{\index}\)] (ptter\index) at (\ptermx,\termy) {};
        \node[terminal] (tter\index) at (\termx,\termy) {};
        \draw[brown] (ptter\index) -- (tter\index);
        \foreach \addi in {1,...,\n}{
            \pgfmathsetmacro{\vindex}{int((\index-1)*((\k+1)*\n)+\addi*(\k+1))}
            \ifnum\vindex=3
                \draw[cyan] (ptter\index)-- (T1\vindex);
            \else
                \ifnum\vindex=18
                    \draw[cyan] (ptter\index)-- (T1\vindex);
                \else
                    \draw[cyan] (ptter\index)-- (T1\vindex);
                \fi
            \fi
        }
    \fi
}

% horizontal edges
\pgfmathsetmacro{\colsLess}{int(\cols-1)}
\foreach \i in {1,...,\rows}{
    \foreach \j in {1,...,\colsLess}{
        \pgfmathtruncatemacro{\jn}{\j+1}
        % Check if this edge should be red!60!white
        \ifnum\i=6
            \pgfmathsetmacro{\r}{int(mod(\j,3))}
            \ifnum\r=0
                \draw[gray!100!black, dotted, thick] (T\i\j) -- (T\i\jn);
            \else
                    \draw[gray!60!black,dashed, thick] (T\i\j) -- (T\i\jn);
            \fi
        \else
            \ifnum\i=2
                \ifnum\j=14
                            \draw[violet!100!white, very thick] (T\i\j) -- (T\i\jn);
                \else
                    \pgfmathsetmacro{\r}{int(mod(\j,3))}
                    \ifnum\r=0
                        \draw[gray!100!black, dotted, thick] (T\i\j) -- (T\i\jn);
                    \else
                            \draw[gray!60!black,dashed, thick] (T\i\j) -- (T\i\jn);
                    \fi
                \fi
            \else
                \ifnum\i=1
                    \pgfmathsetmacro{\r}{int(mod(\j,3))}
                    \ifnum\r=0
                        \draw[gray!100!black, dotted, thick] (T\i\j) -- (T\i\jn);
                    \else
                            \draw[gray!60!black,dashed, thick] (T\i\j) -- (T\i\jn);
                    \fi
                \else
                    \ifnum\i=5
                        \ifnum\j=1
                            \draw[violet!100!white, very thick] (T\i\j) -- (T\i\jn);
                        \else
                            \ifnum\j=16
                                \draw[violet!100!white, very thick] (T\i\j) -- (T\i\jn);
                            \else
                                \pgfmathsetmacro{\r}{int(mod(\j,3))}
                                \ifnum\r=0
                                    \draw[gray!100!black, dotted, thick] (T\i\j) -- (T\i\jn);
                                \else
                                    \draw[gray!100!black,dashed, thick] (T\i\j) -- (T\i\jn);
                                \fi
                            \fi
                        \fi
                    \else   
                        \ifnum\i=4
                            \ifnum\j=4
                                \draw[violet!100!white, very thick] (T\i\j) -- (T\i\jn);
                            \else
                                \ifnum\j=10
                                    \draw[violet!100!white, very thick] (T\i\j) -- (T\i\jn);
                                \else
                                    \pgfmathsetmacro{\r}{int(mod(\j,3))}
                                    \ifnum\r=0
                                        \draw[gray!100!black, dotted, thick] (T\i\j) -- (T\i\jn);
                                    \else
                                        \draw[gray!60!black,dashed, thick] (T\i\j) -- (T\i\jn);
                                    \fi
                                \fi
                            \fi
                        \else
                            \ifnum\i=3
                                \ifnum\j=2
                                    \draw[violet!100!white, very thick] (T\i\j) -- (T\i\jn);
                                \else
                                    \ifnum\j=17
                                        \draw[violet!100!white, very thick] (T\i\j) -- (T\i\jn);
                                    \else
                                        \pgfmathsetmacro{\r}{int(mod(\j,3))}
                                        \ifnum\r=0
                                            \draw[gray!100!black, dotted, thick] (T\i\j) -- (T\i\jn);
                                        \else
                                            \draw[gray!60!black,dashed, thick] (T\i\j) -- (T\i\jn);
                                        \fi
                                    \fi
                                \fi
                            \else
                                \pgfmathsetmacro{\r}{int(mod(\j,3))}
                                \ifnum\r=0
                                    \draw[gray!100!black, dotted, thick] (T\i\j) -- (T\i\jn);
                                \else
                                    \draw[gray!100!black] (T\i\j) -- (T\i\jn);
                                \fi
                            \fi
                        \fi
                    \fi
                \fi
            \fi
        \fi
    }
}

% Vertical edges (all black)
\pgfmathsetmacro{\rowsLess}{int(\rows-1)}
\foreach \i in {1,...,\rowsLess}{
    \foreach \j in {1,...,\cols}{
        \pgfmathtruncatemacro{\inext}{int(\i+1)}
        \ifnum\j=1
            \ifnum\i=5
                \draw[gray!60!black] (T\i\j) -- (T\inext\j);
            \else
                \draw[gray!60!black] (T\i\j) -- (T\inext\j);
            \fi
        \else
            \ifnum\j=2
                \ifnum\i>2
                    \ifnum\i=5
                        \draw[gray!60!black] (T\i\j) -- (T\inext\j);
                    \else
                        \draw[gray!60!black] (T\i\j) -- (T\inext\j);
                    \fi
                \else
                    \draw[gray!60!black] (T\i\j) -- (T\inext\j);
                \fi
            \else
                \ifnum\j=3
                    \ifnum\i<3
                        \draw[gray!60!black] (T\i\j) -- (T\inext\j);
                    \else
                        \draw[gray!60!black] (T\i\j) -- (T\inext\j);
                    \fi
                \else
                    \ifnum\j=16
                        \ifnum\i=5
                            \draw[gray!60!black] (T\i\j) -- (T\inext\j);
                        \else
                            \draw[gray!60!black] (T\i\j) -- (T\inext\j);
                        \fi
                    \else
                        \ifnum\j=17
                            \ifnum\i>2
                                \ifnum\i=5
                                    \draw[gray!60!black] (T\i\j) -- (T\inext\j);
                                \else
                                    \draw[gray!60!black] (T\i\j) -- (T\inext\j);
                                \fi
                            \else
                                \draw[gray!60!black] (T\i\j) -- (T\inext\j);
                            \fi
                        \else
                            \ifnum\j=18
                                \ifnum\i<3
                                        \draw[gray!60!black] (T\i\j) -- (T\inext\j);
                                \else
                                    \draw[gray!60!black] (T\i\j) -- (T\inext\j);
                                \fi
    
                            \else
                                \draw[gray!60!black] (T\i\j) -- (T\inext\j);
                            \fi
                        \fi
                    \fi
                \fi
            \fi
        \fi
    }
}

% connection to rcons
\foreach \i in {1,...,\rows}{
    \ifnum\i=3
        \draw[cyan] (T\i1) -- (rcons);
    \else
        \ifnum\i=6
            \draw[cyan] (T\i1) -- (rcons);
        \else
            \draw[cyan] (T\i1) -- (rcons);
        \fi
    \fi
}

% connection to rselect
\foreach \j in {1,...,\cols}{
    \pgfmathtruncatemacro{\r}{int(mod(\j,\k+1))}
    \ifnum\r=1
        \ifnum\j=1
            \draw[cyan] (rselect) -- (T\rows\j);
        \else
            \ifnum\j=16
                \draw[cyan] (rselect) -- (T\rows\j);
            \else
                \draw[cyan] (rselect) -- (T\rows\j);
            \fi
        \fi
    \fi
}

\end{tikzpicture}

%% file: figs/BLVSolutionInstance.tex
\begin{tikzpicture}[scale=0.7]
\tikzstyle{vertex}=[circle, fill=black, inner sep=1.5pt, scale=0.85]
\tikzset{terminal/.style={fill=gray!30, draw=gray!50,  thick,minimum size=0pt, inner sep=2pt}}

% Parameters
\def\k{2}
\def\n{3}
\pgfmathsetmacro{\rowsVal}{int(\k*\n)}
\pgfmathsetmacro{\colsVal}{int(\k*(\k+1)*\n)}
\def\rows{\rowsVal}
\def\cols{\colsVal}
\def\xspacing{0.7}
\def\yspacing{0.6}

% Bordering
% rows
\pgfmathsetmacro{\hrecwidth}{(\cols-1)*\xspacing+0.5}
\pgfmathsetmacro{\hrecheight}{(\n-1)*\yspacing+0.25+\yspacing/2}
\pgfmathsetmacro{\startYCord}{-6*\yspacing-0.25}
\pgfmathsetmacro{\startXCord}{\xspacing-0.25}
\pgfmathsetmacro{\endYCord}{\startYCord+\hrecheight}
\pgfmathsetmacro{\endXCord}{\startXCord+\hrecwidth}
\pgfmathsetmacro{\SecondYCord}{\endYCord+\hrecheight}
\draw[red,draw opacity=.2] (\startXCord,\startYCord) rectangle (\endXCord,\endYCord);
\draw[red,draw opacity=.2,,pattern={Lines[angle=120,yshift=-5pt],opacity=.2},pattern color = violet!15] (\startXCord,\endYCord) rectangle (\endXCord,\SecondYCord);

% columns
\pgfmathsetmacro{\vrecwidth}{(\n-1)*\xspacing+\xspacing}
\pgfmathsetmacro{\vrecheight}{(\rows-1)*\yspacing+0.5}
\pgfmathsetmacro{\endVXCord}{\startXCord+\vrecwidth-\xspacing/2+0.25}
\pgfmathsetmacro{\endVYCord}{\startYCord+\vrecheight}
\pgfmathsetmacro{\secondVXCord}{\endVXCord+\vrecwidth}
\pgfmathsetmacro{\thirdVXCord}{\secondVXCord+\vrecwidth}
\pgfmathsetmacro{\fourthVXCord}{\thirdVXCord+\vrecwidth}
\pgfmathsetmacro{\fifthVXCord}{\fourthVXCord+\vrecwidth}
\pgfmathsetmacro{\sixthVXCord}{\fifthVXCord+\vrecwidth-\xspacing/2+0.25}
\draw[blue,draw opacity=.2,,pattern={Lines[angle=30,yshift=-5pt],opacity=.25},pattern color = red!15] (\startXCord,\startYCord) rectangle (\endVXCord,\endVYCord);
\draw[blue,draw opacity=.2,,pattern={Lines[angle=30,yshift=-5pt],opacity=.25},pattern color = red!15] (\endVXCord,\startYCord) rectangle (\secondVXCord,\endVYCord);
\draw[blue,draw opacity=.2,,pattern={Lines[angle=30,yshift=-5pt],opacity=.25},pattern color = red!15] (\secondVXCord,\startYCord) rectangle (\thirdVXCord,\endVYCord);
\draw[blue,draw opacity=.2,,pattern={Lines[angle=30,yshift=-5pt],opacity=.25},pattern color = blue!15] (\thirdVXCord,\startYCord) rectangle (\fourthVXCord,\endVYCord);
\draw[blue,draw opacity=.2,,pattern={Lines[angle=30,yshift=-5pt],opacity=.25},pattern color = blue!15] (\fourthVXCord,\startYCord) rectangle (\fifthVXCord,\endVYCord);
\draw[blue,draw opacity=.2,,pattern={Lines[angle=30,yshift=-5pt],opacity=.25},pattern color = blue!15] (\fifthVXCord,\startYCord) rectangle (\sixthVXCord,\endVYCord);

% ----------------- grid -----------------
\foreach \i in {1,...,\rows}{
    \foreach \j in {1,...,\cols}{
        \pgfmathsetmacro{\x}{\j*\xspacing}
        \pgfmathsetmacro{\y}{-\i*\yspacing}
        \node[vertex] (T\i\j) at (\x,\y) {};
    }
}

% rcons
\pgfmathsetmacro{\rconsx}{-\xspacing/2}
\pgfmathsetmacro{\rconsy}{-(\rows/2)*\yspacing-\yspacing/2}
\pgfmathsetmacro{\rconstermy}{\rconsy+2*\yspacing}
\node[vertex,label=left:\(r_{\text{cons}}\)] (rcons) at (\rconsx,\rconsy) {};
\node[terminal] (rconsterm) at (\rconsx,\rconstermy) {};
\draw[red!60!white,  very thick] (rcons) -- (rconsterm);

% rselect
\pgfmathsetmacro{\rselectx}{(\cols/2)*\xspacing+\xspacing/2}
\pgfmathsetmacro{\rselecttermx}{\rselectx+2*\xspacing}
\pgfmathsetmacro{\rselecty}{-\rows*\yspacing-1.5*\yspacing}
\node[vertex,label=below:\(r_{\text{select}}\)] (rselect) at (\rselectx,\rselecty) {};
\node[terminal] (rselectterm) at (\rselecttermx,\rselecty) {};
\draw[blue!60!white,  very thick] (rselect) -- (rselectterm);

% r
\node[vertex,label=below left:\(r\)] (r) at (\rconsx,\rselecty) {};
\draw[blue!60!white,  very thick] (r) -- (rselect);
\draw[red!60!white,  very thick] (r) -- (rcons);

% right terminals
\foreach \i in {1,...,\rows}{
    \pgfmathsetmacro{\r}{int(mod(\rows-\i,\n))}
    \pgfmathsetmacro{\index}{int((\rows-\i)/\n+1)}
    \ifnum\r=0
        \pgfmathsetmacro{\ptermy}{-(\rows-\i+((\n-1)/2)+1)*\yspacing}
        \pgfmathsetmacro{\termy}{\ptermy-\yspacing}
        \pgfmathsetmacro{\termx}{(\cols+1)*\xspacing}
        \node[vertex,label=right:\(t_{\index}\)] (prter\index) at (\termx,\ptermy) {};
        \node[terminal] (rter\index) at (\termx,\termy) {};
        \draw[red!60!white,  very thick] (prter\index) -- (rter\index);
        \foreach \addi in {1,...,\n}{
            \pgfmathsetmacro{\vindex}{int(\rows-\i+\addi)}
            \ifnum\vindex=3
                \draw[red!60!white,  very thick] (prter\index)-- (T\vindex\cols);
            \else
                \ifnum\vindex=5
                    \draw[red!60!white,  very thick] (prter\index)-- (T\vindex\cols);
                \else
                    \draw[gray] (prter\index)-- (T\vindex\cols);
                \fi
            \fi
        }
    \fi
}

%  top terminals
\foreach \j in {1,...,\cols}{
    \pgfmathsetmacro{\r}{int(mod(\j,(\k+1)*\n))}
    \pgfmathsetmacro{\index}{int(\j/((\k+1)*\n)+1)}
    \ifnum\r=1
        \pgfmathsetmacro{\ptermx}{(\j+((((\k+1)*\n)-1)/2)+1)*\xspacing}
        \pgfmathsetmacro{\termx}{\ptermx+2*\xspacing}
        \pgfmathsetmacro{\termy}{0}
        \node[vertex,label=above:\(t'_{\index}\)] (ptter\index) at (\ptermx,\termy) {};
        \node[terminal] (tter\index) at (\termx,\termy) {};
        \draw[blue!60!white,  very thick] (ptter\index) -- (tter\index);
        \foreach \addi in {1,...,\n}{
            \pgfmathsetmacro{\vindex}{int((\index-1)*((\k+1)*\n)+\addi*(\k+1))}
            \ifnum\vindex=3
                \draw[blue!60!white,  very thick] (ptter\index)-- (T1\vindex);
            \else
                \ifnum\vindex=18
                    \draw[blue!60!white,  very thick] (ptter\index)-- (T1\vindex);
                \else
                    \draw[gray] (ptter\index)-- (T1\vindex);
                \fi
            \fi
        }
    \fi
}

% horizontal edges
\pgfmathsetmacro{\colsLess}{int(\cols-1)}
\foreach \i in {1,...,\rows}{
    \foreach \j in {1,...,\colsLess}{
        \pgfmathtruncatemacro{\jn}{\j+1}
        % Check if this edge should be red!60!white
        \ifnum\i=5
            \ifnum\j=1
                \draw[violet!80!white, very thick] (T\i\j) -- (T\i\jn);
            \else
                \ifnum\j=16
                    \draw[violet!80!white, very thick] (T\i\j) -- (T\i\jn);
                \else
                    \pgfmathsetmacro{\r}{int(mod(\j,3))}
                    \ifnum\r=0
                        \draw[red!60!white,  very thick, dotted] (T\i\j) -- (T\i\jn);
                    \else
                        \draw[red!60!white,  very thick] (T\i\j) -- (T\i\jn);
                    \fi
                \fi
            \fi
        \else
            \ifnum\i=3
                \ifnum\j=2
                    \draw[violet!80!white, very thick] (T\i\j) -- (T\i\jn);
                \else
                    \ifnum\j=17
                        \draw[violet!80!white, very thick] (T\i\j) -- (T\i\jn);
                    \else
                        \pgfmathsetmacro{\r}{int(mod(\j,3))}
                        \ifnum\r=0
                            \draw[red!60!white,  very thick, dotted] (T\i\j) -- (T\i\jn);
                        \else
                            \draw[red!60!white,  very thick] (T\i\j) -- (T\i\jn);
                        \fi
                    \fi
                \fi
            \else
                \pgfmathsetmacro{\r}{int(mod(\j,3))}
                \ifnum\r=0
                    \draw[gray!50!black, dotted, thick] (T\i\j) -- (T\i\jn);
                \else
                    \draw[gray!50!black] (T\i\j) -- (T\i\jn);
                \fi
            \fi
        \fi
    }
}

% Vertical edges (all black)
\pgfmathsetmacro{\rowsLess}{int(\rows-1)}
\foreach \i in {1,...,\rowsLess}{
    \foreach \j in {1,...,\cols}{
        \pgfmathtruncatemacro{\inext}{int(\i+1)}
        \ifnum\j=1
            \ifnum\i=5
                \draw[blue!60!white, very thick] (T\i\j) -- (T\inext\j);
            \else
                \draw[gray] (T\i\j) -- (T\inext\j);
            \fi
        \else
            \ifnum\j=2
                \ifnum\i>2
                    \ifnum\i=5
                        \draw[gray] (T\i\j) -- (T\inext\j);
                    \else
                        \draw[blue!60!white, very thick] (T\i\j) -- (T\inext\j);
                    \fi
                \else
                    \draw[gray] (T\i\j) -- (T\inext\j);
                \fi
            \else
                \ifnum\j=3
                    \ifnum\i<3
                        \draw[blue!60!white, very thick] (T\i\j) -- (T\inext\j);
                    \else
                        \draw[gray] (T\i\j) -- (T\inext\j);
                    \fi
                \else
                    \ifnum\j=16
                        \ifnum\i=5
                            \draw[blue!60!white, very thick] (T\i\j) -- (T\inext\j);
                        \else
                            \draw[gray] (T\i\j) -- (T\inext\j);
                        \fi
                    \else
                        \ifnum\j=17
                            \ifnum\i>2
                                \ifnum\i=5
                                    \draw[gray] (T\i\j) -- (T\inext\j);
                                \else
                                    \draw[blue!60!white,  very thick] (T\i\j) -- (T\inext\j);
                                \fi
                            \else
                                \draw[gray] (T\i\j) -- (T\inext\j);
                            \fi
                        \else
                            \ifnum\j=18
                                \ifnum\i<3
                                        \draw[blue!60!white, very thick] (T\i\j) -- (T\inext\j);
                                \else
                                    \draw[gray] (T\i\j) -- (T\inext\j);
                                \fi
    
                            \else
                                \draw[gray] (T\i\j) -- (T\inext\j);
                            \fi
                        \fi
                    \fi
                \fi
            \fi
        \fi
    }
}

% connection to rcons
\foreach \i in {1,...,\rows}{
    \ifnum\i=3
        \draw[red!60!white,  very thick] (T\i1) -- (rcons);
    \else
        \ifnum\i=5
            \draw[red!60!white,  very thick] (T\i1) -- (rcons);
        \else
            \draw[gray] (T\i1) -- (rcons);
        \fi
    \fi
}

% connection to rselect
\foreach \j in {1,...,\cols}{
    \pgfmathtruncatemacro{\r}{int(mod(\j,\k+1))}
    \ifnum\r=1
        \ifnum\j=1
            \draw[blue!60!white,  very thick] (rselect) -- (T\rows\j);
        \else
            \ifnum\j=16
                \draw[blue!60!white,  very thick] (rselect) -- (T\rows\j);
            \else
                \draw[gray] (rselect) -- (T\rows\j);
            \fi
        \fi
    \fi
}

\end{tikzpicture}

%% file: figs/ULV_Solution_In_BLV_Structure.tex
\begin{tikzpicture}[scale=0.7]
\tikzstyle{vertex}=[circle, fill=black, inner sep=1.5pt, scale=0.85]
\tikzset{terminal/.style={fill=gray!30, draw=gray!50,  thick,minimum size=0pt, inner sep=2pt}}

% Parameters
\def\k{2}
\def\n{3}
\pgfmathsetmacro{\rowsVal}{int(\k*\n)}
\pgfmathsetmacro{\colsVal}{int(\k*(\k+1)*\n)}
\def\rows{\rowsVal}
\def\cols{\colsVal}
\def\xspacing{0.7}
\def\yspacing{0.6}

% Bordering
% rows
\pgfmathsetmacro{\hrecwidth}{(\cols-1)*\xspacing+0.5}
\pgfmathsetmacro{\hrecheight}{(\n-1)*\yspacing+0.25+\yspacing/2}
\pgfmathsetmacro{\startYCord}{-6*\yspacing-0.25}
\pgfmathsetmacro{\startXCord}{\xspacing-0.25}
\pgfmathsetmacro{\endYCord}{\startYCord+\hrecheight}
\pgfmathsetmacro{\endXCord}{\startXCord+\hrecwidth}
\pgfmathsetmacro{\SecondYCord}{\endYCord+\hrecheight}
\draw[red,draw opacity=.2] (\startXCord,\startYCord) rectangle (\endXCord,\endYCord);
\draw[red,draw opacity=.2,,pattern={Lines[angle=120,yshift=-5pt],opacity=.2},pattern color = violet!15] (\startXCord,\endYCord) rectangle (\endXCord,\SecondYCord);

% columns
\pgfmathsetmacro{\vrecwidth}{(\n-1)*\xspacing+\xspacing}
\pgfmathsetmacro{\vrecheight}{(\rows-1)*\yspacing+0.5}
\pgfmathsetmacro{\endVXCord}{\startXCord+\vrecwidth-\xspacing/2+0.25}
\pgfmathsetmacro{\endVYCord}{\startYCord+\vrecheight}
\pgfmathsetmacro{\secondVXCord}{\endVXCord+\vrecwidth}
\pgfmathsetmacro{\thirdVXCord}{\secondVXCord+\vrecwidth}
\pgfmathsetmacro{\fourthVXCord}{\thirdVXCord+\vrecwidth}
\pgfmathsetmacro{\fifthVXCord}{\fourthVXCord+\vrecwidth}
\pgfmathsetmacro{\sixthVXCord}{\fifthVXCord+\vrecwidth-\xspacing/2+0.25}
\draw[blue,draw opacity=.2,,pattern={Lines[angle=30,yshift=-5pt],opacity=.25},pattern color = red!15] (\startXCord,\startYCord) rectangle (\endVXCord,\endVYCord);
\draw[blue,draw opacity=.2,,pattern={Lines[angle=30,yshift=-5pt],opacity=.25},pattern color = red!15] (\endVXCord,\startYCord) rectangle (\secondVXCord,\endVYCord);
\draw[blue,draw opacity=.2,,pattern={Lines[angle=30,yshift=-5pt],opacity=.25},pattern color = red!15] (\secondVXCord,\startYCord) rectangle (\thirdVXCord,\endVYCord);
\draw[blue,draw opacity=.2,,pattern={Lines[angle=30,yshift=-5pt],opacity=.25},pattern color = blue!15] (\thirdVXCord,\startYCord) rectangle (\fourthVXCord,\endVYCord);
\draw[blue,draw opacity=.2,,pattern={Lines[angle=30,yshift=-5pt],opacity=.25},pattern color = blue!15] (\fourthVXCord,\startYCord) rectangle (\fifthVXCord,\endVYCord);
\draw[blue,draw opacity=.2,,pattern={Lines[angle=30,yshift=-5pt],opacity=.25},pattern color = blue!15] (\fifthVXCord,\startYCord) rectangle (\sixthVXCord,\endVYCord);

% ----------------- grid -----------------
\foreach \i in {1,...,\rows}{
    \foreach \j in {1,...,\cols}{
        \pgfmathsetmacro{\x}{\j*\xspacing}
        \pgfmathsetmacro{\y}{-\i*\yspacing}
        \node[vertex] (T\i\j) at (\x,\y) {};
    }
}

% rcons
\pgfmathsetmacro{\rconsy}{-(\rows/2)*\yspacing-\yspacing/2}
\pgfmathsetmacro{\rconstermy}{\rconsy+2*\yspacing}
\node[vertex,label=left:\(r_{\text{cons}}\)] (rcons) at (0,\rconsy) {};
\node[terminal] (rconsterm) at (0,\rconstermy) {};
\draw[red!60!white,  very thick] (rcons) -- (rconsterm);

% rselect
\pgfmathsetmacro{\rselectx}{(\cols/2)*\xspacing+\xspacing/2}
\pgfmathsetmacro{\rselecttermx}{\rselectx+2*\xspacing}
\pgfmathsetmacro{\rselecty}{-\rows*\yspacing-1.5*\yspacing}
\node[vertex,label=below:\(r_{\text{select}}\)] (rselect) at (\rselectx,\rselecty) {};
\node[terminal] (rselectterm) at (\rselecttermx,\rselecty) {};
\draw[blue!60!white,  very thick] (rselect) -- (rselectterm);

% r
\node[vertex,label=below left:\(r\)] (r) at (0,\rselecty) {};
\draw[blue!60!white,  very thick] (r) -- (rselect);
\draw[red!60!white,  very thick] (r) -- (rcons);

% right terminals
\foreach \i in {1,...,\rows}{
    \pgfmathsetmacro{\r}{int(mod(\rows-\i,\n))}
    \pgfmathsetmacro{\index}{int((\rows-\i)/\n+1)}
    \ifnum\r=0
        \pgfmathsetmacro{\ptermy}{-(\rows-\i+((\n-1)/2)+1)*\yspacing}
        \pgfmathsetmacro{\termy}{\ptermy-\yspacing}
        \pgfmathsetmacro{\termx}{(\cols+1)*\xspacing}
        \node[vertex,label=right:\(t_{\index}\)] (prter\index) at (\termx,\ptermy) {};
        \node[terminal] (rter\index) at (\termx,\termy) {};
        \draw[red!60!white,  very thick] (prter\index) -- (rter\index);
        \foreach \addi in {1,...,\n}{
            \pgfmathsetmacro{\vindex}{int(\rows-\i+\addi)}
            \ifnum\vindex=3
                \draw[red!60!white,  very thick] (prter\index)-- (T\vindex\cols);
            \else
                \ifnum\vindex=5
                    \draw[red!60!white,  very thick] (prter\index)-- (T\vindex\cols);
                \else
                    \draw[gray] (prter\index)-- (T\vindex\cols);
                \fi
            \fi
        }
    \fi
}

%  top terminals
\foreach \j in {1,...,\cols}{
    \pgfmathsetmacro{\r}{int(mod(\j,(\k+1)*\n))}
    \pgfmathsetmacro{\index}{int(\j/((\k+1)*\n)+1)}
    \ifnum\r=1
        \pgfmathsetmacro{\ptermx}{(\j+((((\k+1)*\n)-1)/2)+1)*\xspacing}
        \pgfmathsetmacro{\termx}{\ptermx+2*\xspacing}
        \pgfmathsetmacro{\termy}{0}
        \node[vertex,label=above:\(t'_{\index}\)] (ptter\index) at (\ptermx,\termy) {};
        \node[terminal] (tter\index) at (\termx,\termy) {};
        \draw[blue!60!white,  very thick] (ptter\index) -- (tter\index);
        \foreach \addi in {1,...,\n}{
            \pgfmathsetmacro{\vindex}{int((\index-1)*((\k+1)*\n)+\addi*(\k+1))}
            \ifnum\vindex=3
                \draw[blue!60!white,  very thick] (ptter\index)-- (T1\vindex);
            \else
                \ifnum\vindex=18
                    \draw[blue!60!white,  very thick] (ptter\index)-- (T1\vindex);
                \else
                    \draw[gray] (ptter\index)-- (T1\vindex);
                \fi
            \fi
        }
    \fi
}

% horizontal edges
\pgfmathsetmacro{\colsLess}{int(\cols-1)}
\foreach \i in {1,...,\rows}{
    \foreach \j in {1,...,\colsLess}{
        \pgfmathtruncatemacro{\jn}{\j+1}
        % Check if this edge should be red!60!white
        \ifnum\i=1
            \pgfmathsetmacro{\r}{int(mod(\j,3))}
            \ifnum\r=0
                \draw[violet!80!white, dotted, very thick] (T\i\j) -- (T\i\jn);
            \else
                \draw[violet!80!white, very thick] (T\i\j) -- (T\i\jn);
            \fi
        \else
            \pgfmathsetmacro{\r}{int(mod(\j,3))}
            \ifnum\r=0
                \draw[gray!50!black, dotted, thick] (T\i\j) -- (T\i\jn);
            \else
                \draw[gray!50!black] (T\i\j) -- (T\i\jn);
            \fi
        \fi
    }
}

% Vertical edges (all black)
\pgfmathsetmacro{\rowsLess}{int(\rows-1)}
\foreach \i in {1,...,\rowsLess}{
    \foreach \j in {1,...,\cols}{
        \pgfmathtruncatemacro{\inext}{int(\i+1)}
        \ifnum\j=16
            \draw[blue!60!white, very thick] (T\i\j) -- (T\inext\j);
        \else
            \ifnum\j=18
                \ifnum\i=5
                    \draw[gray] (T\i\j) -- (T\inext\j);
                \else
                    \draw[red!60!white, very thick] (T\i\j) -- (T\inext\j);
                \fi
            \else
                \draw[gray] (T\i\j) -- (T\inext\j);
            \fi
        \fi
    }
}

% connection to rcons
\foreach \i in {1,...,\rows}{
    \ifnum\i=1
        \draw[red!60!white,  very thick] (T\i1) -- (rcons);
    \else
        \draw[gray] (T\i1) -- (rcons);
    \fi
}

% connection to rselect
\foreach \j in {1,...,\cols}{
    \pgfmathtruncatemacro{\r}{int(mod(\j,\k+1))}
    \ifnum\r=1
        \ifnum\j=16
            \draw[blue!60!white,  very thick] (rselect) -- (T\rows\j);
        \else
            \draw[gray] (rselect) -- (T\rows\j);
        \fi
    \fi
}

\end{tikzpicture}

%% file: figs/planarhardnesspartexample.tex
	\begin{tikzpicture}
		\draw[red,draw opacity=.5,,pattern={Lines[angle=120,yshift=-5pt],opacity=.2},pattern color = red!10] (-.25,-0.25) rectangle ++(5.5,1.75);
		\draw[red,draw opacity=.5,,pattern={Lines[angle=120,yshift=0pt],opacity=.2},pattern color = red!10] (-.25,1.5) rectangle ++(5.5,1.75);
		
		\draw[blue,draw opacity=.5,,pattern={Lines[angle=30,yshift=-5pt],opacity=.2},pattern color = blue!10] (-.25,-.25) rectangle ++(1.75,3.5);
		\draw[blue,draw opacity=.5,,pattern={Lines[angle=30,yshift=-0pt],opacity=.2},pattern color = blue!10] (1.5,-.25) rectangle ++(2,3.5);
		\draw[blue,draw opacity=.5,,pattern={Lines[angle=30,yshift=-5pt],opacity=.2},pattern color = blue!10] (3.5,-.25) rectangle ++(1.75,3.5);
		
        \foreach \s in {0,1}{
			\foreach \t in {0,1,2}{
				\foreach \x in {0,1,2} {
					\foreach \y in {0,1} {
						\draw[dashed] ($(2*\t,2*\s) + (0.5*\y,0.5*\x)$) -- ($(2*\t,2*\s) + (0.5*\y,0.5*\x) + (0.5,0)$);
						\draw[green!50!black] ($(2*\t,2*\s) + (0.5*\x,0.5*\y)$) -- ($(2*\t,2*\s) + (0.5*\x,0.5*\y) + (0,0.5)$);
						\draw[green!50!black] ($(2*\t,1) + (0.5*\x,0)$)--($(2*\t,2) + (0.5*\x,0)$);
						\draw[gray,dotted] ($(1,2*\s) + (0,0.5*\x)$)--($(2,2*\s) + (0,0.5*\x)$);
						\draw[gray,dotted] ($(1,2*\s) + (2,0.5*\x)$)--($(2,2*\s) + (2,0.5*\x)$);
						\draw[cyan] ($(1,3) + (\t*2,0)$)--(2.5,3.5);
						\draw[cyan] ($(0,0) + (\t*2,0)$)--(2.5,-0.5);
					}
				}
			}
		}
		
		\draw[very thick,purple] (0,0)--(0.5,0);
		\draw[very thick,purple] (0.5,2)--(1,2);
		\draw[very thick,purple] (0.5,2.5)--(1,2.5);
		\draw[very thick,purple] (4,0.5)--(4.5,0.5);
		
		\node[vertex,label=below:\(r_{\text{select}}\)] (rselect) at (2.5,-0.5) {};
		\node[vertex,label=above:\(t'_1\)] (t'1) at (2.5,3.5) {};
		
		\foreach \s in {0,1}{
		\foreach \t in {0,1,2}{
		\foreach \x in {0,1,2} {
			\foreach \y in {0,1,2} {
				\node[vertex] (\t\s\x\y) at ($(2*\t,2*\s) + (0.5*\x,0.5*\y)$) {};
			}
		}}}
	\end{tikzpicture}

%% file: ULV_reduction.tex
\label{subsec:W1HardULV}
Our hardness result from the previous section (\Cref{thm:bddepth-ETHlowerbound}) crucially relies on the depth constraint in \BLVshort.
In particular, the depth constraint allowed us to severely restrict which paths could be used to connect terminals to \(r\) in the solution. 
However, to prove that \ULVshort is also W[1]-hard, if we use the same reduction structure and set \(\lambda=+\infty\), then instead of having a path from \(\rcons\) (\(\rselect\)) to \(t_i\) (\(t'_i\)) for each \(i \in [k]\), a solution of lower cost has a fork-like structure consisting of a single row (column) that then branches out on the last column (row) to reach all terminals on the other side of the grid (see \Cref{fig:ULVinBLVStructure}).
I.e.\ parts of the solution that should correspond to choices in separate rows or columns collapse into one.
To prevent this issue and show that \ULVshort is W[1]-hard, we introduce gadgets, called \emph{connectors}, that contain further terminals throughout the inner parts of the grid. Concretely, the grid is divided into \(k\times k\) subgrids of size \(\rowNum \times \colNum\), called \emph{cell gadgets}. Whenever two cell gadgets are adjacent in the original grid, their corresponding boundary vertices are no longer directly connected; instead, they are each linked via a connector which contains terminals that force the part of the solution that must cover them to connect through it and prevent it from collapsing.

For our reduction it will also be more convenient for each valid row choice to have the same cost.
To achieve this, we give a reduction from a variant of \GT called \GTRegular where for each \(x\in [n]\) and each \(i,j \in [k]\), the number of pairs of the form \((x,y)\) in \(S_{i,j}\) is equal. We refer to this number by the \emph{regularity parameter} of a \GTRegular instance.
\begin{theorem}
\label{thm:GTregularW[1]}
    The \GTRegular problem is W[1]-hard parameterized by \(k\), and under the ETH admits
no \(f(k)\cdot n^{o(k)}\)-time algorithm for any computable function \(f\).
\end{theorem}
\begin{proof}
    The \(k\)-\textsc{clique} problem is W[1]-hard even when restricted to regular graphs \cite{BlueBook}. Using the reduction provided by Marx as Theorem 14.27 in \cite{BlueBook} for reducing \(k\)-\textsc{clique} to \GT, any regular input of the \(k\)-\textsc{clique} problem can be reduced to an instance \(\left(S_{i,j}\right)_{(i,j) \in [k]^2}\) of \GT such that for each \(a,a' \in [n]\) and each \(i,i',j,j'\in [k]\), the number of pairs with first entry equal to \(a\) in \(S_{i,j}\) is equal to the number of pairs with first entry equal to \(a'\) in \(S_{i',j'}\). Since \(k\)-\textsc{clique} is W[1]-hard and cannot be solved in time \(f(k)\cdot n^{o(k)}\) for any computable function \(f\), \GTRegular is also W[1]-hard and there is no \(f(k)\cdot n^{o(k)}\)-time algorithm that solves this problem for any computable function \(f\).
\end{proof}
Knowing that \GTRegular is W[1]-hard, we provide a reduction from this problem to prove that \ULV is W[1]-hard in the following theorem.

\begin{theorem}\label{thm:lowerbound:unbounded}
The \ULV problem is W[1]-hard even for \(t=4\) and under the ETH admits no \(f(|T|,t)\cdot |V(G)|^{o(\sqrt{|T|})\cdot f'(t)}\)-time algorithm for any computable functions \(f\) and \(f'\), even when restricted to planar input graphs.
\end{theorem}
\begin{proof}
    By \Cref{thm:GTregularW[1]}, \GTRegular is W[1]-hard parameterized by \(k\) and there is
no \(f(k)\cdot n^{o(k)}\)-time algorithm for any computable function \(f\) that solves it. We show that \ULVshort is also W[1]-hard and cannot be solved in time \(f(|T|,t)\cdot |V(G)|^{o(\sqrt{|T|})\cdot f'(t)}\) for any computable functions \(f\) and \(f'\)  by providing a reduction from \GTRegular. Since the input graph constructed for the input instance through this reduction is planar, this shows that this statement persists even when restricted to planar input graphs. Let \(\mathcal{I}=\left(S_{i,j}\right)_{(i,j) \in [k]^2}\) be an instance of \GTRegular with a regularity parameter \(\regpar\) whose pairs are drawn from \([n]\times [n]\) where \(n \geq 2\).
    Let \(c=6\), \(\connectLength =40(c+1)(n+1)(n(k+1)+1)(\ell_H+d_H+\ell_V)\), \(\rootLength=2\connectLength\), \(\srootLength=5\connectLength\), and \(\termLength=40\connectLength^3\) where the values of \(d_H\), \(\ell_H\), and \(\ell_V\) will be defined later. We construct an equivalent instance \((G,\ell,\dig,r,T,d,\alpha,\beta,t)\) of \ULVshort as follows.
    
    \paragraph*{Instance structure} 
    The main part of the graph \(G\) consists of two types of gadgets, namely \emph{cell} gadgets and \emph{connectors}.

    A connector for a sequence of pairs of vertices \(\mathcal{P}=((p_{1,1}, p_{1,2}), (p_{2,1}, p_{2,2}), \ldots, (p_{|\mathcal{P}|,1}, p_{(|\mathcal{P}|,2)}))\) consists of a set of paths of length \(c+1\), one for each \(x\in[|\mathcal{P}|]\), from \(p_{x,1}\) to \(p_{x,2}\), which are called \emph{linking} paths, and another set of \(c\) paths, referred to as \emph{supporting} paths, where the \(q\)-th path connects two new terminals by passing through the \(q\)-th vertex of the \(x\)-th linking path for \(x \in [|\mathcal{P}|]\) respectively. Each of the edges of the linking paths has length \(\connectLength\) and digging cost \(\frac{\connectLength}{4(c+1)}\) and each of the edges of the supporting paths has length \(\frac{\connectLength}{4(|\mathcal{P}|+1)(c+1)}\) and digging cost \(0\). Note that in all connectors \(\mathcal{P}\) is chosen such that \(\connectLength\) is always divisible by \((|\mathcal{P}|+1)(c+1)\), so the resulting length value is an integer. For \(q\in [c+1]\), we call the set of edges consisting of the \(q\)-th edge of each linking path the \(q\)-th \emph{level} of the connector. In other words, if the \(q\)-th linking path consists of vertices \(p_{q,1}=u^{q}_1,u^{q}_2,\ldots,u^{q}_{c+1}=p_{q,2}\), respectively, then the edge set \(\{u^{x}_qu^{x}_{q+1}\mid x \in [|\mathcal{P}|]\}\) is the \(q\)-th level of the connector. Moreover, we call all vertices of a connector except the endpoints of the linking paths the \emph{internal} vertices of this connector. A schematic view of a connector is shown in \Cref{fig:connector}. A schematic overview of this graph is depicted in \Cref{fig:ULVGraph}. Using these definitions, we construct the graph as follows.

    \begin{figure}[!ht]
        \centering
        \begin{subfigure}[b]{\textwidth}
            \centering
            \input{figs/connector}
            \subcaption{Schematic view of a connector. The third level of the connector consists of the orange edges.}
            \label{fig:connector}
        \end{subfigure}
        \hfill
        \begin{subfigure}[b]{\textwidth}
            \centering
            \input{figs/hconnector_instance}
            \subcaption{An instance of a row connector \(\hconnector^{i,j}\) and an indication of the part of a minimum solution in this connector. The purple part is shared between two trees that intersect this connector, and the red and blue edges distinguish between these two trees. In this solution, \(n-1\) is chosen as the first value of pairs of \(i\)-th row cells.}
            \label{fig:connector-instance}
        \end{subfigure}
    
        \caption{Schematic view of a connector and an instance of it. The terminals are colored gray.  The vertical edges have digging cost \(0\) and are of length \(\lceil\frac{\connectLength}{4(|\mathcal{P}|+1)(c+1)}\rceil\) where \(|\mathcal{P}|=\rowNum\) in the provided instance in \((b)\). The horizontal edges are of length \(\connectLength\) and have digging cost \(\lceil\frac{\connectLength}{c+1}\rceil\). The linking paths are made up of horizontal edges and the supporting paths consist of vertical edges.}
        \label{fig:twofigs}
    \end{figure}
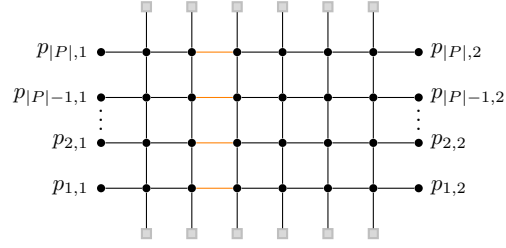
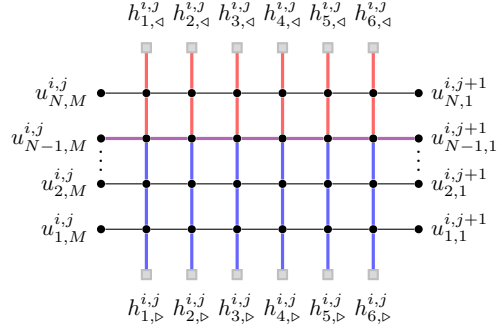

    \begin{itemize}
        \item \textbf{Cell gadgets.}  There are \(k \times k\) cell gadgets where cell gadget \(G^{i,j}\) corresponds to the cell \(S_{i,j}\) from the \GTRegular input instance. Each cell gadget is a grid graph with \(\rowNum=n\) rows and \(\colNum=n(k+1)\) columns. For ease of notation, for \(x \in [\rowNum]\) and \(y \in [\colNum]\), we denote the vertex in the \(x\)-th row and \(y\)-th column of a cell gadget \(G^{i,j}\) by \(u^{i,j}_{x,y}\). The lengths and digging costs of the edges of the grid are defined as follows.
        \begin{itemize}
            \item \textbf{Horizontal edges.} For each choice of \(x \in [\rowNum]\), \(y\in [k]\), and \(y'\in[n]\) there is a unique horizontal edge \(u^{i,j}_{x,(k+1)(y'-1)+y}u^{i,j}_{x,(k+1)(y'-1)+y+1}\) in the grid.  The digging cost of this edge is \(d_H=n(k+1)\) and the length of this edge depends on whether \((x,y') \in S_{i,j}\) or not, which is defined as follows.
            \begin{itemize}
                \item If \((x,y')\in S_{i,j}\) and \(i=y\) then the length of this edge is \(1\).
                \item Otherwise, it has length \(\ell_H=k+1\).
            \end{itemize}
            So, by definition of the \GTRegular problem, each row contains exactly \(\regpar\) number of edges of length \(1\) and the rest are of length \(\ell_H\).
            \item \textbf{Vertical edges.} For each vertical edge, we set its length to \(\ell_V=n(k+1)\) and its digging cost to \(0\).
        \end{itemize} 
        
        For each \(y \in [\colNum]\), we call the path that passes through vertices \(u^{i,j}_{1,y},u^{i,j}_{2,y},\ldots,u^{i,j}_{N,y}\) the \(y\)-th \emph{column path} of this gadget. Similarly, for each \(x\in [\rowNum]\), we call the path that passes through vertices \(u^{i,j}_{x,1},u^{i,j}_{x,2},\ldots,u^{i,j}_{x,\colNum}\) the \(x\)-th \emph{row path} of this gadget. An example of this gadget is shown in \Cref{fig:cellGadget:a}. Note that our choice of \(\connectLength\) ensures that that the total cost (including diggin cost) of the edges in cell gadget is at most \(\frac{\connectLength}{40}\). Thus, the total cost of a column path is at most \(\frac{\connectLength}{40\colNum}<\frac{\connectLength}{\colNum}\) and the total cost of a row path is at most \(\frac{\connectLength}{40\rowNum}<\frac{\connectLength}{\rowNum}\).
        \item \textbf{Row connectors.} For each \(i\in[k]\) and \(j\in[k-1]\),  there is a connector between cell gadgets \(G^{i,j}\) and \(G^{i,j+1}\) on the sequence of pairs \(((u^{i,j}_{1,\colNum},u^{i,j+1}_{1,1}),\ldots,(u^{i,j}_{\rowNum,\colNum},u^{i,j+1}_{\rowNum,1}))\), which is called a \emph{row connector} and is denoted \(\hconnector^{i,j}\). Note that by definition of \(\connectLength\) and \(\rowNum\), we have \(\lceil\frac{\connectLength}{4(\rowNum+1)(c+1)}\rceil=\frac{\connectLength}{4(\rowNum+1)(c+1)}\) which is the length of the edges of the supporting paths in these connectors.
        For each \(q\in[c]\), let \(h^{i,j}_{q,1},\ldots,h^{i,j}_{q,\rowNum}\) denote the interior of the \(q\)-th supporting path, ordered according to the linking paths: \(h^{i,j}_{q,1}\) lies on the first linking path, \(h^{i,j}_{q,2}\) lies on the second linking path, and so on, with \(h^{i,j}_{q,\rowNum}\) lying on the last linking path.
        We denote by \(h^{i,j}_{q,\RstartIndex}\) the terminal adjacent to \(h^{i,j}_{q,1}\), and by \(h^{i,j}_{q,\RendIndex}\) the terminal adjacent to \(h^{i,j}_{q,\rowNum}\). A schematic view of an instance of these connectors is shown in \Cref{fig:connector-instance}.
        \item \textbf{Column connectors.} For each \(i\in[k-1]\) and \(j \in [k]\), there is a connector between cell gadgets \(G^{i,j}\) and \(G^{i+1,j}\) on the sequence pairs \(((u^{i,j}_{\rowNum,1},u^{i,j+1}_{1,1}),\ldots,(u^{i,j}_{\rowNum,\colNum},u^{i,j+1}_{1,\colNum}))\), which is called a \emph{column connector} and is denoted \(\vconnector^{i,j}\). Note that by definition of \(\connectLength\) and \(\colNum\), we have \(\lceil\frac{\connectLength}{4(\colNum+1)(c+1)}\rceil=\frac{\connectLength}{4(\colNum+1)(c+1)}\) which is the length of the edges of the supporting paths in these connectors. 
        For each \(q\in[c]\), let \(v^{i,j}_{q,1},\ldots,v^{i,j}_{q,\rowNum}\) denote the interior of the \(q\)-th supporting path, ordered according to the linking paths: \(v^{i,j}_{q,1}\) lies on the first linking path, \(v^{i,j}_{q,2}\) lies on the second linking path, and so on, with \(v^{i,j}_{q,\rowNum}\) lying on the last linking path.
        We denote by \(v^{i,j}_{q,\RstartIndex}\) the terminal adjacent to \(v^{i,j}_{q,1}\), and by \(v^{i,j}_{q,\RendIndex}\) the terminal adjacent to \(v^{i,j}_{q,\rowNum}\).
    \end{itemize}
    \begin{figure}[!ht]
    \centering

    \begin{subfigure}{\textwidth}
        \centering
        \input{figs/cellGadget_a}
        \subcaption{The cell gadget \(G^{2,1}\) which corresponds to \(S_{2,1}\), where red edges are of length \(1\) and the other horizontal edges are of length \(\ell_H\). The vertical edges are of length \(\ell_V\). The digging cost of horizontal edges is \(d_H\) and the digging cost of vertical edges is \(0\).}
        \label{fig:cellGadget:a}
    \end{subfigure}

    \vspace{1em}

    \begin{subfigure}{\textwidth}
        \centering
        \input{figs/cellGadget_b} 
        \subcaption{The part of a minimum solution  in cell gadget \(G^{2,1}\) which can be transformed to a solution for \GTRegular. The first value of all chosen pairs from the second row of \GTRegular instance is \(2\) and the second value of all chosen pairs of the first column of this instance is \(3\). The yellow edges belong to \(H_{R,\RstartIndex}\) and \(H_{R,\RendIndex}\), blue edges belong to \(H_{C,\CstartIndex}\) and \(H_{C,\CendIndex}\), and green edges belong to all trees.}
        \label{fig:cellGadget:b}
    \end{subfigure}

    \caption{Let \(n=4\), \(k=2\), and \(S_{2,1}=\{(1,2),(1,4),(2,1),(2,3),(3,2),(3,4),(4,1),(4,3)\}\).}
    \label{fig:cellGadget}
\end{figure}
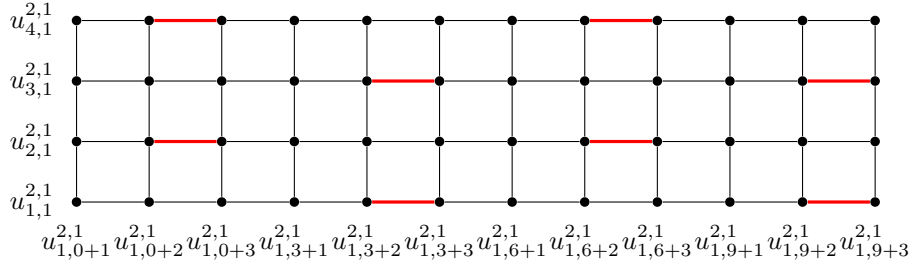
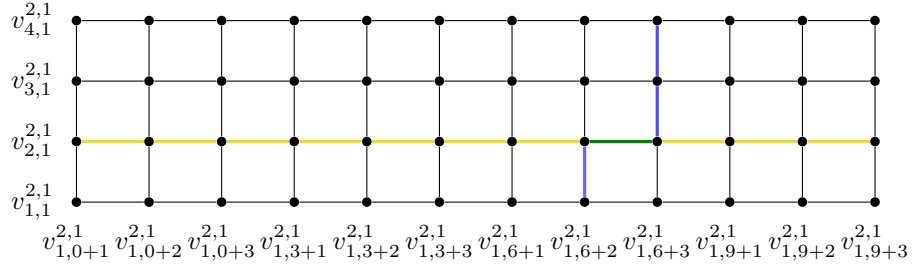
    \textbf{Root part.} To complete the description of the graph, additional sets of edges, vertices, and terminals are defined as follows. We refer to this additional part of the graph as its \emph{root part}. The root part can be divided into three parts, which are defined as follows. 
    
    \begin{itemize}
        \item \textbf{Row root part.} 
        \begin{itemize}
            \item There is a set of vertices \(b^{i,z}\) for \(i \in [k]\) and \(z \in \{\rightarrow,\leftarrow\}\), each with a pair of terminals \(p_\RstartIndex(b^{i,z})\) and \(p_\RendIndex(b^{i,z})\) as pendants, where \(b^{i,\rightarrow}\) is connected to vertices \(u^{i,1}_{x,1}\) for all \(x \in [\rowNum]\) and \(b^{i,\leftarrow}\) is connected to vertices \(u^{i,k}_{x,M}\) for all \(x \in [\rowNum]\). The length of each of the edges between these vertices and their pendant terminals is \(\termLength\), and the length of each of the rest of these edges is \(\rootLength\). The digging cost of all of them is \(0\). 
            \item There is also a vertex called \(\rcons\) connected to the root \(r\) with an edge of length \(\termLength\), which is connected to the vertices \(b^{i,\rightarrow}\) for all \(i \in [k]\) with edges of length \(\srootLength\) and digging cost \(0\) and is connected to two new terminals \(p_\RstartIndex(\rcons)\) and \(p_\RendIndex(\rcons)\) as pendants. Similarly, the length of the edges between these vertices and their pendant terminals is \(\termLength\), and their digging cost is \(0\).
        \end{itemize}
        We refer to the subgraph consisting of the above as \emph{row root part}.
        \item \textbf{Column root part.} 
        \begin{itemize}
            \item There is another set of vertices \(b^{z,j}\) for \(z \in \{\uparrow,\downarrow\}\) and \(j \in [k]\), referred to as \emph{column root part}, each with a pair of terminals \(p_\CstartIndex(b^{z,j})\) and \(p_\CendIndex(b^{z,j})\) as pendants, where \(b^{\uparrow,j}\) is connected to vertices \(u^{i,j}_{1,y}\) with \(y = y'\cdot(k+1)+1\) for all \(0\leq y'<n\), and \(b^{\downarrow,j}\) is connected to vertices \(u^{i,j}_{\rowNum,y}\) with \(y = y'\cdot(k+1)\) for all \(1\leq y'\leq n\). The length of the edges between these vertices and their pendant terminals is \(\termLength\), and the length of the rest of these edges is \(\rootLength\). The digging cost of all of them is \(0\). 
            \item  There is also a vertex called \(\rselect\) connected to the root \(r\) with an edge of length \(\termLength\), which is connected to the vertices \(b^{\uparrow,j}\) for all \(j \in [k]\) with edges of length \(\srootLength\) and digging cost \(0\) and is connected to two new terminals \(p_\CstartIndex(\rselect)\) and \(p_\CendIndex(\rselect)\) as pendants. The length of the edges between these vertices and their pendant terminals is \(\termLength\), and their digging cost is \(0\).
        \end{itemize}
        We refer to the subgraph consisting of the above as \emph{column root part}.
    \end{itemize}

    \begin{figure}[!ht]
        \centering
        \input{figs/ULVGraphSmall}
        \caption{Schematic view of the reduction in \Cref{thm:lowerbound:unbounded}. Connectors are indicated in yellow and cell gadgets in gray. The orange edges are of weight \(\rootLength\),the green edges are of weight \(\srootLength\), and the blue edges are of length \(\termLength\). An instance of a connector is shown for \(\hconnector^{k,1}\).} 
        \label{fig:ULVGraph}
    \end{figure}
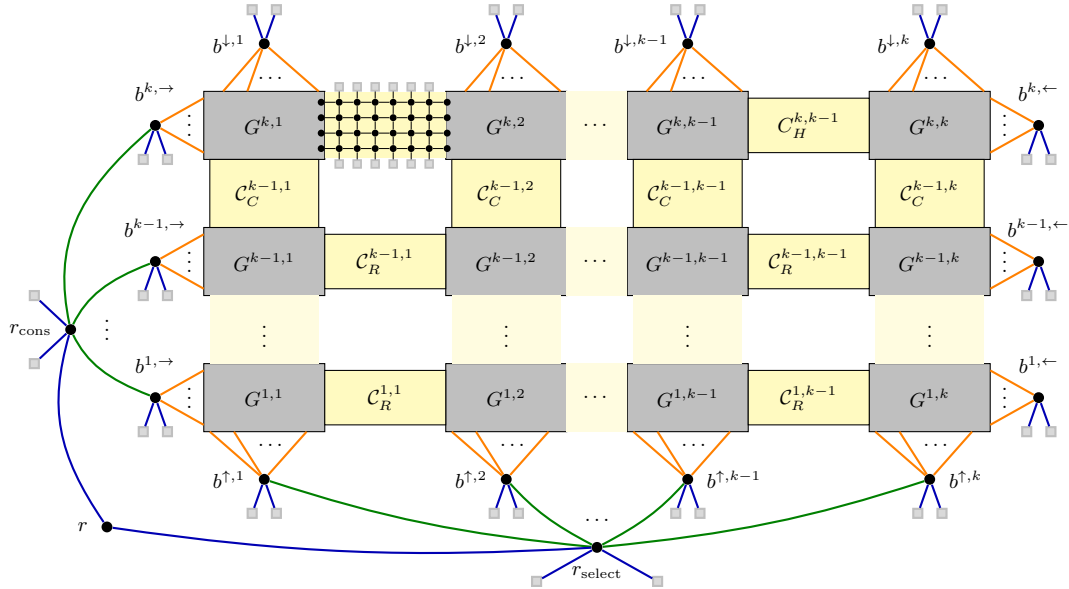
    
    \paragraph*{Reduction overview.}
Before finalizing the unspecified parts of the 
\ULVshort-instance, we provide some intuition for the structural design we aim to achieve.

Similarly to the reduction provided for \Cref{thm:bddepth-ETHlowerbound}, to encode a choice of pairs of numbers for each cell of the \GTRegular instance, our \ULVshort-solution should consist of two parts, one selecting $k$ rows and one selecting $k$ columns, such that the intersection of a selected row and a selected column consists of a single shared horizontal edge, which needs to be dug only once.

We will use \Cref{lemma:UniqueTerminalPartitioning} to find suitable demand values on the terminals that ensure the following distribution of terminals: the row-selection trees cover the terminals belonging to the row connectors and the terminals of the row root part, while the column selection trees only contain the terminals belonging to the column connectors and the terminals of the column root part. 

The digging costs and the edge lengths within the connector gadgets are defined so that these connectors are used only when necessary to establish connections between cell gadgets and cover the terminals. More specifically, in a minimum-cost solution, this prevents a tree encoding row choices to intersect column connectors and a tree encoding column choices to intersect row connectors.

To make the selection of linking paths within each connector independent of their distance from terminals, which affects the selection of row and column values, while preserving planarity, each connector includes supporting paths each with two terminals as its endpoints. 
In a cost-optimal solution, two trees should intersect each connector, each covering terminals on one side of all supporting paths, but together covering all supporting paths and their terminals; see \Cref{fig:connector-instance}.  

To ensure that within a row (column) connector the solution does not switch between rows (columns), our construction enforces a solution consisting of four (two for row and two for column choices) trees rather than two (one for row and one for column choices); see \Cref{fig:4TreesVs2Trees}.

\input{2TreesVs4Trees}
 
The edge lengths, particularly those between \(r\) and \(\rcons\) and between \(r\) and \(\rselect\), and \(\beta\) are chosen so that every root-terminal-path in the row-selection part passes through \(\rcons\), while every root-terminal-path in the column selection part passes through \(\rselect\). 

By construction and the cost budget, we will make, all edges of the two trees encoding row (column) choices are shared, except for those corresponding to the supporting paths inside the connectors and those that are only for the connection to the terminals of that tree.
This structure also constrains each \(\rcons\)-\(b^{i,\leftarrow}\) path to pass through \(b^{i,\rightarrow}\) and the interior of the path from \(b^{i,\leftarrow}\) to \(b^{i,\rightarrow}\) lies entirely within the subgraph
\(
\bigcup_{j=1}^{k} G^{i,j} \cup \bigcup_{j=1}^{k-1} \hconnector^{i,j},
\)
and using vertical edges from the cell gadgets is impossible since it would incur a prohibitively large cost. This design allows us to encode the first entry of the pair selected for the \(i\)-th row of the \(\GT\)-instance through the index of the row whose vertices appear on the \(\rcons\)-\(b^{i,\leftarrow}\) path in the trees encoding row choices.

Similarly, for each \(j \in [k]\), the \(\rselect\)-\(b^{\downarrow,j}\) path must pass through \(b^{\uparrow,j}\) and the interior of the path from \(b^{\downarrow,j}\) to \(b^{\uparrow,j}\) lies within exactly one of the \(n\) subgraphs of the form
\(
\bigcup_{i=1}^{k} G^{i,j}_y \cup \bigcup_{i=1}^{k-1} \vconnector^{i,j},
\)
for some \(y \in [n]\), where \(G^{i,j}_y\) denotes the union of the \((y-1)(k+1)\)-th through \((y(k+1)-1)\)-th column paths of \(G^{i,j}\) together with all horizontal edges connecting their vertices. Each such path contains exactly \(k\) horizontal edges in total from cell gadgets. This allows us to encode the second entry of the selected pair for the \(j\)-th column of the \(\GT\)-instance.

Finally, we must ensure that the encoded pair \((a,b) \in [n] \times [n]\) for each cell \((i,j) \in [k] \times [k]\) satisfies \((a,b) \in S_{i,j}\). To enforce this, we assign digging costs so that, for each \(j \in [k]\), every \(\rcons\)-\(b^{i,\leftarrow}\) path in row-selection trees necessarily shares a (horizontal) edge with each \(\rselect\)-\(b^{\downarrow,j}\) path in column-selection trees. We then make the horizontal edges prohibitively long whenever they connect the \(y(k+1)\)-th and \((y(k+1)+1)\)-th column paths on the \(x\)-th row path of \(G^{i,j}\) if either \((x,y) \notin S_{i,j}\) -- thereby excluding invalid pairs -- or \(y \neq i\) -- thereby ensuring that the selected column entry is compatible with the selected row entry for each \(i \in [k]\). 

We now proceed with the formal construction of the \ULVshort-instance to realize the intended behavior described above.

\paragraph*{\(d\), \(t\), \(\alpha\), and \(\beta\)}
    We set \(t = 4\); so a solution consists of four trees \(H_{R,\RstartIndex},H_{R,\RendIndex},H_{C,\CstartIndex},H_{C,\CendIndex}\). Additionally, we set  
    \begin{align*}
    \beta=&\underbrace{(8+8k)\termLength+8k\cdot \rootLength+4k\cdot \srootLength}_{\text{root part}}\\
    &\underbrace{+k^2(2(\colNum-1-\regpar)\ell_H+2\regpar+2(\rowNum-1)\ell_V+2+(\colNum-1)d_H)}_{\text{cell gadgets}}\\
    &\underbrace{+2k(k-1)\cdot(2c+3+\frac{c}{4(c+1)} )\cdot \connectLength}_{\text{connectors}},
    \end{align*}
    so that the existence of a solution for \ULVshort is guaranteed if and only if the given \GTRegular instance with the regularity parameter \(\regpar\) has a solution. 
    Consider a partition \(\mathcal{P}\) of terminals into four parts as follows. 
    \begin{align*}
    &T_{R,\RstartIndex}=\{p_\RstartIndex(\rcons)\}\cup\{p_\RstartIndex(b^{i,j})\mid i\in[k],j\in\{\rightarrow,\leftarrow\}\} \cup \{h^{i,j}_{q,\RstartIndex}\mid i\in[k],j\in[k-1],q\in[c]\},\\ 
    &T_{R,\RendIndex}=\{p_\RendIndex(\rcons)\}\cup\{p_\RendIndex(b^{i,j})\mid i\in[k],j\in\{\rightarrow,\leftarrow\}\} \cup \{h^{i,j}_{q,\RendIndex}\mid i\in[k],j\in[k-1],q\in[c]\},\\ 
    &T_{C,\CstartIndex}=\{p_\CstartIndex(\rselect)\}\cup\{p_\CstartIndex(b^{i,j})\mid i\in \{\uparrow,\downarrow\},j\in[k]\} \cup \{v^{i,j}_{q,\CstartIndex}\mid i\in[k-1],j\in[k],q\in[c]\},\\
    &T_{C,\CendIndex}=\{p_\CendIndex(\rselect)\cup\{p_\CendIndex(b^{i,j})\mid i\in \{\uparrow,\downarrow\},j\in[k]\} \cup \{v^{i,j}_{q,\CendIndex}\mid i\in[k-1],j\in[k],q\in[c]\}.
    \end{align*} 
    The first two parts cover all terminals in the row connectors together with the terminals of the row root part. The main difference of these two sets is in the group of terminals that they contain from each row connector. In more detail, one covers the terminals on the top of each row connector and the other one covers the ones on the bottom (see \Cref{fig:connector-instance}). The two other parts cover all terminals in the column connector as well as the terminals of the column root part. Similarly, the main difference is in covering the terminals of each side of the column connectors: right and left.
    By \Cref{lemma:UniqueTerminalPartitioning}, there exists a function \(f: T \to [|T|^{6}]\) such that \(\mathcal{P}\) is the only way of partitioning the terminals into at most four parts where the summation of \(f\)-values of each part is at most \(\sum\limits_{u \in T}\frac{f(u)}{4}\). Since the demand values are integers, by setting the demand function \(d=f\) and \(\alpha=\lfloor\sum\limits_{u \in T}\frac{d(u)}{4}\rfloor\), the goal becomes finding four trees that each covers exactly one part of \(\mathcal{P}\). 
    
    So, from now on, we can assume that a solution to this \ULVshort-instance consists of four trees where \(H_{a,a'}\)  covers \(T_{a,a'}\) for \(a\in \{R,C\}\) and \(a' \in \{\RstartIndex,\RendIndex\}\). Note that this guarantees the validity of the covering constraint and the demand constraint.

This concludes the description of the \ULVshort-instance. In the remainder of the proof, we establish that it is equivalent to the instance of \GTRegular that we started from.
    
\paragraph*{If \(\left(S_{i,j}\right)_{(i,j) \in [k]^2}\) is a yes-instance then \((G,\ell,\dig,r,T,d,\alpha,\beta,t)\) is a yes-instance.} Fix a solution \(((a_i,b_j))_{(i,j) \in [k] \times [k]}\) for \(\left(S_{i,j}\right)_{(i,j) \in [k]^2}\).
We transform this into a solution for \ULVshort on the input instance \((G,\ell,\dig,r,T,d,\alpha,\beta,\lambda,t)\) by specifying the four trees \(H_{R,\RstartIndex}\), \(H_{R,\RendIndex}\), \(H_{C,\CstartIndex}\) and \(H_{C,\CendIndex}\), two passing horizontally through rows and two passing vertically through columns. Since each pair of trees with the same orientation has a large overlap, we define their shared parts first. For each part of the solution, we bound sum of the length values occurring in that subgraph. Afterwards, we bound the digging costs for the entire solution.

We define \(H_{R,\RstartIndex}\cap H_{R,\RendIndex}\) as the union of the following parts.
\begin{enumerate}
\item The \(a_i\)-th row path from each cell gadget \(G^{i,j}\)  for each \(i,j \in [k]\)  where each has length \((M-1-\regpar)\ell_H+\regpar\).
\item The \(a_i\)-th linking path of the row connector \(\hconnector^{i,j}\) for each \(i \in[k]\) and \(j \in [k-1]\), where each has length \((c+1 )\cdot \connectLength\).
\item The edge \(r\rcons\) which has length \(\termLength\) along with the following edges. For each \(i \in [k]\),
\begin{itemize}
    \item the edges \(\rcons b^{i,\rightarrow}\) where each has length \(\srootLength\),
    \item the edges \(b^{i,\rightarrow}u^{i,1}_{a_i,1}\), and \(u^{i,k}_{a_i,\colNum}b^{i,\leftarrow}\) where each has length \(\rootLength\).
\end{itemize}
 These edges have a total length \(\termLength+k\cdot \srootLength+2k\cdot \rootLength\).
\end{enumerate}
So, the total length of \(H_{R,\RstartIndex}\cap H_{R,\RendIndex}\)  is \[k^2((\colNum-1-\regpar)\ell_H+\regpar)+k(k-1)(c+1 )\cdot \connectLength+\termLength+k\cdot \srootLength+2k\cdot \rootLength.\] Additionally, for each \(a\in \{\RstartIndex,\RendIndex\}\), the tree \(H_{R,a}\) contains the following parts. 
\begin{enumerate}
    \item For each \(i \in [k]\), the edges connecting \(\rcons\), \(b^{i,\rightarrow}\), and \(b^{i,\leftarrow}\) to their corresponding pendant terminals \(p_a(\rcons)\), \(p_a(b^{i,\rightarrow})\), and \(p_a(b^{i,\leftarrow})\); recall that each of these edges has length \(\termLength\).
    \item For each \(i \in [k]\) and each \(j \in [k-1]\), for all \(q \in [c]\),
    \begin{itemize}
        \item the \(h^{i,j}_{q,\RstartIndex}\)-\(h^{i,j}_{q,a_i}\)-subpath from the \(q\)-th supporting path of the row connector \(\hconnector^{i,j}\) if \(a=\RstartIndex\),
        \item the \(h^{i,j}_{q,\RendIndex}\)-\(h^{i,j}_{q,a_i}\)-subpath from the \(q\)-th supporting path of the row connector \(\hconnector^{i,j}\) if \(a=\RendIndex\).
    \end{itemize}
    Note that these subpaths does not share any edges and their union consists of all supporting paths of the corresponding row connector. Since each supporting path has length \(\frac{\connectLength}{4(c+1)}\), the total length of all supporting paths in a connector is \(\frac{\connectLength}{4(c+1)} \cdot c\). 
\end{enumerate}
So, the total length of the non-shared parts of trees \(H_{R,\RstartIndex}\) and \(H_{R,\RendIndex}\) is
\[(4+4k)\termLength+k(k-1)(\frac{c}{c+1})\cdot \connectLength.
\]

For each \(j \in [k]\), there should be a path \(P_j\) appearing in both \(H_{C,\CstartIndex}\) and \(H_{C,\CendIndex}\) that corresponds to the chosen value \(b_j\). This path from \(u^{1,j}_{1,(b_j-1)(k+1)+1}\) to \(u^{k,j}_{\rowNum,b_j(k+1)}\) which lies within \(G^{1,j}\cup \vconnector^{1,j} \cup G^{2,j} \cup \ldots \cup \vconnector^{k-1,j} \cup G^{k,j}\). The path \(P_j\) consists of exactly \(k\) horizontal edges, one in each of the cell gadgets \(G^{1,j},\ldots,G^{k,j}\) where the one in \(G^{i,j}\) belongs to the \(x\)-th row path if \(s_{i,j}=(x,y)\) and is from \(u^{i,j}_{(b_j-1)(k+1)+i}\) to \(u^{i,j}_{(b_j-1)(k+1)+(i+1)}\) (see \Cref{fig:cellGadget:b}). The rest of the edges of \(P_j\) are vertical edges and the edges of linking paths. 
We formally define \(H_{C,\CstartIndex}\cap H_{C,\CendIndex}\) as the union of the following parts. The path \(P_j\) can be achieved through \ref{VerticalPathInCellGadgets} and \ref{VerticalLinkingPaths}.
\begin{enumerate}
\item For each \(i,j \in [k]\), a path from \(u^{i,j}_{1,(b_j-1)(k+1)+i}\) which lies on the first row path of \(G^{i,j}\) to \(u^{i,j}_{\rowNum,(b_j-1)(k+1)+i+1}\) which belongs to the last row path of \(G^{i,j}\) consisting of the following parts.
\begin{itemize}
    \item The subpath of the \(((b_j-1)(k+1)+i)\)-th column path of each cell gadget \(G^{i,j}\) that is from \(u^{i,j}_{1,(b_j-1)(k+1)+i}\) to \(u^{i,j}_{a_i,(b_j-1)(k+1)+i}\).
    \item The subpath of the \(((b_j-1)(k+1)+i+1)\)-th column path of each cell gadget \(G^{i,j}\) that is from \(u^{i,j}_{a_i,(b_j-1)(k+1)+i+1}\) to \(u^{i,j}_{\rowNum,(b_j-1)(k+1)+i+1}\).
    \item The edge \(u^{i,j}_{a_i,(b_j-1)(k+1)+i}u^{i,j}_{a_i,(b_j-1)(k+1)+i+1}\)
\end{itemize}
These subpaths have \(\rowNum-1\) vertical edges of length \(\ell_V\). The horizontal edge corresponds to the choice of \((a_i,b_j)\) which has length \(1\). So, each of these paths has a total length \((\rowNum-1)\ell_V+1\).\label{VerticalPathInCellGadgets}
\item The \(b_j\)-th linking path of the column connector \(\vconnector^{i,j}\) for each \(i \in[k-1]\) and \(j \in [k]\). Each of these paths has length \((c+1 )\cdot \connectLength\). \label{VerticalLinkingPaths}
\item  The edge \(r\rselect\), which has length \(\termLength\), along with the following edges. For all \(j\in[k]\),
\begin{itemize}
    \item \(\rselect b^{\uparrow,j}\) where each has length \(\srootLength\),
    \item the edges \(b^{\uparrow,j}u^{1,j}_{1,(b_j-1)(k+1)+1}\) and \(u^{k,j}_{n,b_j\cdot(k+1)}b^{\downarrow,j}\) where each has length \(\rootLength\).
\end{itemize} 
These edges have a total length \(\termLength+k\cdot \srootLength+2k\cdot \rootLength\).
\end{enumerate}
Note that the total length of edges other than horizontal edges in each mentioned path \(P_j\) is the fixed value \(k(\rowNum-1)\ell_V+(k-1)(c+1 )\cdot \connectLength\) regardless of where the side steps are taken. Since the sidesteps are taken in edges with minimum length (the ones that correspond to the choice of \(s_{i,j}\) which has length \(1\)), each of these paths has length \(k((\rowNum-1)\ell_V+1)+(k-1)(c+1 )\cdot \connectLength\) which is the minimum possible value for such a path.
Thus, the total length of \(H_{C,\CstartIndex}\cap H_{C,\CendIndex}\) is 
\[k^2((\rowNum-1)\ell_V+1)+k(k-1)(c+1 )\cdot \connectLength+\termLength+k\cdot \srootLength+2k\cdot \rootLength.\]

Additionally, for each \(a\in \{\CstartIndex,\CendIndex\}\), the tree \(H_{C,a}\) contains the following parts.
\begin{enumerate}
    \item The edge \(r\rselect\) along with the edges connecting \(\rselect\), \(b^{\uparrow,j}\), and \(b^{\downarrow,j}\) to their corresponding pendant terminals \(p_a(\rselect)\), \(p_a(b^{\uparrow,j})\), and \(p_a(b^{\downarrow,j})\) for each \(j \in [k]\); recall that each of these edges has length \(\termLength\) and so, they have a total length\((2+2k)\termLength\).
    \item For each \(i \in [k-1]\) and \(j \in [k]\), for all \(q \in [c]\),
    \begin{itemize}
        \item the \(v^{i,j}_{q,\CstartIndex}\)-\(h^{i,j}_{q,b_j}\)-subpath from \(q\)-th supporting paths of the column connector \(\vconnector^{i,j}\) if \(a=\CstartIndex\),
        \item the \(v^{i,j}_{q,\CendIndex}\)-\(h^{i,j}_{q,b_j}\)-subpath from \(q\)-th supporting path of the column connector \(\vconnector^{i,j}\) if \(a=\CendIndex\).
    \end{itemize}
    Note that these subpaths does not share any edges and their union consists of all supporting paths of the corresponding column connector. Since each supporting path has length \(\frac{\connectLength}{4(c+1)}\), the total length of all supporting paths in a connector is \(\frac{\connectLength}{4(c+1)} \cdot c\).
\end{enumerate} 
Hence, the total length of the non-shared parts of trees \(H_{C,\CstartIndex}\) and \(H_{C,\CendIndex}\) is
\[(4+4k)\termLength+k(k-1)\cdot \frac{c}{4(c+1)}\cdot \connectLength.
\]

Since the edge with a non-zero digging cost only belong to cell gadgets and connectors, the digging cost can be computed as follows.
\begin{enumerate}
    \item The digging cost arising within each connector is \(\connectLength\) since the edges on supporting paths have digging cost \(0\), while the solution uses all \(c+1\) edges of exactly one linking path and each such edge has a digging cost \(\frac{\connectLength}{c+1}\).
    \item The digging cost arising within each cell gadget is \((\colNum-1)d_H\) since the vertical edges have digging cost \(0\) and the solution uses all \(\colNum-1\) edges of exactly one row path and each such edge has a digging cost \(d_H\).
\end{enumerate}
Hence, the total digging cost is \[2k(k-1)\cdot \connectLength+k^2(\colNum-1)d_H.\] 

Therefore, we can conclude that the total cost of this solution \(\text{cost}(H_{C_1,}, H_{C_2}, H_{R_1}, H_{R_2})\) is at most
\begin{align*}
&\underbrace{2k^2((\colNum-1-\regpar)\ell_H+\regpar)+2k(k-1)(c+1 )\cdot \connectLength+2\termLength+2k\cdot \srootLength+4k\cdot \rootLength}_{\text{length of two copies of \(H_{R,\RstartIndex}\cap H_{R,\RendIndex}\)}}\\
&+\underbrace{(4+4k)\termLength+k(k-1)(\frac{c}{4(c+1)})\cdot \connectLength}_{\text{length of non-shared parts of \(H_{R,\RstartIndex}\) and \(H_{R,\RendIndex}\)}}\\
&\underbrace{2k^2((\rowNum-1)\ell_V+1)+2k(k-1)(c+1 )\cdot \connectLength+2\termLength+2k\cdot \srootLength+4k\cdot \rootLength}_{\text{length of two copies of \(H_{C,\CstartIndex}\cap H_{C,\CendIndex}\)}}\\
&+\underbrace{(4+4k)\termLength+k(k-1)(\frac{c}{4(c+1)})\cdot \connectLength}_{\text{length of non-shared parts of \(H_{C,\CstartIndex}\) and \(H_{C,\CendIndex}\)}}\\
&+\underbrace{2k(k-1)\cdot \connectLength+k^2(\colNum-1)d_H}_{\text{total digging cost}}\\
&=\beta,
\end{align*}
and so, the cost constraint holds. 

One can easily see that, by construction, the shared part of \(H_{R,\RstartIndex}\) and \(H_{R,\RendIndex}\) in \(
\bigcup_{j=1}^{k} G^{i,j} \cup \bigcup_{j=1}^{k-1} \hconnector^{i,j},
\) for each \(i \in [k]\) is a path where these paths are disjoint and become connected through a subtree that is defined in the row root part while non-shared parts are paths that has exactly one endpoint in the shared part which keeps the subgraph connected and acyclic. A similar argument holds for \(H_{C,\CstartIndex}\) and \(H_{C,\CendIndex}\). Therefore, these subgraphs are trees. Moreover,
since \(H_{R,\RstartIndex}\), \(H_{R,\RendIndex}\), \(H_{C,\CstartIndex}\) and \(H_{C,\CendIndex}\) cover \(T_{R,\RstartIndex}\), \(T_{R,\RendIndex}\), \(T_{C,\CstartIndex}\), and \(T_{C,\CendIndex}\), respectively, the rest of constraints also hold for this solution, thus, this is a valid solution for this \ULVshort-instance.

    \paragraph*{If \((G,\ell,\dig,r,T,d,\alpha,\beta,t)\) is a yes-instance then \(\left(S_{i,j}\right)_{(i,j) \in [k]^2}\) is a yes-instance.} 
    
    Let \(\mathcal{S}=\{H_{R,\RstartIndex},H_{R,\RendIndex},H_{C,\CstartIndex},H_{C,\CendIndex}\}\) be the trees defined by a solution with minimum cost that witnesses that \((G,\ell,\dig,r,T,d,\alpha,\beta,t)\) is a yes-instance. Without loss of generality, assume that \(H_{R,\RstartIndex}\) and \(H_{R,\RendIndex}\) cover \(T_{R,\RstartIndex}\) and \(T_{R,\RendIndex}\) respectively and \(H_{C,\CstartIndex}\) and \(H_{C,\CendIndex}\) cover \(T_{C,\CstartIndex}\) and \(T_{C,\CendIndex}\) respectively. 
    
    First, we provide some necessary structural properties of a minimum solution in \Cref{claim:EmptyConnectors}, \Cref{claim:EmptyRootPart}, and \Cref{claim:SingleEdgeLevelConnector}. Then, considering these properties, we provide a lower bound for the cost of different parts of a minimum solution, namely the parts that lie within each cell gadget and the parts that lie within each connector, respectively in \Cref{claim:cellgadgetstructure} and \Cref{claim:connectorstructure}. Finally, we show that these lower bounds are tight in \Cref{claim:tightcosts} and discuss how we can extract a solution to the \GTRegular instance from this minimum solution \(\mathcal{S}\) in the rest of the proof. 

    We start with two properties regarding the locations of the trees of a minimum-cost solution are that each tree only intersects: a specific subset of connectors (\cref{claim:EmptyConnectors}) and a specific subset of vertices in the root part of \(G\) (\cref{claim:EmptyRootPart}). 
    To establish these properties, we show that if this is not true, we can provide a new solution \(\mathcal{S}'\) to the considered \ULV-instance, by replacing a part of \(\mathcal{S}\) by another subgraph where \(\text{cost}(\mathcal{S}')<\text{cost}(\mathcal{S})\), contradicting the cost-minimality assumption on \(\mathcal{S}\). \Cref{claim:MakeConnectedSubgraph} provides a tool for this replacement that will be used in the proofs of \Cref{claim:EmptyConnectors} and \Cref{claim:EmptyRootPart}. The following claim facilitates the discussion on the effect of a replacement that will be done in the proof of \Cref{claim:MakeConnectedSubgraph} .

    \begin{claim}
    \label{claim:ConnectorConnectionToSides}
        If \(H\) is a subgraph of \(G\) that, for some \(a \in \{R,C\}\) and \(a'\in \{\RstartIndex,\RendIndex\}\), 
        \begin{itemize}
            \item \(T_{a,a'} \cup \{r\} \subseteq V(H)\), and
            \item has at most two connected components, each of which contains at least two vertices,
        \end{itemize}
        then for each connector \(\connector\) of \(G\) that is defined on a sequence of pairs of vertices \(\mathcal{P}\) where \(\connector\) is a row connector if \(a=R\) and is a column connector if \(a=C\), we have \(\ell(H\cap \connector) \geq \connectLength\cdot (c-1)+\frac{\connectLength}{4(|\mathcal{P}|+1)(c+1)} \cdot c\).
    \end{claim}
    \begin{proof}
    We assume that \(\connector\) is a row connector \(\hconnector^{i,j}\). The proof for the column connectors is symmetric, since we only use this assumption to identify the indices.

      We now show that from at least \(c-1\) out of the \(c+1\) levels of \(\hconnector^{i,j}\), at least one edge appears in \(H\). Suppose, for the sake of contradiction, that this is not the case. Then, there exist at least three levels of \(\hconnector^{i,j}\) whose edges do not appear in \(H\). Let \(q_1<q_2<q_3\) be the indices of these levels. 
      Because a connector has a path-like structure with respect to the linear order of the levels (see~\Cref{fig:connector-instance}), removing the edges of levels \(q_1, q_2, q_3\) from a connector splits the connector into at least four connected components, just as removing three edges from a path yields four components \(F_1, F_2, F_3, F_4\). The two middle components, \(F_2\) and \(F_3\), do not contain endpoints of linking paths. When connectors are placed into the graph \(G\), the only connections from the outside to the inside of a connector are through the endpoints of the linking paths. Hence there can be a path in \(G\) between vertices of \(F_1\) and \(F_4\) that avoids levels \(q_1, q_2, q_3\), but no vertex of \(F_2\) or \(F_3\) can connect through the rest of the graph without using an edge of layers \(q_1, q_2\), or \(q_3\). Consequently, in the graph obtained from \(H\) by removing the edges of levels \(q_1, q_2, q_3\) (which is identical to \(H\) itself, by assumption), there are at least three connected components: two corresponding to \(F_2\) and \(F_3\), and one additional component containing the root. This contradicts the assumption about the number of connected components of \(H\). Thus, \(H\) contains an edge from at least \(c-1\) levels of \(\hconnector^{i,j}\) and hence the total length of the edges of \(H\) that lie on the linking paths is at least \((c-1)\cdot \connectLength\).
    
    On the other hand, the only way to connect to each terminal is via the unique edge incident to it. Since each connected component of \(H\) contains at least two vertices, all terminals of \(H\) are connected to some other vertices, and so all these edges are in \(H\). Recalling that \(T_{a,a'}\) contains all \(c\) terminals of \(\hconnector^{i,j}\) that are on one side of it (bottom or top) and each edge that is incident on those terminals has length \(\frac{\connectLength}{4(|\mathcal{P}|+1)(c+1)}\), the total length of these \(c\) edges is \(\frac{\connectLength}{4(|\mathcal{P}|+1)(c+1)} \cdot c\). This is a lower bound for the total length of the edges of \(H\) that are on the supporting paths of \(\hconnector^{i,j}\). This implies that \(\ell(H \cap \hconnector^{i,j})\) is at least \((c-1) \cdot \connectLength + \frac{\connectLength}{4(|\mathcal{P}|+1)(c+1)} \cdot c\).
    \end{proof}

    Next, we present the tool required to prove \Cref{claim:EmptyConnectors} and \Cref{claim:EmptyRootPart}. This tool shows that any disconnected subgraph that contains a subset of the terminals \(T_{a,a'}\) for some \(a\in\{R,C\}\) and \(a' \in \{\RstartIndex,\RendIndex\}\) can be transformed into a connected one while increasing the cost by only a bounded amount. We can consider the operation of replacing \(H\) by \(H'\) as a two-step process: first removing \(H' \setminus H\), which decreases the cost by at least \(\ell(H\setminus H')\), and then inserting \(H \ H'\), which increases the cost by at most \(\ell(H'\setminus H)+\dig(H'\setminus H)\).
    
    \begin{claim}
    \label{claim:MakeConnectedSubgraph}
        If \(H\) is a subgraph of \(G\) that, for some \(a\in\{R,C\}\) and \(a' \in \{\RstartIndex,\RendIndex\}\),
        \begin{itemize}
            \item contains all vertices in \(T_{a,a'} \cup \{r\}\), and
            \item has exactly two connected components, each of which contains at least two vertices,
        \end{itemize}
       then there exists a tree \(H' \subseteq G\) such that \(T_{a,a'}\cup \{r\} \subseteq V(H')\) and \(\dig(H'\setminus H)+\ell(H'\setminus H)-\ell(H\setminus H') \leq 5\connectLength\).
    \end{claim}
    \begin{proof}
        We assume that \(a=R\); the proof for the other case is symmetric. Let \(C_r\) be the connected component of \(H\) that contains \(r\), and let \(\bar{C}_r\) denote the other connected component.
        For any two subgraphs \(H_x\) and \(H_y\) of \(G\), we define \(\Delta(H_x,H_y)=\dig(H_y\setminus H_x)+\ell(H_y\setminus H_x)-\ell(H_x\setminus H_y)\). Since every terminal in \(V(H) \cap T\) belongs to either \(C_r\) or \(\bar{C_r}\), we now perform a case distinction based on the structure of \(H\). For each case, we construct a connected subgraph \(H'\) that contains \(T_{a,a'}\cup \{r\}\) and satisfies \(\dig(H'\setminus H)+\ell(H'\setminus H)-\ell(H\setminus H') \leq 5\connectLength\). Finally, in each case, if the constructed subgraph is not acyclic, we replace it with a spanning tree of it. Clearly, this replacement does not increase the value of \(\Delta(H,H')\) and yields a tree that satisfies all required conditions.
        \begin{enumerate}
            \renewcommand{\labelenumi}{\Roman{enumi}.}
            \renewcommand{\theenumi}{\Roman{enumi}}
            \item \textbf{For some \(i\in [k]\) and \(j\in[k-1]\), there exists at least one terminal in \(\hconnector^{i,j}\) that belongs to \(C_r\) and one terminal in \(\hconnector^{i,j}\) that belongs to \(\bar{C_r}\).} In this case, \(H'\) will be obtained from a copy of \(H\) via the following modification steps:
            \begin{enumerate}
                \renewcommand{\labelenumii}{\arabic{enumii}.}
                \renewcommand{\theenumii}{\arabic{enumii}}
                \item We remove all edges of \(H \cap \hconnector^{i,j}\) along with all internal vertices of \( \hconnector^{i,j}\) that are in \(H\); Note that the total length of the removed edges is at least \(\connectLength\cdot (c-1)+\frac{\connectLength}{4(\rowNum+1)(c+1)} \cdot c\) by \Cref{claim:ConnectorConnectionToSides} since \(a=R\) and \(\hconnector^{i,j}\) is a row connector.\label{Step:MakeConnectedSubgraphCase1:Removal}
                \item We insert the last column path of \(G^{i,j}\), the first column path of \(G^{i,j+1}\), the linking path that contains all unique neighbors of terminals in \(T_{a,a'}\cap \hconnector^{i,j}\), and all edges that are incident on these terminals; recall that each column path has length \((\rowNum-1)\cdot \ell_V\) and digging cost \(0\), each linking path has length \( \connectLength\cdot(c+1)\) and digging cost \(\frac{\connectLength}{c+1}\cdot(c+1)=\connectLength\), and each unique edge that is connected to a terminal of \(\hconnector^{i,j}\) has length \(\frac{\connectLength}{4(\rowNum+1)(c+1)}\) and digging cost \(0\).\label{Step:MakeConnectedSubgraphCase1:Add}
            \end{enumerate}
            The set of edges that are removed in step~\ref{Step:MakeConnectedSubgraphCase1:Removal} consists of the set \(E(H\setminus H')\) and \(E(H\cap H' \cap \hconnector^{i,j})\), and so the total length of the removed edges in step~\ref{Step:MakeConnectedSubgraphCase1:Removal} is \(\ell(H\setminus H')+\ell(H\cap H' \cap \hconnector^{i,j})\). Furthermore, the set of edges that are added in step~\ref{Step:MakeConnectedSubgraphCase1:Add} consists of the set \(E(H'\setminus H)\) and \(E(H\cap H' \cap \hconnector^{i,j})\). So, the total length of added edges in step~\ref{Step:MakeConnectedSubgraphCase1:Add} is \(\ell(H'\setminus H)+\ell(H\cap H' \cap \hconnector^{i,j})\). Therefore, recalling that there are \(c\) terminals in \(V(\hconnector^{i,j})\cap T_{a,a'}\), we have 
            \begin{align*}
                \ell(H'\setminus H)-\ell(H\setminus H')&=\ell(H'\setminus H)+\ell(H\cap H' \cap \hconnector^{i,j})-\ell(H\setminus H')-\ell(H\cap H' \cap \hconnector^{i,j})\\
                &\leq
                \begin{aligned}[t]
                &2(\rowNum-1)\cdot \ell_V+\connectLength\cdot(c+1)+\frac{\connectLength}{4(\rowNum+1)(c+1)}\cdot c\\
                &- \connectLength\cdot (c-1)-\frac{\connectLength}{4(\rowNum+1)(c+1)} \cdot c
                \end{aligned}\\
                &\leq 2(\rowNum-1)\cdot \ell_V+2\connectLength 
            \end{align*}
            Since we have \((\rowNum-1)\cdot \ell_V<\frac{\colNum\cdot \rowNum\cdot \ell_V}{\colNum}<\frac{\connectLength}{\colNum}\), we can conclude that \(\ell(H'\setminus H)-\ell(H\setminus H')\leq 2\frac{\connectLength}{\colNum}+2\connectLength\). Therefore, we have \(\Delta(H,H')=\dig(H'\setminus H)+\ell(H'\setminus H)-\ell(H\setminus H') \leq \connectLength + 2\frac{\connectLength}{\colNum}+2\connectLength\leq 5\connectLength\).

            We now show \(H'\) is connected and contains all terminals in \(T_{a,a'}\) and \(r\). Let \(V_C\) be the vertices in the last column path of \(G^{i,j}\) together with the vertices in the first column path of \(G^{i,j+1}\). All connected components that remain after step~\ref{Step:MakeConnectedSubgraphCase1:Removal} should contain at least one vertex from \(V_C\) since they do not belong to the internal vertices of \(\hconnector^{i,j}\) and there is a path from their vertices to a terminal in \(\hconnector^{i,j}\) in the corresponding connected component of \(H\). Noting that the edges and vertices of these connected components are also in \(H'\) and the added part in step~\ref{Step:MakeConnectedSubgraphCase1:Add} connects all vertices of \(V_C\), hence \(H'\) is connected. Since \(r\) belongs to one the remaining connected components after step~\ref{Step:MakeConnectedSubgraphCase1:Removal} and each terminal of \(T_{a,a'}\) is either in the added part or in one of these connected components, therefore, \(H'\) covers \(T_{a,a'}\cup \{r\}\).
            
            \item \textbf{For each row connector \(\hconnector\), all terminals of \(T_{a,a'}\) contained in \(\hconnector\) lie in the same connected component of \(H\).}
            We have following subcases for this case.
            \begin{enumerate}
                \item \textbf{For some \(i\in [k],j\in [k-1]\), the terminals of \(T_{a,a'}\) that lie in \(\hconnector^{i,j}\) belong to a different connected component of \(H\) than those in \(\hconnector^{i,j+1}\).} Let \(u\) be the unique neighbor of the terminal \(h^{i,j}_{c,a'}\) and \(v\) be the unique neighbor of \(h^{i,j+1}_{1,a'}\). Let \(x=1\) if \(a=\RstartIndex\) and \(x=\rowNum\) if \(a=\RendIndex\). Then, we know that \(u\) lies on the \(x\)-th linking path of \(\hconnector^{i,j}\) and \(v\) lies on the \(x\)-th linking path of 
            \(\hconnector^{i,j+1}\). We set \(H'\) to be a copy of \(H\) union the following parts:
            \begin{itemize}
                    \item The first and last column paths and the first row path of \(G^{i,j}\). Recall that each column path has a total cost at most \(\frac{\connectLength}{\colNum}\) and each row path has a total cost at most \(\frac{\connectLength}{\rowNum}\).
                    \item Edges \(u^{i,j}_{x,\colNum}u\) and \(u^{i,j+1}_{x,1}v\), i.e.\ the edges connecting \(u\) and \(v\) to their neighbor in the cell gadget \(G^{i,j}\). Each of these edges has length \(\connectLength\) and digging cost \(\frac{\connectLength}{c+1}\).
            \end{itemize}
            Noting that all edges of \(H\) are still in \(H'\), we have \(\ell(H\setminus H')=0\). Therefore, we have  \(\Delta(H,H')=\ell(H'\setminus H)+\dig(H'\setminus H)\) and so, we can rewrite \(\Delta(H,H')\) as 
            \begin{align*}
                \Delta(H,H')\leq&2\cdot \frac{\connectLength}{\colNum}+\frac{\connectLength}{\rowNum}+2\cdot (\connectLength+\frac{\connectLength}{c+1})\\
                \leq& 3\cdot \frac{\connectLength}{\rowNum}+2\cdot (\connectLength+\frac{\connectLength}{c+1})\tag*{\text{\((\colNum>\rowNum)\)}}\\
                \leq& 5\connectLength\tag*{\text{\((\rowNum \geq 2 \land c=6)\)}}.
            \end{align*}
            Noting that \(V(H)\subseteq V(H')\) and \(T_{a,a'}\cup \{r\} \subseteq V(H)\), we have \(T_{a,a'}\cup \{r\}\subseteq V(H')\) as well. Now we show that \(H'\) is connected. Since no connected component of \(H\) consists of a single vertex by the precondition, both \(h^{i,j}_{c,a'}\) and  \(h^{i,j+1}_{1,a'}\) are connected to at least one other vertex in \(H\). Since \(u\) and \(v\) are their unique neighbors, the edges \(h^{i,j}_{c,a'}u\) and \(h^{i,j+1}_{1,a'}v\) are in \(E(H)\) and so in \(E(H')\). Since \(h^{i,j}_{c,a'}\) and  \(h^{i,j+1}_{1,a'}\) belong to different connected components of \(H\), \(u\) and \(v\) are also in different connected components of \(H\). Noting that \(u\) and \(v\) are connected through the added part \(H'\setminus H\), the connected components of \(H\), which are also in \(H'\) before adding the new part \(H'\setminus H\), must become connected through the added part in \(H'\).

            \item \textbf{For some \(i \in [k]\), the terminals of \(\hconnector^{i,1}\) belong to a different connected component of \(H\) than \(p_{a'}(b^{i,\rightarrow})\).} Let \(u\) be the unique neighbor of \(h^{i,1}_{1,a'}\) and let \(x\) be the index of the linking path of \(\hconnector^{i,1}\) that contains \(u\). We set
            \(H'\) to be a copy of \(H\) union the following parts:
            \begin{itemize}
                \item the last column path and the first row path of the cell gadget \(G^{i,1}\). Recall that each column path has a total cost at most \(\frac{\connectLength}{\colNum}\) and each row path has a total cost at most \(\frac{\connectLength}{\rowNum}\).
                \item the edges \(b^{i,\rightarrow}u^{i,1}_{1,1}\) and \(uu^{i,1}_{x,\colNum}\), i.e.\ the edge connecting \(b^{i,\rightarrow}\) to one of its neighbors in the cell gadget \(G^{i,j}\) and the edge that connects \(u\) to its unique neighbor in the cell gadget \(G^{i,j}\). These edges have lengths \(\rootLength\) and \(\connectLength\) and digging costs \(0\) and \(\frac{\connectLength}{c+1}\), respectively.
            \end{itemize}
            Since \(E(H)\subseteq E(H')\), we have \(\ell(H\setminus H')=0\). Thus, we can conclude that  \(\Delta(H,H')=\ell(H'\setminus H)+\dig(H'\setminus H)\) and so, we can rewrite \(\Delta(H,H')\) as 
            \begin{align*}
                \Delta(H,H')\leq&\frac{\connectLength}{\colNum}+\frac{\connectLength}{\rowNum}+\rootLength+\connectLength+\frac{\connectLength}{c+1}\\
                \leq& 2\cdot\frac{\connectLength}{\rowNum}+\rootLength+\connectLength+\frac{\connectLength}{c+1}\tag*{\text{\((\colNum>\rowNum)\)}}\\
                \leq& 5\connectLength\tag*{\text{\((\rowNum\geq 2\land \rootLength=2\connectLength \land c=6)\).}}
            \end{align*}
            We now show that \(H'\) is connected. Since  \(p_{a'}(b^{i,\rightarrow})\) and \(h^{i,1}_{1,a'}\) are connected to some other vertices in \(H\), the edges \(b^{i,\rightarrow}p_{a'}(b^{i,\rightarrow})\) and \(uh^{i,1}_{1,a'}\) should be in \(H\) as well as these are the only edges that are incident on \(p_{a'}(b^{i,\rightarrow})\) and \(h^{i,1}_{1,a'}\) respectively. Since \(p_{a'}(b^{i,\rightarrow})\) and \(h^{i,1}_{1,a'}\) are in different connected components of \(H\), the vertices \(u\) and \(b^{i,\rightarrow}\) also must belong to different connected components of \(H\). These connected components are in \(H'\) as well before adding the new part. Then, since the added part connects \(u\) and \(b^{i,\rightarrow}\), these connected components become connected and form a single connected component in \(H'\).
            Moreover, noting that \(V(H)\subseteq V(H')\), we can conclude that \(H'\) covers \(T_{a,a'}\cup\{r\}\).
            
            \item \textbf{For some \(i \in [k]\), the vertex \(p_{a'}(b^{i,\leftarrow})\) belongs to a different connected component than the terminals in \(\hconnector^{i,k}\).} This is symmetric to the previous case:
                \begin{itemize}
                    \item Instead of the unique neighbor of \(h^{i,1}_{1,a'}\), the unique neighbor of \(h^{i,k-1}_{c,a'}\) is considered;
                    \item Instead of the first row path and the last column path of the cell gadget \(G^{i,1}\), the first row path and the first column path of \(G^{i,k}\) are added;
                    \item Instead of the edge \(b^{i,\rightarrow}u^{i,1}_{1,1}\) and the edge between the unique neighbor of \(h^{i,1}_{1,a'}\) and its neighbor in the last column path of \(G^{i,1}\), the edge \(b^{i,\leftarrow}u^{i,1}_{1,\colNum}\) and the edge between the unique neighbor of \(h^{i,k-1}_{c,a'}\) and its neighbor in the first column path of \(G^{i,k}\) are added.
                \end{itemize}
            \item \textbf{None of the previous subcases hold.} In this case, for each \(i\in [k]\), we know that the terminals \(T_{a,a'}\cap \left(\bigcup\limits_{j\in[k-1]}\hconnector^{i,j}\right)\) belong to the same connected component of \(H\) as \(p_{a'}(b^{i,\rightarrow})\) and \(p_{a'}(b^{i,\leftarrow})\). Since \(\rcons\) is the unique neighbor of \(p_{a'}(\rcons)\), it is either in \(C_r\) or \(\bar{C_r}\). We have three subcases for this scenario.
            \begin{enumerate}
                \item \textbf{For some \(i\in[k]\),\(p_{a'}(b^{i,\rightarrow})\) belongs to a connected component other than the one that \(\rcons\) belongs to.} Let \(H'\) be a copy of \(H\), together with the edge \(b^{i,\rightarrow}\rcons\). Since this edge has length \(\srootLength\) and digging cost \(0\), we have \(\Delta(H,H')=\srootLength\) where \(\srootLength =5\connectLength\). We prove that the constructed \(H'\) is connected. Since each connected component has at least two vertices, the vertex \(b^{i,\rightarrow}\) is in the same connected component as \(p_{a'}(b^{i,\rightarrow})\) in \(H\). Hence the inserted edge \(b^{i,\rightarrow}\rcons\) connects the two connected components of \(H\), resulting in a connected subgraph. Furthermore, since \(V(H)\subseteq V(H')\), \(H'\) includes \(T_{a,a'}\cup\{r\}\).
                \item \textbf{All terminals of \(H\) belong to \(C_r\).} In this case, \(C_r\) contains all terminals of \(H\) as well as \(r\). We consider \(H'\) as a copy of \(C_r\). Clearly, it is a connected graph that covers \(T_{a,a'}\cup \{r\}\) where \(\Delta(H,H')\leq 0\) since \(E(H')\subseteq E(H)\) and so \(\ell(H'\setminus H)+\dig(H'\setminus H)=0\). 
                
                \item \textbf{All terminals of \(H\) belong to \(\bar{C_r}\).} Let \(H'\) be a copy of \(\bar{C_r}\) along with the edge \(r\rcons\). Recall that this edge has length \(\termLength\) and digging cost \(0\). Hence \(\ell(H'\setminus H)+\dig(H'\setminus H)\leq \termLength\). On the other hand, considering that \(\rcons\) is the unique neighbor of \(p_{a'}(\rcons) \in T_{a,a'}\), it should belong to \(\bar{C_r}\), and so, the edge \(r\rcons\) is not in \(H\). Since \(C_r\) contains at least two vertices, the edge \(r\rselect\) must be in \(C_r\) since \(r\) must pass through this edge to reach other vertices. Thus, we have \(\ell(H\setminus H')\geq \termLength\) as the edges of \(C_r\) are in \(H\) but not in \(H'\). Therefore, we can conclude that \(\Delta(H,H')\leq \termLength-\termLength=0\). Clearly, \(H'\) is connected as \(\bar{C_r}\) is connected and contains \(\rcons\), and since \(\bar{C_r}\) contain all terminals and \(r\) is also added to \(H'\) through the edge \(r\rcons\), it contains \(T_{a,a'}\cup \{r\}\).
                \end{enumerate}
        \end{enumerate}
        \end{enumerate}
        
        Since the case distinction is exhaustive, this completes the proof for \Cref{claim:MakeConnectedSubgraph} in the \(a=R\) case. We can show a similar argument for \(a=C\) as follows. Firstly, we swap the role of row and column connectors, the role of the row and columns of each cell gadget, and the role of \(\rcons\) and \(\rselect\) in the proof. We also note the following facts: 
        \begin{itemize}
            \item The cost of adding a row or a column path of a cell gadget is at most \(\frac{\connectLength}{\rowNum}\);
            \item The cost of adding and removing an edge from a linking or supporting path from a connector, and the cost of connections to \(\rselect\), \(b^{\uparrow,1},\ldots,b^{\uparrow,k}\) and \(b^{\downarrow,1},\ldots,b^{\downarrow,k}\) compared to the cost of connections to \(\rcons\), \(b^{1,\rightarrow},\ldots,b^{k,\rightarrow}\) and \(b^{1,\leftarrow},\ldots,b^{k,\leftarrow}\) would remain the same.
        \end{itemize}
        This completes the proof of this claim.
    \end{proof}

    We are not ready to prove \Cref{claim:EmptyConnectors} -- that each tree only intersects a specific subset of connectors. In more detail, we show that \(H_{R,\RstartIndex}\cup H_{R,\RendIndex}\) does not intersect column connectors and \(H_{C,\CstartIndex}\cup H_{C,\CendIndex}\) does not intersect row connectors.
    
    \begin{claim}
        \label{claim:EmptyConnectors}
        For each \(a \in \{\RstartIndex,\RendIndex\}\), we have
        \begin{enumerate}
            \item \(E(H_{R,a}) \cap E(\vconnector^{i,j}) = \emptyset\) for each \(i \in [k], j\in[k-1]\), and
            \item \(E(H_{C,a}) \cap E(\hconnector^{i,j}) = \emptyset\)  for each \(i \in [k-1], j\in[k]\).
        \end{enumerate}  
    \end{claim}
    \begin{proof}
        Assume that for some \(i,j\in[k]\) and \(a \in \{\RstartIndex,\RendIndex\}\), we have \(E(H_{R,a}) \cap E(\vconnector^{i,j}) \neq  \emptyset\).
        We construct \(H'_{R,a}\) as follows. Initially, let \(H'_{R,a}=H_{R,a}\). 
        \begin{enumerate}
            \item \textbf{If \(H_{R,a}\) contains a path from one of the vertices of the last row path of \(G^{i,j}\) to one of the vertices of the first row path of \(G^{i+1,j}\) that lies completely in \(\vconnector^{i,j}\):} 
            \begin{enumerate}
                \item We remove all edges of \(\vconnector^{i,j}\) from \(H'_{R,a}\). Since there is at least one path from one of the vertices the vertices of the last row path of \(G^{i,j}\) to one of the vertices of the first row path of \(G^{i+1,j}\) that lies completely in \(\vconnector^{i,j}\), after removing these edges, we reduce the cost by at least \(\connectLength\cdot (c+1)\).
                \item Add the last row path of \(G^{i,j}\) and the first row path of \(G^{i+1,j}\) to \(H'_{R,a}\). 
            Recall that each of these paths has a total cost \( \frac{\connectLength}{\rowNum}\) as shown before.
            \end{enumerate}
            After this, \(H'_{R,a}\) has at most two connected components, since we only lost connection between the vertices of the last row path of \(G^{i,j}\) and the vertices of the first row path of \(G^{i+1,j}\) and the vertices within each path are still connected.
            \begin{itemize}
                \item \textbf{If \(H'_{R,a}\) has exactly one connected component:} In this case,  the cost increases by at most \(\frac{2\connectLength}{\rowNum}-2\cdot \connectLength < 0\), which contradicts the cost-minimality of the solution.
                \item \textbf{If \(H'_{R,a}\) contains exactly two connected components.} Since \(H'_{R,a}\) is a subgraph that contains all vertices \(T_{R,a} \cup \{r\}\) and has exactly two connected components, each of which has at least two vertices, by \Cref{claim:MakeConnectedSubgraph}, there exists a tree \(H''_{R,a}\) where \( \dig(H''_{R,a}\setminus H'_{R,a})+\ell(H''_{R,a}\setminus H'_{R,a})-\ell(H'_{R,a}\setminus H''_{R,a}) \leq 5\connectLength\). Therefore, by replacing \(H_{R,a}\) in the considered minimum solution with \(H''_{R,a}\), we can construct a solution where the cost increases by at most \(\frac{2\connectLength}{\rowNum}+5\connectLength-\connectLength\cdot(c+1)<0\). This contradicts the cost-minimality assumption on the considered solution. 
            \end{itemize}
            \item \textbf{If \(H_{R,a}\) does not contain any paths from the vertices of the last row path of \(G^{i,j}\) to the vertices of the first row path of \(G^{i+1,j}\) that lies completely in \(\vconnector^{i,j}\):}
            We modify \(H'_{R,a}\) as follows. 
            \begin{enumerate}
                \item We remove all edges of \(\vconnector^{i,j}\) from \(H'_{R,a}\). \label{mod:EmptyConnectRemove}
                \item If there is a path between two vertices among the vertices of the last row path of \(G^{i,j}\) in \(H_{R,a}\), we add the the last row path of \(G^{i,j}\) to \(H'_{R,a}\).\label{mod:EmptyConnectLastRowPathAdd}
                \item Similarly, if there is a path between two vertices among the vertices of the first row path of \(G^{i+1,j}\) in \(H_{R,a}\), we add the the first row path of \(G^{i+1,j}\) to \(H'_{R,a}\).\label{mod:EmptyConnectFirstRowPathAdd}
            \end{enumerate}
            In this case, the part of paths in \(H_{R,a}\) that intersect the interior of \(\vconnector^{i,j}\) should be either between two vertices among the vertices of the last row path of \(G^{i,j}\) or between two vertices among the vertices of the first row path of \(G^{i+1,j}\). Thus, \(H'_{R,a}\) must be still connected after these modification steps. On the other hand, since no terminal in \(\vconnector^{i,j}\) belongs to \(T_{R,a}\) and we are considering the minimum-cost solution, there is at least one path in \(H_{R,a}\) that is either between two vertices among the vertices of the last row path of \(G^{i,j}\) or between two vertices among the vertices of the first row path of \(G^{i+1,j}\) in \(H_{R,a}\) and lies completely in \(\vconnector^{i,j}\). Thus, there are at least two edges of length \(\connectLength\) in \(\vconnector^{i,j}\) that belong to \(H_{R,a}\) and so, after step~\ref{mod:EmptyConnectRemove}, the cost decreases by at least \(2\cdot \connectLength\).
            Recall that each row path has a total cost \( \frac{\connectLength}{\rowNum}\). Hence, after these modifications, the cost increases by at most \(2\cdot\frac{\connectLength}{\rowNum}-2\cdot \connectLength < 0\), which contradicts the cost-minimality of the solution.
        \end{enumerate}
        This completes the proof for the first case of \Cref{claim:EmptyConnectors}. We can show similarly that \(E(H_{C,a}) \cap E(\hconnector^{i,j}) = \emptyset\) by swapping the role of row and column connectors, the role of the rows and columns of each cell gadget, and the role of \(\rcons\) and \(\rselect\) in the proof and noting that 
        \begin{itemize}
            \item the cost of adding a row or a column path of a cell gadget costs at most \(\frac{\connectLength}{\rowNum}\), and
            \item the cost of adding or removing the supporting paths from connectors is at most \(\frac{\connectLength}{4(\rowNum+1)(c+1)}\), and
            \item the cost of adding and removing an edge from a linking or supporting path from a connector, and the cost of connections to \(\rselect\), \(b^{\uparrow,1},\ldots,b^{\uparrow,k}\) and \(b^{\downarrow,1},\ldots,b^{\downarrow,k}\) compared to the cost of connections to \(\rcons\), \(b^{1,\rightarrow},\ldots,b^{k,\rightarrow}\) and \(b^{1,\leftarrow},\ldots,b^{k,\leftarrow}\) would remain the same.
        \end{itemize}
        This proves the claim.
    \end{proof}

    The other property that we mentioned for a minimum solution is that each tree only intersects a specific subset of the vertices of the root part of \(G\). More elaborately, \(H_{R,\RstartIndex} \cup H_{R,\RendIndex}\) does not contain any vertices from the column root part, and \(H_{C,\CstartIndex} \cup H_{C,\CendIndex}\) does not contain any vertices from the row root part of \(G\). 
    
    \begin{claim}
        \label{claim:EmptyRootPart}
        For \(a \in \{\RstartIndex,\RendIndex\}\), we have
        \begin{enumerate}
            \item \(V(H_{R,a}) \cap \{b^{i,j}\mid i \in \{\uparrow,\downarrow\} \land j \in [k]\} = \emptyset\), and
            \item \(V(H_{C,a}) \cap \{b^{i,j}\mid i \in [k] \land j \in \{\rightarrow,\leftarrow\}\} = \emptyset\).
        \end{enumerate}
    \end{claim}
    \begin{proof}
        Firstly, we show that \(V(H_{R,a}) \cap \{b^{i,j}\mid i \in \{\uparrow,\downarrow\} \land j \in [k]\} = \emptyset\). The proof for showing that \(V(H_{C,a}) \cap \{b^{i,j}\mid i \in [k] \land j \in \{\rightarrow,\leftarrow\}\} = \emptyset\) can be obtained similarly, where instead of considering the cell gadgets of the first and last rows, we consider the cell gadgets of the first and last columns.
        
        Assume towards a contradiction that for some \(j \in [k]\), the vertex \(b^{\uparrow,j}\) is in \(H_{R,a}\); we will treat ~$b^{\downarrow,j}$ later. We prove the claim by making a case distinction on whether \(r\rselect\) belongs to \(H_{R,a}\) or not.
        \begin{itemize}
            \item \textbf{\(r\rselect \in E(H_{R,a})\).} We construct a subgraph \(H'_{R,a}\) from \(H_{R,a}\) by modifying it as follows.
            \begin{itemize}
                \item We remove all edges that are incident on the vertices \(\{b^{\uparrow,j'}\mid j' \in [k]\}\) as well as the edge \(r\rselect\). Recall that the edge \(r\rselect\) has length \(\termLength\) and hence, the removed part has length at least \(\termLength\).\label{mod:RemoveRselectEdges}
                \item For each \(j' \in [k]\), we add the first row path of \(G^{1,j'}\). Recall that the cost of each row path is at most \(\frac{\connectLength}{\rowNum}\).
                \item For each \(j' \in [k-1]\), we add the first linking path of \(\hconnector^{1,j'}\) along with all of its supporting paths. Recall that the cost of a linking path of a connector is \((c+2)\cdot\connectLength\), and the cost of each of \(c\) supporting paths of a connector is \(\frac{\connectLength}{4(c+1)}\).
                \item We add the edges 
                \begin{itemize}
                    \item \(\rcons b^{1,\rightarrow}\) which has cost \(\srootLength\),
                    \item \(b^{1,\rightarrow}u^{1,1}_{1,1}\) and \(u^{1,k}_{1,\colNum}b^{1,\leftarrow}\) where each costs \(\rootLength\).
                \end{itemize}
            \end{itemize}
             This procedure increases the cost by at most \[\srootLength+2\rootLength+k\cdot\frac{\connectLength}{\rowNum}+(k-1)\cdot((c+2)\cdot\connectLength+\frac{\connectLength}{4(c+1)}\cdot c) - X<0.\]
             
             We claim that this subgraph is connected and includes all terminals covered by \(H_{R,a}\) as well as \(r\). Since we did not remove any edges that are incident on a terminal of \(H_{R,a}\), it is clear that \(H'_{R,a}\) contains all those terminals. Furthermore, by \Cref{claim:EmptyConnectors}, 
            we have \(E(H_{R,a}) \cap E(\vconnector^{1,j'}) = \emptyset\) for all \(j' \in [k-1]\) which shows that removal edges in step~\ref{mod:RemoveRselectEdges} can only disconnect a part of the terminals of \(T_{R,a}\) that lies within the first row, i.e.\ the terminals of the row connectors \(\hconnector^{1,1},\ldots,\hconnector^{1,k-1}\) along with \(p_a(b^{1,\rightarrow})\) and \(p_a(b^{1,\leftarrow})\), from the rest of the terminals. The added part connects all the terminals in the first row to \(r\). On the other hand, \(\rcons\) should be in the same connected component as \(r\) and so the same connected component as the terminals that does not lie within the first row. Thus, all terminals should be connected to \(r\) in \(H'_{R,a}\) which means that \(H'_{R,a}\) is connected and contains all the terminals in \(H_{R,a}\) and \(r\). So, replacing \(H_{R,a}\) with \(H'_{R,a}\) in the solution results in a smaller cost, which contradicts the cost-minimality assumption on the considered solution. 
            \item \textbf{\(r\rselect \notin E(H_{R,a})\).} Since the pendant vertices of \(b^{\uparrow,j}\) are not in \(H_{R,a}\), neither \(b^{\uparrow,j}\) nor its pendant vertices can be a leaf of \(H_{R,a}\) due to cost-minimality. Therefore, there should be some \(y\in [\colNum]\) such that the edge \(b^{\uparrow,j}u^{1,j}_{1,y}\) is in \(H_{R,a}\). We construct a subgraph \(H'_{R,a}\) as follows. Let \(H'_{R,a}\) initially be the same as \(H_{R,a}\) which we modify it as follows, step by step. 
            \begin{enumerate}
                \item All edges between \(b^{\uparrow,j}\) and the vertices of \(G^{1,j}\) are removed. Recall that each of these edges has length \(\rootLength\).
                \item All edges of the first row path of \(G^{1,j}\) are added. Recall a row path has length \(\frac{\connectLength}{\rowNum}\).
                \item If the connected component that contains \(b^{\uparrow,j}\) does not intersect \(T_{R,a} \cup \{r\}\), remove it from \(H'_{R,a}\).
            \end{enumerate}
            This increases the cost by at most \(\frac{\connectLength}{\rowNum}-\rootLength\). 
            
            If \(H'_{R,a}\) is connected after this modification, then replacing \(H_{R,a}\) by \(H'_{R,a}\) results in a solution with smaller cost since \(\rootLength=2\connectLength\).
            This contradicts the cost-minimality assumption on the solution. So, there should be exactly two connected components in \(H'_{R,a}\) that arise from losing connectivity of the vertices of the first row path in \(G^{1,j}\) and the vertex \(b^{\uparrow,j}\). In this case, we continue the modification as follows. 
             Noting that \(b^{\uparrow,j}\) is connected to at least one terminal in \(T_{R,a}\) or \(r\) in \(H'_{R,a}\), the connected component that contains \(b^{\uparrow,j}\) must cover the edge \(\rselect b^{\uparrow,j}\) since this is the only edge that is incident on \(b^{\uparrow,j}\) that is not removed from \(H'_{R,a}\). This shows that \(b^{\uparrow,j}\) is a leaf of \(H'_{R,a}\). Since this vertex is not a terminal, we remove it from \(H'_{R,a}\) along with its edge \(\rselect b^{\uparrow,j}\), which reduces the cost by at least \(\srootLength\). Note that the connected component that contains \(\rselect\) still has a terminal of \(H_{R,a}\) or \(r\) as well. Thus, this procedure results in a subgraph that still covers \(T_{R,a} \cup \{r\}\) and has exactly two connected components, each having at least two vertices. Hence by \Cref{claim:MakeConnectedSubgraph}, there exists a tree \(H''_{R,a}\) where \(\dig(H''_{R,a}\setminus H'_{R,a})+\ell(H''_{R,a}\setminus H'_{R,a})-\ell(H'_{R,a}\setminus H''_{R,a}) \leq 5\connectLength\). Therefore, by replacing \(H_{R,a}\) in the minimum solution considered with \(H''_{R,a}\), we can construct a solution where the cost increases by at most \(\frac{\connectLength}{\rowNum}+5\connectLength-\rootLength-\srootLength<0\) recalling that \(\rootLength=2\connectLength\) and \(\srootLength=5\connectLength\). This contradicts the cost-minimality assumption on the considered solution.
        \end{itemize}

        Now we show that for all \(j \in [k]\), the vertex \(b^{\downarrow,j}\) is not in \(H_{R,a}\); the argument for \(b^{\uparrow,j}\) is symmetric. Assume towards a contradiction that this is not the case and so, for some \(j\in[k]\), \(b^{\downarrow,j}\) is in \(H_{R,a}\). Then, since no terminal of \(T_{R,a}\) is connected to \(b^{\downarrow,j}\) and \(b^{\downarrow,j}\) is only connected to the vertices of the last row path of \(G^{k,j}\), we can remove \(b^{\downarrow,j}\) along with all incident edges of \(H_{R,a}\) and add the last row path of \(G^{k,j}\) to construct a subgraph \(H'_{R,a}\). Clearly, this new subgraph is connected and still covers all the terminals in \(T_{R,a}\). However, through this modification, the cost is increased by at most \(\frac{\connectLength}{\rowNum}-\rootLength<0\), contradicting the cost-minimality of the solution.
    \end{proof}

    The following claim restricts the number of edges that appear in the solution from each level of a connector. This restriction is useful in providing a lower bound for the cost that arises from each connector. In more detail, this restriction not only provides a precise value for the cost arising from the linking paths of the connectors, but forces the solution to include all supporting paths, resulting in the lower bound presented in \Cref{claim:connectorstructure}. From now on, we refer to the union of the solution trees \(H_{R,\RstartIndex}\cup H_{R,\RendIndex}\cup H_{C,\CstartIndex} \cup H_{C,\CendIndex}\) by \(H_{\mathcal{S}}\).
    
    \begin{claim}
        \label{claim:SingleEdgeLevelConnector}
        For each \(q \in [c+1]\) and each connector \(\connector\), the union of solution trees \(H_{\mathcal{S}}\) contains exactly one edge from the \(q\)-th level of \(\connector\).
    \end{claim}
    \begin{proof}
        By \Cref{claim:EmptyConnectors}, there is no edge of \(H_{R,\RstartIndex}\) or \(H_{R,\RendIndex}\) in a column connector, and there is no edge of \(H_{C,\CstartIndex}\) or \(H_{C,\CendIndex}\) in a row connector. Therefore, to explore the edges of \(H_{\mathcal{S}}=H_{R,\RstartIndex}\cup H_{R,\RendIndex}\cup H_{C,\CstartIndex} \cup H_{C,\CendIndex}\) in a column connector, it suffices to consider the edges of \(H_{R,\RstartIndex}\) and \(H_{R,\RendIndex}\).
        Similarly, to explore the edges of \(H_{\mathcal{S}}\) in a row connector, it suffices to consider the edges of \(H_{C,\CstartIndex}\) and \(H_{C,\CendIndex}\). On the other hand, considering both \Cref{claim:EmptyConnectors} and \Cref{claim:EmptyRootPart} and noting that there should be a path from \(b^{i,\rightarrow}\) to \(b^{i,\leftarrow}\) for each \(i \in [k]\) in each of the trees \(H_{R,\RstartIndex}\) and \(H_{R,\RendIndex}\), each of these trees must contain at least one edge in each level of each row connector, and noting that there should be a path from \(b^{\uparrow,j}\) to \(b^{\downarrow,j}\) for each \(j \in [k]\) in each of trees \(H_{C,\CstartIndex}\) and \(H_{C,\CendIndex}\), each of these trees must contain at least one edge from each level of each column connector.
        
        To complete the proof, We demonstrate that there can be no more than one edge from each level of each connector. Assume for contradiction that for a connector \(\connector\) and a value of \(q \in [c+1]\), the subgraph \(H_{\mathcal{S}}\) contains more than one edge of the \(q\)-th level of this connector. First, we discuss the case where \(\connector\) is a row connector. I.e.\ let \(\connector = \hconnector^{i,j}\) for some \(i \in [k]\) and \(j\in [k-1]\).  We prove the claim considering three subcases.
        
        \begin{enumerate}
            \item \textbf{\(q\in\{2,\ldots,c\}\).} Construct two subgraphs \(H'_{R,\RstartIndex}\) and \(H'_{R,\RendIndex}\) as follows. Let them initially be the same as \(H_{R,\RstartIndex}\) and \(H_{R,\RendIndex}\), respectively. We modify them as follows.
            \begin{enumerate}
                \item We remove all edges of the \(q\)-th level  of \(\connector\) from each of these subgraphs. Noting that only \(H_{R,\RstartIndex}\) and \(H_{R,\RendIndex}\) intersect this connector and recall that each of these edges has a cost \(\connectLength+\frac{\connectLength}{c+1}\), this step decreases the cost by at least \(2(\connectLength+\frac{\connectLength}{c+1})\).
                \item We add one of the edges of the \(q\)-th level to both of subgraphs. This edge is of length \(\connectLength\) and digging cost \(\frac{\connectLength}{c+1}\).
                \item We add the interior of the \(q\)-th and the interior of the \((q+1)\)-th supporting paths of \(\connector\) to both \(H'_{R,\RstartIndex}\) and \(H'_{R,\RendIndex}\). Recall that each edge of a supporting path is of cost \(\frac{\connectLength}{4(c+1)(\rowNum+1)}\) and the interior contains \(\rowNum+1-2\) such edges.
            \end{enumerate}
            Since each edge in the \(q\)-th level of \(\connector\) makes a connection between a vertex of the interior of the \(q\)-th supporting path and a vertex of the interior of the \((q+1)\)-th supporting path, both \(H'_{R,\RstartIndex}\) and \(H'_{R,\RendIndex}\) are connected after this procedure. If either of \(H'_{R,\RstartIndex}\) and \(H'_{R,\RendIndex}\) contains a cycle, we consider a spanning tree of that subgraph. Now, if we replace \(H_{R,\RstartIndex}\) and \(H_{R,\RendIndex}\) by \(H'_{R,\RstartIndex}\) and \(H'_{R,\RendIndex}\) in the initial minimum-cost solution, the cost increases by at most \(2\connectLength+\frac{\connectLength}{c+1}+4(\rowNum-1)\cdot \frac{\connectLength}{4(c+1)(\rowNum+1)}-2(\connectLength+\frac{\connectLength}{c+1})<0\) which contradicts the cost-minimality assumption on the initial solution. 
            \item \textbf{\(q=1\).}  From the previous case, we know that for each \(q \in \{2,\ldots,k-1\}\), subgraph \(H_{R,\RstartIndex} \cup H_{R,\RendIndex}\) contains at most one edge from the \(q\)-th level of \(\connector\). Let \(h^{i,j}_{1,x}h^{i,j}_{2,x}\) be the edge from the second level that is in \(H_{R,\RstartIndex} \cup H_{R,\RendIndex}\) where \(x\in [\rowNum]\) which we know to be unique from the previous case. We construct \(H'_{R,\RstartIndex}\) and \(H'_{R,\RendIndex}\) from \(H_{R,\RstartIndex}\) and \(H_{R,\RendIndex}\) with following modifications.
            \begin{enumerate}
                \item  We remove all edges of the first level of \(\connector\) from both of these subgraphs, where each of these edges has length \(\connectLength\) and digging cost \(\frac{\connectLength}{c+1}\) and
                \item We add \(u^{i,j}_{y,\colNum}h^{i,j}_{1,x}\) which is of length \(\connectLength\) and digging cost \(\frac{\connectLength}{c+1}\),
                \item We add the last column path of \(G^{i,j}\) and the interior of the first supporting path of the connector to both; recall that the length of each column path is at most \(\frac{\connectLength}{40\colNum}\), each edge of a supporting path is of cost \(\frac{\connectLength}{4(c+1)(\rowNum+1)}\), and the interior of a supporting path contains \(\rowNum+1-2\) such edges.
            \end{enumerate}
            Since the edges in the first level make a connection between the vertices of the last column path of \(G^{i,j}\) and the vertices of the first supporting path of \(\connector\), and considering that the first supporting path is already in \(H_{R,\RstartIndex} \cup H_{R,\RendIndex}\) and hence also in \(H'_{R,\RstartIndex} \cup H'_{R,\RendIndex}\), the added edges connect all the vertices of the last column path of \(G^{i,j}\) to \(h^{i,j}_{1,x}\) in both \(H'_{R,\RstartIndex}\) and \(H'_{R,\RendIndex}\).
            This implies the connectivity of both of these subgraphs. Moreover, clearly, \(H'_{R,\RstartIndex}\) still covers all terminals covered  by \(H_{R,\RstartIndex}\) and \(H'_{R,\RendIndex}\) still covers all terminals covered by \(H_{R,\RendIndex}\). If we replace \(H_{R,\RstartIndex}\) and \(H_{R,\RendIndex}\) by \(H'_{R,\RstartIndex}\) and \(H'_{R,\RendIndex}\), this modification increases the cost of the solution by at most \(2\connectLength+\frac{\connectLength}{c+1}+2\cdot\frac{\connectLength}{40\colNum}+2(\rowNum-1)\frac{\connectLength}{4(\rowNum+1)(c+1)}-2(\connectLength+\frac{\connectLength}{c+1}) <  0\) which contradicts the cost-minimality assumption on the considered solution.
            \item \textbf{\(q=c+1\).} This is symmetric to the previous case. Using a similar argument considering the last supporting path of \(\connector\) and the first column path of \(G^{i,j+1}\), we can show that in the last level of \(\connector\), there cannot be more than one edge of \(H_{R,\RstartIndex}\cup H_{R,\RendIndex}\).
        \end{enumerate}
        In each of these three cases, we have proved that there is at most one edge of \(H_{R,\RstartIndex}\cup H_{R,\RendIndex}\) in each level of \(\connector\).
        
        The proof to show that \(H_{C,\CstartIndex}\cup H_{C,\CendIndex}\) contains at most one edge from each level of a column connector is similar to the analogue statement for row connectors;
        instead of column paths, we consider the row paths of adjacent cell gadgets. Recalling that \(H_{R,\RstartIndex}\) and \(H_{R,\RendIndex}\) contain at least one edge from each level of each row connector and \(H_{C,\CstartIndex}\) and \(H_{C,\CendIndex}\) contain at least one edge from each level of each column connector, we can conclude that \(H_{\mathcal{S}}\) contains exactly one edge from each level of each connector.
    \end{proof}

    Knowing which connectors and root parts each tree in a minimum-cost solution intersects, and that exactly one edge from each level of a connector appears in \(H_{\mathcal{S}}\), we can now obtain a lower bound for the cost that arises from the part of a solution within each connector. \Cref{claim:connectorstructure} presents this lower bound and a property that must hold in the case that this lower bound is tight, i.e.\ the cost that arises from the part of a solution in a connector is precisely the lower bound.
    
    \begin{claim} 
        \label{claim:connectorstructure}
        For each connector \(\connector\), the value of \(\text{cost}_{\connector}(\mathcal{S})\) is at least \((2c+3+\frac{c}{4(c+1)} )\cdot \connectLength\).
        If \(\text{cost}_{\connector}(\mathcal{S}) = (2c+3+\frac{c}{4(c+1)} )\cdot \connectLength\), then
        \begin{enumerate}
            \item \(\connector\) is a row connector \(\hconnector^{i,j}\) and \(H_{R,\RstartIndex} \cap H_{R,\RendIndex} \cap \hconnector^{i,j}\) is a single linking path of \(\hconnector^{i,j}\), or
            \item \(\connector\) is a column connector \(\vconnector^{i,j}\) and \(H_{C,\CstartIndex} \cap H_{C,\CendIndex} \cap \vconnector^{i,j}\) is a single linking path of \(\vconnector^{i,j}\).
        \end{enumerate}
    \end{claim}
    \begin{proof}
 Recall that by \Cref{claim:EmptyConnectors} there is no edge of \(H_{C,\RstartIndex}\) and  \(H_{C,\CendIndex}\) in any row connector and there is no edge of \(H_{R,\RstartIndex}\) and  \(H_{R,\RendIndex}\) in any column connector.
 By \Cref{claim:SingleEdgeLevelConnector}, there is exactly one edge from each level of a row connector in \(H_{R,\RstartIndex}\cup H_{R,\RendIndex}\) and exactly one edge from each level of a column connector in \(H_{C,\CstartIndex}\cup H_{C,\CendIndex}\).
 Because, by construction, there must be a path from \(b^{i,\rightarrow}\) to \(b^{i,\leftarrow}\) for each \(i \in [k]\) in each of the trees \(H_{R,\RstartIndex}\) and \(H_{R,\RendIndex}\), both trees must contain the single edge of \(H_{\mathcal{S}}\) from each level of a row connector. With a similar argument, we can conclude that the single edge of \(H_{\mathcal{S}}\) in each level of a column connector must appear in both \(H_{C,\CstartIndex}\) and \(H_{C,\CendIndex}\). Hence, to cover all the terminals of \(\connector\), all supporting paths of \(\connector\) should be in \(H_{\mathcal{S}} \cap \connector\). Therefore, the total length cost arising from \(\mathcal{S}\) in \(\connector\) is at least \(2(c+1)\cdot \connectLength+c\cdot\frac{\connectLength}{4(c+1)}=(2c+2+\frac{c}{4(c+1)} )\cdot \connectLength\) and the total digging cost is at least \(\connectLength\) arising from single edges of each level of \(\connector\). This implies that \((2c+3+\frac{c}{4(c+1)} )\cdot \connectLength\) is a lower bound for \(\text{cost}_{\connector}(\mathcal{S})\).
        
          This lower bound is tight if each edge of the supporting paths is used exactly once, which implies that the single edges of different levels must be from the same linking path of \(\connector\) and thus, if \(\connector\) is a row connector \(\hconnector^{i,j}\), \(H_{R,\RstartIndex} \cap H_{R,\RendIndex} \cap \hconnector^{i,j}\) is a single linking path of \(\hconnector^{i,j}\), and if \(\connector\) is a column connector \(\hconnector^{i,j}\), \(H_{C,\CstartIndex} \cap H_{C,\CendIndex} \cap \vconnector^{i,j}\) is a single linking path of the connector \(\vconnector^{i,j}\).
    \end{proof}

    We now turn our attention from the behavior of minimum-cost solutions in connectors to that in cell gadgets.
    In our analysis the following notation will be useful:
    For a subgraph \(H\) of \(G\), we denote the set of horizontal edges of \(G^{i,j}\) that are in \(H\) by \(\hedgesetCell[i]{j}{H}\) and the set of vertical edges of \(G^{i,j}\) that are also in this subgraph by \(\vedgesetCell[i]{j}{H}\). In addition, we define \(\hedgeset[H]=\bigcup_{i,j\in[k]}\hedgesetCell[i]{j}{H}\) and \(\vedgeset[H]=\bigcup_{i,j\in[k]}\vedgesetCell[i]{j}{H}\) as the horizontal and vertical edges of \(H\).

    Knowing that \(H_{R,\RstartIndex}\) and \(H_{R,\RendIndex}\) only intersect row connectors and \(H_{C,\CstartIndex}\) and \(H_{C,\CendIndex}\) only intersect column connectors, we can observe that by the graph structure, cell gadgets play an important role in the connectivity of each tree in \(\mathcal{S}\). With this information, the following claim demonstrates a lower bound for the cost of the part of \(H_{\mathcal{S}}\) in each cell gadget. Additionally, if this lower bound is tight, some structural properties are presented in this claim for the part of each tree that lies within a cell gadget.
    
    \begin{claim}
        \label{claim:cellgadgetstructure}
        For each \(i,j\in[k]\), let \(L_{i,j}=\ell(\hedgesetCell[i]{j}{H_{C,\CstartIndex}})+\ell(\hedgesetCell[i]{j}{H_{C,\CendIndex}})\). The value of \(\text{cost}_{G^{i,j}}(\mathcal{S})  \) is at least
            \[\underbrace{2(\colNum-1-\regpar)\ell_H+2\regpar+L_{i,j}}_{\text{length of horizontal edges}}+\underbrace{(\colNum-1)d_H}_{\text{digging cost of horizontal edges}}+\underbrace{2(\rowNum-1)\ell_V}_{\text{length of vertical edges}}.\]
        If equality is attained, then, 
        \begin{enumerate}
            \item \(H_{R,\RstartIndex}\cap G^{i,j}\) is equal to \(H_{R,\RendIndex}\cap G^{i,j}\) and is a row path of \(G^{i,j}\), and \label{cellgadgetstructureCond1}
            \item \(H_{C,\CendIndex}\cap G^{i,j}\) is equal to \(H_{C,\CendIndex} \cap G^{i,j}\), \label{cellgadgetstructureCond2}
            \item all the edges in \(\hedgesetCell[i]{j}{H_{C,\CstartIndex}}\) are also in \(H_{R,\RstartIndex}\), and \label{cellgadgetstructureCond3}
            \item there are exactly \(n-1\) edges in \(\vedgesetCell[i]{j}{H_{C,\CstartIndex}}\) where for each \(x \in [n-1]\), there is exactly one edge of \(\vedgesetCell[i]{j}{H_{C,\CstartIndex}}\) that connects a vertex of \(x\)-th row path to a vertex of \((x+1)\)-th row path. \label{cellgadgetstructureCond4}
        \end{enumerate}
    \end{claim}
    \begin{proof}
        Considering \Cref{claim:EmptyConnectors} and \Cref{claim:EmptyRootPart} and noting that there should be a path from \(b^{i,\rightarrow}\) to \(b^{i,\leftarrow}\) for each \(i \in [k]\) in each of trees \(H_{R,\RstartIndex}\) and \(H_{R,\RendIndex}\), each of these trees must contain a path \(P\) from one of the vertices of the first column path of \(G^{i,j}\) to a vertex of the last column path of \(G^{i,j}\). So, for each \(y\in[\colNum-1]\),  the path \(P\) contains at least one edge among the horizontal edges between the vertices of \(y\)-th and \((y+1)\)-th column paths of \(G^{i,j}\). For each \(y\in[\colNum-1]\) if the edge of \(P\) between the vertices of \(y\)-th and \((y+1)\)-th column paths and the edge of \(P\) between the vertices of \((y+1)\)-th and \((y+2)\)-th column paths of \(G^{i,j}\) are from different row paths, there should be at least one edge of \((y+1)\)-th column path of \(G^{i,j}\) in \(P\). Noting that  by construction, each row path contains \(\regpar\) edges of length \(1\), if \(P\) contains horizontal edges from \(x\) different row paths, \(\ell(P)\) is at least 
        \begin{align*}
        &x\cdot \regpar+(\colNum-1-x\cdot \regpar)\ell_H+(x-1)\ell_V\\
        &=(\colNum-1)\ell_H+x\cdot (\regpar-\regpar\cdot \ell_H+\ell_V)-\ell_V\\
        & \geq (\colNum-1)\ell_H+ (\regpar-\regpar\cdot \ell_H+\ell_V)-\ell_V \\
        &=(\colNum-1-\regpar)\ell_H+\regpar
        \end{align*}
        since 
        \begin{align*}
            \regpar+\ell_V-\regpar\cdot \ell_H&\geq \regpar+\ell_V-n\cdot \ell_H\\&=\regpar+\colNum-\colNum=\regpar >0.
        \end{align*}
        As there are two trees, namely, \(H_{R,\RstartIndex}\) and \(H_{R,\RendIndex}\) that contain such a path, this adds a value of at least \(2(\colNum-1-\regpar)\ell_H+2\regpar\) to the cost.
        Again, considering \Cref{claim:EmptyConnectors} and \Cref{claim:EmptyRootPart} and noting that there should be a path from \(b^{\uparrow,j}\) to \(b^{\downarrow,j}\) for each \(j \in [k]\) in each of trees \(H_{C,\CstartIndex}\) and \(H_{C,\CendIndex}\), each of these trees must contain a path \(P'\) from one of the vertices of the first row path of \(G^{i,j}\) to a vertex of the last row path of \(G^{i,j}\). Thus, for each \(y\in[\colNum-1]\), the path \(P'\) contains at least one edge among the vertical edges between the vertices of \(y\)-th and \((y+1)\)-th row paths of \(G^{i,j}\) and hence, \(\ell(\vedgesetCell[i]{j}{P'})\) is at least \((\rowNum-1)\ell_V\). Since such a path appears in both \(H_{C,\CstartIndex}\) and \(H_{C,\CendIndex}\), this adds a value of at least \(2(\rowNum-1)\ell_V+L_{i,j}\) to the cost.

        To compute the total digging cost of \(\mathcal{S}\) in \(G^{i,j}\), assume that we first pay for the digging cost of \((H_{R,\RstartIndex}\cup  H_{R,\RendIndex})\) and then pay the additional cost for \(H_{C,\CstartIndex} \cup H_{C,\CendIndex}\). Since \(P\) contains at least \(\colNum-1\) horizontal edges, the total digging cost of the edges of \(P\) is at least \((\colNum-1) d_H\). So, the total digging cost of of \(\mathcal{S}\) in \(G^{i,j}\) is also at least \((\colNum-1) d_H\). Therefore, the value of \(\text{cost}_{G^{i,j}}(\mathcal{S})\), which is the summation of total length and total digging cost of \(\mathcal{S}\) in \(G^{i,j}\) is at least \[\underbrace{2(\colNum-1-\regpar)\ell_H+2\regpar}_{{\text{length of two copies of } P}}+\underbrace{2(\rowNum-1)\ell_V+L_{i,j}}_{{\text{length of two copies of } P'}}+\underbrace{(\colNum-1)d_H}_{\text{digging cost}}.\]
        
        Clearly, if \(H_{\mathcal{S}} \cap G^{i,j}\) contains any edges other than the edges of \(P\) and \(P'\), the cost must be more than \(2(\colNum-1-\regpar)\ell_H+2\regpar+2(\rowNum-1)\ell_V+L_{i,j}+(\colNum-1)d_H\). So, we can conclude that if the total cost is \(2(\colNum-1-\regpar)\ell_H+2\regpar+2(\rowNum-1)\ell_V+L_{i,j}+(\colNum-1)d_H\), we must have \(H_{R,\RstartIndex}\cap G^{i,j}=H_{R,\RendIndex}\cap G^{i,j}=P\) \(H_{C,\CstartIndex}\cap G^{i,j}=H_{C,\CendIndex}\cap G^{i,j}=P'\) which implies a part of property~\ref{cellgadgetstructureCond1}, property~\ref{cellgadgetstructureCond2}, and property~\ref{cellgadgetstructureCond4} of the claim. On the other hand, \(\ell(P)\) is at least \(x\cdot \regpar+(\colNum-1-x\cdot \regpar)\ell_H+(x-1)\ell_V\) which must be exactly \((\colNum-1-\regpar)\ell_H+\regpar\) and so \(x=1\). This means that all edges of \(P\) are from a single row path of \(\connector\) and property~\ref{cellgadgetstructureCond1} of the claim holds. Moreover, since we only computed the digging cost of \(P\), to achieve this total cost, the digging cost of \(P'\) should be \(0\). Since the digging cost of none of the horizontal edges is  \(0\), all horizontal edges of \(P'\) must be shared with \(P\) and thus, property~\ref{cellgadgetstructureCond3} holds as well. So, if the value \(2(\colNum-1-\regpar)\ell_H+2\regpar+2(\rowNum-1)\ell_V+L_{i,j}+(\colNum-1)d_H\) is tight, all three mentioned properties hold.
    \end{proof}

    To cover terminals connected to \(b^{\downarrow,j}\) and considering \Cref{claim:EmptyConnectors} and \Cref{claim:EmptyRootPart}, for each \(j\in[k]\), there must be a path \(P_j\) from \(b^{\uparrow,j}\) to \(b^{\downarrow,j}\) in trees \(H_{C,\CstartIndex}\) and \(H_{C,\CendIndex}\). Taking this into account, \Cref{claim:tightcosts} explores a minimum solution and proves that the lower bounds provided in \Cref{claim:connectorstructure} and \Cref{claim:cellgadgetstructure} must be tight in all connectors and cell gadgets in a minimum-cost solution. Then, using these claims, it provides some additional properties of a minimum solution that help us to construct a solution for the given \GTRegular instance. Specifically, they allow us to fix the second value of the chosen pair from each cell \(S_{i,j}\).
    
    \begin{claim}
        \label{claim:tightcosts}
        If \(\mathcal{S}\) is a solution with minimum cost, then we have following facts.
        \begin{enumerate}
            \item \label{tightcostsCond1} For each \(i,j \in[k]\), let \(L_{i,j}=\ell(\hedgesetCell[i]{j}{H_{C,\CstartIndex}})+\ell(\hedgesetCell[i]{j}{H_{C,\CendIndex}})\). The value of \(\text{cost}_{G^{i,j}}(\mathcal{S})\) is \(2(\colNum-1-\regpar)\ell_H+2\regpar+2(\rowNum-1)\ell_V+(\colNum-1)d_H+L_{i,j},\)
            \item \label{tightcostsCond2} For each connector \(\connector\), the value of \(\text{cost}_{\connector}(\mathcal{S})\) is \((2c+3+\frac{c}{4(c+1)} )\cdot \connectLength\), 
            \item \label{tightcostsCond3} For each \(j \in [k]\), there exists some value \(y_j\) such that the edges \(b^{\uparrow,j}u^{1,j}_{1,(y_j-1)(k+1)+1}\) and \(u^{k,j}_{n,y_j(k+1)}b^{\downarrow,j}\) are in both \(H_{C,\CstartIndex}\) and \(H_{C,\CendIndex}\), and
            \item \label{tightcostsCond4} For each \(j\in [k]\), \(\bigcup_{i=1}^{k}\hedgesetCell[i]{j}{H_{C,\CstartIndex}}= \bigcup_{i=1}^{k}\hedgesetCell[i]{j}{H_{C,\CendIndex}}\) is a set of \(k\) edges of length \(1\).
        \end{enumerate}
    \end{claim}
    \begin{proof}
    If we exclude the first and last edges of \(P_j\), what remains is a path from \(u^{1,j}_{1,y\cdot(k+1)+1}\) to \(u^{k,j}_{n,y'\cdot(k+1)}\) for some \(0 \leq y < n\) and \(1 \leq y' \leq n\). If \(y \neq y'\), then \(P_j\) must have at least \(k\) horizontal edges within cell gadgets and the edge sets of supporting paths in column connectors.
    This implies that the total length of horizontal edges is at least \(k\) since all such edges have length at least \(1\). If \(y=y'\), it must contain at least one edge among the horizontal edges between the vertices of \((y(k+1))\)-th column path and the vertices of \((y(k+1)+1)\)-th column path of a cell gadget and the edges of supporting paths in column connectors. Such an edge has length at least \(\ell_H\) since the horizontal edges between the vertices of \((y(k+1))\)-th column path and the vertices of \((y(k+1)+1)\)-th column path of a cell gadget are of length \(\ell_H\) and the edges of supporting paths are of length \(\frac{\connectLength}{4(c+1)(\colNum+1)}>\ell_H\). Therefore, the total length of horizontal edges in each of \(H_{C,\CstartIndex}\) and \(H_{C,\CendIndex}\) is at least \(k\cdot k\). Therefore, by \Cref{claim:cellgadgetstructure}, the total cost arising from all cell gadgets is at least 
    \begin{align*}
    C_{\text{cells}}=&2k^2((\colNum-1-\regpar)\ell_H+\regpar+(\rowNum-1)\ell_V+(\colNum-1)d_H)+\sum\limits_{i,j\in[k]}L_{i,j}\\
    &=2k^2((\colNum-1-\regpar)\ell_H+\regpar+(\rowNum-1)\ell_V+(\colNum-1)d_H)+2k^2\\
    &=2k^2((\colNum-1-\regpar)\ell_H+\regpar+(\rowNum-1)\ell_V+(\colNum-1)d_H+1).
    \end{align*} 
    
    Furthermore, by \Cref{claim:EmptyRootPart}, for each \(i \in [k]\), the only way to connect the root \(r\) to the terminal \(b^{i,\leftarrow}\) is through \(\rcons b^{i,\rightarrow}\), an edge between \(b^{i,\rightarrow}\) and the vertices of the first column path of \(G^{i,1}\), and an edge between the vertices of the last column path of \(G^{i,k}\) and \(b^{i,\leftarrow}\). Also, for each \(j \in [k]\), the only way to connect the root \(r\) to the terminal \(b^{\downarrow,j}\) is through the edge \(\rselect b^{\uparrow,j}\), an edge between \(b^{\uparrow,j}\) and the vertices of the first row path of \(G^{1,j}\), and an edge between the vertices of the last row path of \(G^{k,j}\) and \(b^{\downarrow,j}\). Since each of these edges must appear in two trees, this costs in total \(8k\cdot \rootLength+4k\cdot \srootLength\). Moreover, for these connections we need to have \(r\rcons\) in both \(H_{R,\RstartIndex}\) and \(H_{R,\RendIndex}\) and \(r\rselect\) in both \(H_{C,\CstartIndex}\) and \(H_{C,\CendIndex}\). In addition, for connection to the pendant terminals of \(\rcons\), \(b^{i,\rightarrow}\), \(b^{i,\leftarrow}\), \(\rselect\), \(b^{\uparrow,j}\), and \(b^{\downarrow,j}\) we need a length edge \(\termLength\) in each tree. So, the total cost arising from the root part is at least \(C_{\text{root part}}=8k\cdot \rootLength+4k\cdot \srootLength+(8+8k)\termLength\).
    
     So, the total cost within the connectors must be at most 
    \(
    \beta-C_{\text{cells}}-C_{\text{root part}}=2k(k-1)\cdot(2c+3+\frac{c}{4(c+1)} )\cdot \connectLength.
    \)
    Since we have \(2k(k-1)\) connectors and considering \Cref{claim:connectorstructure}, the cost within each connector must be exactly \((2c+3+\frac{c}{4(c+1)} )\cdot \connectLength\). This shows that the total cost within the connectors must be exactly \(2k(k-1)\cdot(2c+3+\frac{c}{4(c+1)} )\cdot \connectLength\). 
     As a result, the lower bound \(C_{\text{cells}}\) for the total cost of parts in the cell gadgets is also tight, which implies that the total cost of the horizontal edges in each of \(H_{C,\CstartIndex}\) and \(H_{C,\CendIndex}\) must be exactly \(k^2\) and so the total cost of the edges in \(\bigcup_{i=1}^{k}\hedgesetCell[i]{j}{H_{C,\CstartIndex}}=\bigcup_{i=1}^{k}\hedgesetCell[i]{j}{H_{C,\CendIndex}}\) for each \(j\in [k]\) must be precisely \(k\). As we explained, this is possible only if there are exactly \(k\) edges of length \(1\) and so \(y_j=y'=y+1\) since if \(y+1<y'\), by graph construction, there would be at least one edge between the vertices of the \(((y+1)(k+1))\)-th column path and the vertices of the \(((y+1)(k+1)+1)\)-th column path of a cell gadget that appears in the horizontal edges of \(P_j\) which has length \(k+1\).
\end{proof}

Note that the value provided in property~\ref{tightcostsCond1} of \Cref{claim:tightcosts} is equal to the minimum cost of a cell gadget presented in \Cref{claim:cellgadgetstructure} which shows that this cost is tight for each cell gadget in a minimum-cost solution. Moreover, the value provided in property~\ref{tightcostsCond2} is equal to the minimum cost of a connector presented in \Cref{claim:connectorstructure} yielding that this cost is tight for each connector in a minimum-cost solution.

To cover the terminals connected to \(b^{i,\leftarrow}\) and considering \Cref{claim:EmptyConnectors} and \Cref{claim:EmptyRootPart}, for each \(i\in[k]\), there must be a path \(P'_i\) from \(b^{i,\rightarrow}\) to \(b^{i,\leftarrow}\) in trees \(H_{R,\RstartIndex}\) and \(H_{R,\RendIndex}\). Taking into account this and knowing that the lower bounds provided by \Cref{claim:connectorstructure} and \Cref{claim:cellgadgetstructure} are tight, we introduce some new properties for the parts of the solution that lie within a cell gadget and the parts of the solution that lie within a connector, which guide us to fix the first value of the pairs chosen from each cell \(S_{i,j}\) for a solution to the given \GTRegular instance.

\begin{claim}
    \label{claim:rowNumber}
    For each \(i\in[k]\), there exists a value \(x_i\) such that 
    \begin{enumerate}
        \item for each \(j\in[k]\), the subgraph \(H_{R,\RstartIndex}\cap G^{i,j}\) is the \(x_i\)-th row path of \(G^{i,j}\), and
        \item for each \(j\in[k-1]\), the subgraph \(H_{R,\RstartIndex}\cap H_{R,\RendIndex} \cap \hconnector^{i,j}\) is the \(x_i\)-th linking path of \(\hconnector^{i,j}\).
    \end{enumerate} 
\end{claim}
     \begin{proof}
    Since by \Cref{claim:tightcosts}, the lower bound for the cost arising from each cell gadget must be tight to achieve the total cost \(\beta\) and considering \Cref{claim:cellgadgetstructure}, \(H_{R,\RstartIndex}\cap G^{i,j}=H_{R,\RendIndex}\cap G^{i,j}\) must be a single row path which is part of \(P'_i\). Furthermore, for each row connector \(\hconnector{i,j}\), as we explained, the cost must be exactly \((2c+3+\frac{c}{4(c+1)} )\cdot \connectLength\), so by \Cref{claim:connectorstructure}, \(H_{R,\RstartIndex}\cap H_{R,\RendIndex} \cap G^{i,j}\) must be a single linking path which is part of \(P'_i\). Recalling that the cost of solution within each cell gadget and the cost solution within each row connector is tight, the union of these single paths from the cell gadgets and connectors must be equal to \(P'_i\) without its first and last edge. Thus, the union of these single paths from the cell gadgets and the connectors must be equal to the interior of \(P'_i\). Thus, for each cell gadget \(G^{i,j}\), if this single path is the \(x_{i,j}\)-th row path, due to the connectivity of \(P'_i\), the single path from \(\hconnector^{i,j}\) and the single path from \(\hconnector^{i,j-1}\) (if \(j>1\)) must also be the \(x_{i,j}\)-th linking path of them. Therefore, we have \(x_{i,j}=x_{i,j'}\) for all \(j,j'\in[k]\), and so, for \(x_i=x_{i,1}\), the conditions of the claim hold.
     \end{proof}

    \Cref{claim:singleHorizontalEdge} states the last necessary property of a minimum solution, which we need to prove the existence of a pair with the first and second chosen values in each \(S_{i,j}\) of the given \GTRegular instance.
    
     \begin{claim}
     \label{claim:singleHorizontalEdge}
         For each \(i,j\in [k]\), there is exactly one horizontal edge in \((H_{C,\CstartIndex} \cup H_{C,\CendIndex}) \cap G^{i,j}\) and this edge has length \(1\).
     \end{claim}
     \begin{proof}
         By property~\ref{tightcostsCond3} of \Cref{claim:tightcosts}, we have \( H_{C,\CstartIndex}\cap G^{i,j} = H_{C,\CendIndex}\cap G^{i,j}\) and this is equal to \((H_{C,\CstartIndex} \cup H_{C,\CendIndex}) \cap G^{i,j}\). By property~\ref{tightcostsCond4} of the same Claim, \(\bigcup_{i=1}^{k}\hedgesetCell[i]{j}{H_{C,\CstartIndex}}=\bigcup_{i=1}^{k}\hedgesetCell[i]{j}{H_{C,\CendIndex}}\) is a set of edges of length \(1\). By \Cref{claim:cellgadgetstructure}, these edges must be shared with \(H_{R,\RstartIndex}\). Taking into account \Cref{claim:rowNumber}, for each \(i,j\in[k]\), subgraph \(H_{R,\RstartIndex}\cap G^{i,j}\) is exactly one row path, the \(x_i\)-th row path of \(G^{i,j}\). So, \(\hedgesetCell[i]{j}{H_{C,\CstartIndex}}\) must also belong to this row path. Taking into account property~\ref{cellgadgetstructureCond4} of \Cref{claim:cellgadgetstructure}, there are exactly \(n-1\) vertical edges of \(H_{C,\CstartIndex}\cap G^{i,j}\) where for each \(x' \in [n-1]\), there is exactly one edge of \(H_{C,\CstartIndex}\cap G^{i,j}\) that connects a vertex of the \(x'\)-th row path to a vertex of the \((x'+1)\)-th row path. Thus, we can conclude that \(\hedgesetCell[i]{j}{H_{C,\CstartIndex}}\) should be a subset of the edges of the \(x_i\)-th row path that appear consecutively in this path. Since by graph construction, the length of the edge appears after each edge of length \(1\) is \(\ell_H=k+1\) and we know that all the edges in \(\hedgesetCell[i]{j}{H_{C,\CstartIndex}}\) have length \(1\), there cannot be more than one horizontal edge in \(H_{C,\CstartIndex}\cap G^{i,j}=H_{C,\CendIndex}\cap G^{i,j}=(H_{C,\CstartIndex} \cup H_{C,\CendIndex}) \cap G^{i,j}\). On the other hand, since \(\bigcup_{i=1}^{k}\hedgesetCell[i]{j}{H_{C,\CstartIndex}}\) is a set of \(k\) edges and noting that there cannot be more than one of them in each \(G^{i,j}\), there must be exactly one of them in each of these cell gadgets and hence the Claim.
     \end{proof}

    Now, with the properties shown for a solution to the \ULVshort-instance, we construct the solution to the \GTRegular instance as follows. For each \(i,j \in [k]\), we choose the pair \((x_i,y_j)\) from \(S_{i,j}\) where \(x_i\) is the value found by \Cref{claim:rowNumber} and \(y_j\) is the value found by property~\ref{tightcostsCond3} of \Cref{claim:tightcosts}. We prove that such a pair exists in \(S_{i,j}\) by showing that the edge \(u^{i,j}_{x_i,(k+1)(y_j-1)+i}u^{i,j}_{x_i,(k+1)(y_j-1)+i+1}\) must have length \(1\). Considering that costs are tight due to \Cref{claim:tightcosts}, by \Cref{claim:cellgadgetstructure} and \Cref{claim:connectorstructure}, \(H_{C,\CstartIndex}\) and \(H_{C,\CendIndex}\) share all edges in connectors and cell gadgets except for the edges of the supporting paths in the connectors which are completely disjoint, \(P_j \cap G^{i,j}\) must be equal to \(H_{C,\CstartIndex}\cap G^{i,j}=H_{C,\CendIndex}\cap G^{i,j}\) and \(P_j \cap \vconnector^{i,j}\) must be equal to \(H_{C,\CstartIndex}\cap H_{C,\CendIndex}\cap \vconnector^{i,j}\).  By property~\ref{tightcostsCond3} of \Cref{claim:tightcosts}, if we exclude the first and last edges of \(P_j\), it is a path from \(u^{1,j}_{1,(y_j-1)(k+1)+1}\) to \(u^{k,j}_{n,y_j(k+1)}\). If we go through this path, the first vertex of the cell gadget \(G^{i,j}\) that we reach should be the vertex \(u^{i,j}_{1,(y_j-1)(k+1)+i}\). This is because we pass through exactly \(i-1\) horizontal edges in total in the previous cell gadgets, and by \Cref{claim:connectorstructure}, the parts in the column connectors are straight paths. Since by \Cref{claim:singleHorizontalEdge}, there is exactly one horizontal edge in \(P_j \cap G^{i,j}\) that by \Cref{claim:connectorstructure} is shared with \(H_{R,\RstartIndex}\) and so, by \Cref{claim:rowNumber}, it belongs to the \(x_i\)-th row path of \(G^{i,j}\), this edge must be \(u^{i,j}_{x_i,(y_j-1)(k+1)+i}u^{i,j}_{x_i,(y_j-1)(k+1)+i+1}\). Finally, by \Cref{claim:tightcosts}, we know that this edge must have length \(1\) and so, by graph construction, \((x_i,y_j)\) must be in \(S_{i,j}\). The part of the solution in a cell gadget is shown in \Cref{fig:cellGadget:b} for an instance of \(S_{i,j}\).

    We have proven that a \GTRegular-solution for \((S_{i,j})_{(i,j) \in [k]^2}\) induces a \ULVshort-solution for \(\mathcal{I}\).
    This concludes the proof of the correctness of our reduction.
    Recall that the number of terminals in our constructed \ULVshort-instance is \(\mathcal{O}(k^2)\) and that \(t=4\) is constant.
    Hence if there were \(f,f'\) such that \ULVshort can be solved in time  \(f(|T|,t)\cdot |V(G)|^{o(\sqrt{|T|})\cdot f'(t)}\) for any instance \(\mathcal{I}\) then we could solve the original \GTRegular instance in time \(\tilde{f}(k) \cdot n^{o(k)}\) which by \Cref{thm:GTregularW[1]} contradicts the ETH.
\end{proof}

%% file: figs/connector.tex
\begin{tikzpicture}[scale=0.6,every node/.style={ scale=0.85}]
\tikzstyle{vertex}=[circle, fill=black, inner sep=1.5pt, scale=0.85]
\tikzset{terminal/.style={fill=gray!30, draw=gray!50,  thick,minimum size=0pt, inner sep=2pt}}
% Parameters
\def\rows{6}
\def\cols{8}

% Draw vertices (no corners)
\foreach \i in {1,...,\rows}{
  \foreach \j in {1,...,\cols}{
    % Skip corners
    \ifnum\i=1
      \ifnum\j=1
      \else\ifnum\j=\cols
      \else
        \node[terminal] (v\i\j) at (\j,\rows-\i) {};
      \fi\fi
    \else
        \ifnum\i=\rows
            \ifnum\j=1
            \else
                \ifnum\j=\cols
                \else
                    \node[terminal] (v\i\j) at (\j,\rows-\i) {};
                \fi
            \fi
        \else
            \ifnum\j=1
                \ifnum\i>3
                    \pgfmathtruncatemacro{\newi}{6-\i}
                    \node[vertex,label=left:$p_{\newi,1}$] (v\i\j) at (\j,\rows-\i) {};
                \else
                    \ifnum\i=3
                        \node[vertex,label=left:$p_{|P|-1,1}$] (v\i\j) at (\j,\rows-\i) {};
                    \else
                        \node[vertex,label=left:$p_{|P|,1}$] (v\i\j) at (\j,\rows-\i) {};
                    \fi
                \fi
            \else
                \ifnum\j=8
                    \ifnum\i>3
                        \pgfmathtruncatemacro{\newi}{6-\i}
                        \node[vertex,label=right:$p_{\newi,2}$] (v\i\j) at (\j,\rows-\i) {};
                    \else
                        \ifnum\i=3
                            \node[vertex,label=right:$p_{|P|-1,2}$] (v\i\j) at (\j,\rows-\i) {};
                        \else
                            \node[vertex,label=right:$p_{|P|,2}$] (v\i\j) at (\j,\rows-\i) {};
                        \fi
                    \fi
                \else
                    \node[vertex] (v\i\j) at (\j,\rows-\i) {};
                \fi
            \fi
        \fi
    \fi
  }
}

% Draw horizontal edges (skip first/last column)
\foreach \i in {2,...,\numexpr\rows-1\relax}{
  \foreach \j in {1,...,\numexpr\cols-1\relax}{
    \pgfmathtruncatemacro{\jnext}{\j+1}
    \ifnum\j=3
        \draw[color=orange] (v\i\j) -- (v\i\jnext);
    \else
        \draw (v\i\j) -- (v\i\jnext);
    \fi
  }
}

% Draw vertical edges (skip first/last row)
\foreach \i in {1,...,\numexpr\rows-1\relax}{
  \foreach \j in {2,...,\numexpr\cols-1\relax}{
    \pgfmathtruncatemacro{\inext}{\i+1}
    \draw (v\i\j) -- (v\inext\j);
  }
}

% Add \vdots at specified positions
\node[draw=none, fill=none, inner sep=0pt, yshift=3pt] at (1, \rows-3.5) {$\vdots$};
\node[draw=none, fill=none, inner sep=0pt, yshift=3pt] at (\cols, \rows-3.5) {$\vdots$};

\end{tikzpicture}

%% file: figs/hconnector_instance.tex
\begin{tikzpicture}[scale=0.6,every node/.style={scale=0.85}]
\tikzstyle{vertex}=[circle, fill=black, inner sep=1.5pt, scale=0.85]
\tikzset{terminal/.style={fill=gray!30, draw=gray!50, thick, minimum size=0pt, inner sep=2pt}}

% Parameters
\def\rows{6}
\def\cols{8}

% Draw vertices (no corners) and add terminal labels
\foreach \i in {1,...,\rows}{
  \pgfmathtruncatemacro{\halfcols}{\cols-2} % number of inner terminals per row
  \foreach \j in {1,...,\cols}{
    % Skip corners
    \ifnum\i=1
      \ifnum\j=1
      \else\ifnum\j=\cols
      \else
        % Upper terminals (top row)
        \pgfmathtruncatemacro{\q}{\j-1} % left to right 1..6
        \node[terminal] (v\i\j) at (\j,\rows-\i) {};
        \node[above=2pt of v\i\j] {$h^{i,j}_{\q,\triangleleft}$};
      \fi\fi
    \else
        \ifnum\i=\rows
            \ifnum\j=1
            \else
                \ifnum\j=\cols
                \else
                    % Bottom terminals (bottom row)
                    \pgfmathtruncatemacro{\q}{\j-1} % left to right 1..6
                    \node[terminal] (v\i\j) at (\j,\rows-\i) {};
                    \node[below=2pt of v\i\j] {$h^{i,j}_{\q,\triangleright}$};
                \fi
            \fi
        \else
            % Internal vertices (same as before)
            \ifnum\j=1
                \ifnum\i=2
                    \node[vertex,label=left:$u^{i,j}_{\rowNum,\colNum}$] (v\i\j) at (\j,\rows-\i) {};
                \else
                    \ifnum\i=3
                        \node[vertex,label=left:$u^{i,j}_{\rowNum-1,\colNum}$] (v\i\j) at (\j,\rows-\i) {};
                    \else
                        \ifnum\i=4
                            \node[vertex,label=left:$u^{i,j}_{2,\colNum}$] (v\i\j) at (\j,\rows-\i) {};
                        \else
                            \node[vertex,label=left:$u^{i,j}_{1,\colNum}$] (v\i\j) at (\j,\rows-\i) {};
                        \fi
                    \fi
                \fi
            \else
                \ifnum\j=8
                    \ifnum\i=2
                        \node[vertex,label=right:$u^{i,j+1}_{\rowNum,1}$] (v\i\j) at (\j,\rows-\i) {};
                    \else
                        \ifnum\i=3
                            \node[vertex,label=right:$u^{i,j+1}_{\rowNum-1,1}$] (v\i\j) at (\j,\rows-\i) {};
                        \else   
                            \ifnum\i=4
                                \node[vertex,label=right:$u^{i,j+1}_{2,1}$] (v\i\j) at (\j,\rows-\i) {};
                            \else
                                \node[vertex,label=right:$u^{i,j+1}_{1,1}$] (v\i\j) at (\j,\rows-\i) {};
                            \fi
                        \fi
                    \fi
                \else
                    \node[vertex] (v\i\j) at (\j,\rows-\i) {};
                \fi
            \fi
        \fi
    \fi
  }
}

% Draw horizontal edges (skip first/last column)
\foreach \i in {2,...,\numexpr\rows-1\relax}{
  \foreach \j in {1,...,\numexpr\cols-1\relax}{
    \pgfmathtruncatemacro{\jnext}{\j+1}
    \ifnum\i=3
        \draw[very thick,draw=violet!60!white] (v\i\j) -- (v\i\jnext);
    \else
        \draw (v\i\j) -- (v\i\jnext);
    \fi
  }
}

% Draw vertical edges (skip first/last row)
\foreach \i in {1,...,\numexpr\rows-1\relax}{
  \foreach \j in {2,...,\numexpr\cols-1\relax}{
    \pgfmathtruncatemacro{\inext}{\i+1}
    \ifnum\i<3
        \draw[very thick, draw=red!60!white] (v\i\j) -- (v\inext\j);
    \else
        \draw[very thick, draw=blue!60!white] (v\i\j) -- (v\inext\j);
    \fi
  }
}

% Add \vdots at specified positions
\node[draw=none, fill=none, inner sep=0pt, yshift=3pt] at (1, \rows-3.5) {$\vdots$};
\node[draw=none, fill=none, inner sep=0pt, yshift=3pt] at (\cols, \rows-3.5) {$\vdots$};

\end{tikzpicture}

%% file: figs/cellGadget_a.tex
\begin{tikzpicture}[scale=0.8]
\tikzstyle{vertex}=[circle, fill=black, inner sep=1.5pt, scale=0.85]

% Parameters
\def\rows{4}
\def\cols{12}
\def\xspacing{1.2}
\def\yspacing{1.0}

% ----------------- Top grid -----------------
\foreach \i in {1,...,\rows}{
    \foreach \j in {1,...,\cols}{
        \pgfmathsetmacro{\x}{\j*\xspacing}
        \pgfmathsetmacro{\y}{-\i*\yspacing}
        \node[vertex] (T\i\j) at (\x,\y) {};
    }
}

\foreach \i in {1,...,\rows}{
    \foreach \j in {1,...,\numexpr\cols-1\relax}{
        \pgfmathtruncatemacro{\jn}{\j+1}

        % Check if this edge should be red
        \ifnum\i=4
            \ifnum\j=5 \draw[red, very thick] (T\i\j) -- (T\i\jn);
            \else\ifnum\j=11 \draw[red, very thick] (T\i\j) -- (T\i\jn);
            \else \draw (T\i\j) -- (T\i\jn); \fi\fi
        \else\ifnum\i=3
            \ifnum\j=2 \draw[red, very thick] (T\i\j) -- (T\i\jn);
            \else\ifnum\j=8 \draw[red, very thick] (T\i\j) -- (T\i\jn);
            \else \draw (T\i\j) -- (T\i\jn); \fi\fi
        \else\ifnum\i=2
            \ifnum\j=5 \draw[red, very thick] (T\i\j) -- (T\i\jn);
            \else\ifnum\j=11 \draw[red, very thick] (T\i\j) -- (T\i\jn);
            \else \draw (T\i\j) -- (T\i\jn); \fi\fi
        \else\ifnum\i=1
            \ifnum\j=2 \draw[red, very thick] (T\i\j) -- (T\i\jn);
            \else\ifnum\j=8 \draw[red, very thick] (T\i\j) -- (T\i\jn);
            \else \draw (T\i\j) -- (T\i\jn); \fi\fi
        \fi\fi\fi\fi
    }
}

% Vertical edges (all black)
\foreach \i in {1,...,\numexpr\rows-1\relax}{
    \foreach \j in {1,...,\cols}{
        \pgfmathtruncatemacro{\inext}{\i+1}
        \draw (T\i\j) -- (T\inext\j);
    }
}

% Leftmost column labels u^{i,1}
\foreach \i in {1,...,\rows}{
    \pgfmathtruncatemacro{\index}{5-\i}
    \node[left=3pt of T\i1, fill=none] {$u^{2,1}_{\index,1}$};
}

% Bottom row labels u^{1,j} with j'*3 + r
\foreach \j in {1,...,\cols}{
    \pgfmathtruncatemacro{\r}{mod(\j,3)}
    \ifnum\r=0
        \pgfmathtruncatemacro{\jr}{3}
    \else
        \pgfmathtruncatemacro{\jr}{\r}
    \fi
    \pgfmathtruncatemacro{\jprime}{\j-\jr}
    \node[below=3pt of T\rows\j, fill=none] {$u^{2,1}_{1,\jprime+\jr}$};
}

\end{tikzpicture}

%% file: figs/cellGadget_b.tex
\begin{tikzpicture}[scale=0.8]
\tikzstyle{vertex}=[circle, fill=black, inner sep=1.5pt, scale=0.85]

% Parameters
\def\rows{4}
\def\cols{12}
\def\xspacing{1.2}
\def\yspacing{1.0}

\def\yoffset{-\rows*\yspacing-2.25} % space between grids
\foreach \i in {1,...,\rows}{
    \foreach \j in {1,...,\cols}{
        \pgfmathsetmacro{\x}{\j*\xspacing}
        \pgfmathsetmacro{\y}{-\i*\yspacing + \yoffset}
        \node[vertex] (B\i\j) at (\x,\y) {};
    }
}

% ----- Horizontal edges -----
\foreach \i in {1,...,\rows}{
  \foreach \j in {1,...,\numexpr\cols-1\relax}{
    \pgfmathtruncatemacro{\jn}{\j+1}

    % Special coloring for row 2
    \ifnum\i=3
      \ifnum\j=8
        \draw[green!50!black, very thick] (B\i\j) -- (B\i\jn);
      \else
        \draw[yellow!90!black, very thick] (B\i\j) -- (B\i\jn);
      \fi
    \else
      % Default black for other rows
      \draw (B\i\j) -- (B\i\jn);
    \fi
  }
}

% ----- Vertical edges -----
\foreach \i in {1,...,\numexpr\rows-1\relax}{
  \foreach \j in {1,...,\cols}{
    \pgfmathtruncatemacro{\inext}{\i+1}

    % Blue segments in columns 8 and 9
    \ifnum\j=8
      % Only up to (2,8): edges (4,8)-(3,8) and (3,8)-(2,8)
        \ifnum\i=3
            \draw[blue!60!white, very thick] (B\i\j) -- (B\inext\j);
        \else
            \draw (B\i\j) -- (B\inext\j); % above (2,8) is black
        \fi

    \else\ifnum\j=9
      % Only from (2,9) up: edges (2,9)-(1,9)
      \ifnum\i<3
        \draw[blue!70!white, very thick] (B\i\j) -- (B\inext\j);
      \else
        \draw (B\i\j) -- (B\inext\j);
      \fi
    \else
      % Default black
      \draw (B\i\j) -- (B\inext\j);
    \fi\fi
  }
}

% Leftmost column labels v^{i,1}
\foreach \i in {1,...,\rows}{
    \pgfmathtruncatemacro{\index}{5-\i}
    \node[left=3pt of B\i1, fill=none] {$v^{2,1}_{\index,1}$};
}

% Bottom row labels v^{1,j} with j'*3 + r
\foreach \j in {1,...,\cols}{
    \pgfmathtruncatemacro{\r}{mod(\j,3)}
    \ifnum\r=0
        \pgfmathtruncatemacro{\jr}{3}
    \else
        \pgfmathtruncatemacro{\jr}{\r}
    \fi
    \pgfmathtruncatemacro{\jprime}{\j-\jr}
    \node[below=3pt of B\rows\j, fill=none] {$v^{2,1}_{1,\jprime+\jr}$};
}

% Label (b) below bottom grid
% \node[fill=none] at (\cols*\xspacing/2, \yoffset - \rows*\yspacing - 1.75) {(b)};

\end{tikzpicture}

%% file: figs/ULVGraphSmall.tex
\begin{tikzpicture}[yscale=0.9, xscale=1.6, every node/.style={font=\small, align=center, scale=0.85}]
\tikzset{terminal/.style={fill=gray!30, draw=gray!50,  thick,minimum size=0pt, inner sep=2pt}}

% Size parameters
\def\squaresize{1}        % square cells
\def\hrectwidth{1}      % horizontal rectangle width
\def\hrectheight{0.8}     % horizontal rectangle height
\def\vrectwidth{0.9}      % vertical rectangle width
\def\vrectheight{1}     % vertical rectangle height
\def\narrowcolwidth{0.5}

\newcommand{\labelfun}[2]{%
    \edef\tempI{#1}%
    \edef\tempJ{#2}%
    \ifnum#1=0 \edef\tempI{\uparrow}\fi
    \ifnum#1=6 \edef\tempI{\downarrow}\fi
    \ifnum#2=0 \edef\tempJ{\rightarrow}\fi
    \ifnum#2=8 \edef\tempJ{\leftarrow}\fi
    \ifnum#1=5 \edef\tempI{1}\fi
    \ifnum#1=3 \edef\tempI{k-1}\fi
    \ifnum#1=2 \edef\tempI{k-1}\fi
    \ifnum#1=1 \edef\tempI{k}\fi
    \ifnum#2=1 \edef\tempJ{1}\fi
    \ifnum#2=2 \edef\tempJ{1}\fi
    \ifnum#2=3 \edef\tempJ{2}\fi
    \ifnum#2=5 \edef\tempJ{k-1}\fi
    \ifnum#2=6 \edef\tempJ{k-1}\fi
    \ifnum#2=7 \edef\tempJ{k}\fi
    \tempI,\tempJ%
}

% Gray squares
\foreach \i/\j in {1/1,1/3,1/5,1/7,
                   3/1,3/3,3/5,3/7,
                   5/1,5/3,5/5,5/7} {
    \ifnum\j<5
        \pgfmathsetmacro{\xstart}{(\j-1)*\squaresize}
    \else
        \pgfmathsetmacro{\xstart}{3*\squaresize + \narrowcolwidth + (\squaresize*(\j-5))}
    \fi
    \ifnum\j=4
        \draw[fill=gray!50] (\xstart,5-\i) rectangle ++(\narrowcolwidth,\squaresize);
        \node at (\xstart+0.5*\narrowcolwidth, 5-\i+0.5*\squaresize) { $G^{\labelfun{\i}{\j}}$};
    \else
        \draw[fill=gray!50] (\xstart,5-\i) rectangle ++(\squaresize,\squaresize);
        \node at (\xstart+0.5*\squaresize, 5-\i+0.5*\squaresize) { $G^{\labelfun{\i}{\j}}$};
    \fi
}

% 2) C^{i,j}_H
\foreach \i/\j in {1/2,1/6,3/2,3/6,5/2,5/6} {
    \ifnum\j<5
        \pgfmathsetmacro{\xstart}{(\j-1)*\squaresize}
    \else
        \pgfmathsetmacro{\xstart}{3*\squaresize + \narrowcolwidth + (\squaresize*(\j-5))}
    \fi
    \ifnum\i=1
        \ifnum\j=2
            \draw[fill=yellow!30,draw=none] 
    (\xstart+0.5*\squaresize-0.5*\hrectwidth,
     5-\i+0.5*\squaresize-0.5*\hrectheight-0.08)
    rectangle ++(\hrectwidth,\hrectheight+0.16);
            \node  at (\xstart+0.5*\vrectwidth+0.07,5-\i+0.5*\vrectheight) {
                \begin{tikzpicture}[yscale=0.24,xscale=0.28]
                    \tikzstyle{vertex}=[circle, fill=black, inner sep=1.5pt, scale=0.85]
                    \tikzset{terminal/.style={fill=gray!30, draw=gray!50, thick, minimum size=0pt, inner sep=2pt}}

                    % Parameters
                    \def\rows{6}
                    \def\cols{8}

                    % Draw vertices (no corners)
                    \foreach \ii in {1,...,\rows}{
                      \foreach \jj in {1,...,\cols}{
                        % Skip corners
                        \ifnum\ii=1
                          \ifnum\jj=1
                          \else\ifnum\jj=\cols
                          \else
                            \node[terminal] (v\ii\jj) at (\jj,\rows-\ii) {};
                          \fi\fi
                        \else
                            \ifnum\ii=\rows
                                \ifnum\jj=1
                                \else
                                    \ifnum\jj=\cols
                                    \else
                                        \node[terminal] (v\ii\jj) at (\jj,\rows-\ii) {};
                                    \fi
                                \fi
                            \else
                                \ifnum\jj=1
                                    \node[vertex] (v\ii\jj) at (\jj,\rows-\ii) {};
                                \else
                                    \ifnum\jj=8
                                        \node[vertex] (v\ii\jj) at (\jj,\rows-\ii) {};
                                    \else
                                        \node[vertex] (v\ii\jj) at (\jj,\rows-\ii) {};
                                    \fi
                                \fi
                            \fi
                        \fi
                      }
                    }

                    % Horizontal edges
                    \foreach \ii in {2,...,\numexpr\rows-1\relax}{
                      \foreach \jj in {1,...,\numexpr\cols-1\relax}{
                        \pgfmathtruncatemacro{\jnext}{\jj+1}
                        \draw (v\ii\jj) -- (v\ii\jnext);
                      }
                    }

                    % Vertical edges
                    \foreach \ii in {1,...,\numexpr\rows-1\relax}{
                      \foreach \jj in {2,...,\numexpr\cols-1\relax}{
                        \pgfmathtruncatemacro{\inext}{\ii+1}
                        \draw (v\ii\jj) -- (v\inext\jj);
                      }
                    }

                \end{tikzpicture}
            };
        \else
            % Normal cell
            \draw[fill=yellow!30] (\xstart+0.5*\squaresize-0.5*\hrectwidth,
                                   5-\i+0.5*\squaresize-0.5*\hrectheight)
                  rectangle ++(\hrectwidth,\hrectheight);
            \node at (\xstart+0.5*\squaresize,5-\i+0.5*\squaresize) { $C^{\labelfun{\i}{\j}}_H$};
        \fi
    \else
        % Normal cells
        \draw[fill=yellow!30] (\xstart+0.5*\squaresize-0.5*\hrectwidth,
                               5-\i+0.5*\squaresize-0.5*\hrectheight)
              rectangle ++(\hrectwidth,\hrectheight);
        \node at (\xstart+0.5*\squaresize,5-\i+0.5*\squaresize) { $\hconnector^{\labelfun{\i}{\j}}$};
    \fi
}

% 3) C^{i,j}_V
\foreach \i/\j in {2/1,2/3,2/5,2/7} {
    \ifnum\j<5
        \pgfmathsetmacro{\xstart}{(\j-1)*\squaresize}
    \else
        \pgfmathsetmacro{\xstart}{3*\squaresize + \narrowcolwidth + (\squaresize*(\j-5))}
    \fi
    \draw[fill=yellow!30] (\xstart+0.5*\squaresize-0.5*\vrectwidth,
                           5-\i+0.5*\squaresize-0.5*\vrectheight)
          rectangle ++(\vrectwidth,\vrectheight);
    \node at (\xstart+0.5*\squaresize,5-\i+0.5*\squaresize) { $\vconnector^{\labelfun{\i}{\j}}$};
}

% 4) Vertical 3 dots
\foreach \j in {1,3,5,7} {
    \ifnum\j<5
        \pgfmathsetmacro{\xstart}{(\j-1)*\squaresize}
    \else
        \pgfmathsetmacro{\xstart}{3*\squaresize + \narrowcolwidth + (\squaresize*(\j-5))}
    \fi
    \draw[fill=yellow!15,draw=none] (\xstart+0.5*\squaresize-0.5*\vrectwidth,
                           5-4+0.5*\squaresize-0.5*\vrectheight)
          rectangle ++(\vrectwidth,\vrectheight);
    \node at (\xstart+0.5*\squaresize, 5-4+0.5*\squaresize) {$\vdots$};
}

% 5) Horizontal 3 dots
\foreach \i in {1,3,5} {
    \draw[fill=yellow!15,draw=none] (4-1, 5-\i)
          rectangle ++(\narrowcolwidth,\vrectheight);
    \node at (4-1+0.5*\narrowcolwidth, 5-\i+0.5*\squaresize) {$\cdots$};
}

% 6) External vertices and connections on the left with labels b^{i,\rightarrow}
\def\xoffset{-0.4}  % horizontal distance to the left
\foreach \i in {1,3,5} {
    % position of the gray square (column 1)
    \pgfmathsetmacro{\xleft}{0} % left x of gray square
    \pgfmathsetmacro{\ytop}{5-\i+\squaresize}
    \pgfmathsetmacro{\ybottom}{5-\i}
    \pgfmathsetmacro{\ymid}{5-\i+0.5*\squaresize}

    % draw vertex
    \node[circle, fill=black, inner sep=2pt, scale=0.85] (v\i) at (\xoffset,\ymid) {};

    % label above vertex
    \node[above=4pt of v\i] {$b^{\labelfun{\i}{0}}$};

    % terminals
    \coordinate (t1\i) at (\xoffset+0.1, \ymid-0.5);
    \coordinate (t2\i) at (\xoffset-0.1, \ymid-0.5);
    \node[terminal] (vt1\i) at (t1\i) {};
    \node[terminal] (vt2\i) at (t2\i) {};
    
    \draw[thick, blue!70!black] (v\i) -- (vt1\i);
    \draw[thick, blue!70!black] (v\i) -- (vt2\i);

    % connect vertex to top-left and bottom-left corners of gray square
    
    \draw[thick,color=orange!100!black] (v\i) -- (\xleft,\ytop-0.1);
    \draw[thick,color=orange!100!black] (v\i) -- (\xleft,\ybottom+0.1);

    \node at (\xleft-0.12,\ytop-0.4) {$\vdots$};
}

% 7) Right external vertices and connections (b^{i,\leftarrow})
\def\xrightoffset{6*\squaresize+\narrowcolwidth+0.4} % horizontal offset to the right
\foreach \i in {1,3,5} {
    \pgfmathsetmacro{\xright}{6*\squaresize+\narrowcolwidth} % rightmost cell's x start
    \pgfmathsetmacro{\ytop}{5-\i+\squaresize}
    \pgfmathsetmacro{\ybottom}{5-\i}
    \pgfmathsetmacro{\ymid}{5-\i+0.5*\squaresize}

    % vertex
    \node[circle, fill=black, inner sep=2pt, scale=0.85] (vR\i) at (\xrightoffset,\ymid) {};
    \node[above=4pt of vR\i] {$b^{\labelfun{\i}{8}}$};

    % terminals
    \coordinate (t1R\i) at (\xrightoffset+0.1, \ymid-0.5);
    \coordinate (t2R\i) at (\xrightoffset-0.1, \ymid-0.5);
    \node[terminal] (vt1R\i) at (t1R\i) {};
    \node[terminal] (vt2R\i) at (t2R\i) {};
    
    \draw[thick, blue!70!black] (vR\i) -- (vt1R\i);
    \draw[thick, blue!70!black] (vR\i) -- (vt2R\i);

    % connect to top-right and bottom-right corners
    
    \draw[thick,color=orange!100!black] (vR\i) -- (\xright, \ytop-0.1);
    \draw[thick,color=orange!100!black] (vR\i) -- (\xright, \ybottom+0.1);
    \node at (\xright+0.12,\ytop-0.4) {$\vdots$};
}

% 8) rcons
\def\xrcons{-1.1} % further left than \xoffset (-0.7)
\def\yrcons{5-4+0.5*\squaresize} % roughly centered vertically
\node[circle, fill=black, inner sep=2pt, scale=0.85] (vrcons) at (\xrcons,\yrcons) {};
\node[left=2pt of vrcons] {$\rcons$};

% terminals
\coordinate (t1rcons) at (\xrcons-0.3, \yrcons-0.5);
\coordinate (t2rcons) at (\xrcons-0.3, \yrcons+0.5);
\node[terminal] (term1rcons) at (t1rcons) {};
\node[terminal] (term2rcons) at (t2rcons) {};

\draw[thick, blue!70!black] (vrcons) -- (term1rcons);
\draw[thick, blue!70!black] (vrcons) -- (term2rcons);

% connect rcons to all left-side vertices
\foreach \i in {1,3,5} {
    \ifnum\i<4
        
        \draw[thick, color=green!50!black, bend left=20] (vrcons) to (v\i);
    \else
        
        \draw[thick, color=green!50!black, bend right=20] (vrcons) to (v\i);
    \fi
}

\node at (\xrcons+0.3,\yrcons+0.1) {$\vdots$};

% 9) Bottom vertices \(b^{\uparrow,\j}\)
\def\yoffset{-0.7} % vertical distance below the table

\foreach \j in {1,3,5,7} {
    % compute bottom midpoint of gray cell
    \pgfmathsetmacro{\ybottom}{0}
    \ifnum\j<5
        \pgfmathsetmacro{\xleft}{\j-1}
        \pgfmathsetmacro{\xright}{(\j-1)*\squaresize+\squaresize}
        \pgfmathsetmacro{\xbottommid}{\j-1+0.5*\squaresize}
    \else
        \pgfmathsetmacro{\xbottommid}{3*\squaresize+\narrowcolwidth+(\j-5)*\squaresize+0.5*\squaresize}
        \pgfmathsetmacro{\xleft}{3*\squaresize+\narrowcolwidth+(\j-5)*\squaresize}
        \pgfmathsetmacro{\xright}{3*\squaresize+\narrowcolwidth+(\j-5)*\squaresize+\squaresize}
    \fi

    % vertex below
    \node[circle, fill=black, inner sep=2pt, scale=0.85] (vB\j) at (\xbottommid,\ybottom+\yoffset) {};

    % label below vertex
    \ifnum\j<4
        \node[left=2pt of vB\j] {$b^{\labelfun{0}{\j}}$};
    \else
        \node[right=2pt of vB\j] {$b^{\labelfun{0}{\j}}$};
    \fi

    % terminals
    \coordinate (t1B\j) at (\xbottommid+0.1, \ybottom+\yoffset-0.5);
    \coordinate (t2B\j) at (\xbottommid-0.1, \ybottom+\yoffset-0.5);
    \node[terminal] (vt1B\j) at (t1B\j) {};
    \node[terminal] (vt2B\j) at (t2B\j) {};
    
    \draw[thick, blue!70!black] (vB\j) -- (vt1B\j);
    \draw[thick, blue!70!black] (vB\j) -- (vt2B\j);

    % connect vertex to bottom-left and bottom-right corners of gray cell
    
    \draw[thick,color=orange!100!black] (vB\j) -- (\xleft+0.05,\ybottom);
    \draw[thick,color=orange!100!black] (vB\j) -- (\xright-0.15,\ybottom);
    \draw[thick,color=orange!100!black] (vB\j) -- (\xleft+0.25,\ybottom);
    \node at (\xleft+0.55,\ybottom-0.2) {$\ldots$};
}

% 10) rselect
\pgfmathsetmacro{\xrselect}{(6*\squaresize + \narrowcolwidth)/2} % horizontal center of table
\pgfmathsetmacro{\yrselect}{-1.7}

\node[circle, fill=black, inner sep=2pt, scale=0.85] (vrselect) at (\xrselect,\yrselect) {};
\node[below=2pt of vrselect] {$\rselect$}; 

% terminals
\coordinate (t1rselect) at (\xrselect-0.5, \yrselect-0.5);
\coordinate (t2rselect) at (\xrselect+0.5, \yrselect-0.5);
\node[terminal] (term1rselect) at (t1rselect) {};
\node[terminal] (term2rselect) at (t2rselect) {};

\draw[thick, blue!70!black] (vrselect) -- (term1rselect);
\draw[thick, blue!70!black] (vrselect) -- (term2rselect);

% Connect to all bottom-row vertices (row i=7)
\foreach \j in {1,3,5,7} {
    \ifnum\j<4
        
        \draw[thick, color=green!50!black, bend left=10] (vrselect) to (vB\j);
    \else
        
        \draw[thick, color=green!50!black, bend right=10] (vrselect) to (vB\j);
    \fi
}

\node at (\xrselect,\yrselect+0.4) {$\ldots$};

% 11) b^{\downarrow}{j}
\def\yTopVertexOffset{0.7} % vertical distance above top row
\foreach \j in {1,3,5,7} {
    \ifnum\j<5
        \pgfmathsetmacro{\xleft}{\j-1}
        \pgfmathsetmacro{\xright}{(\j-1)*\squaresize+\squaresize}
        \pgfmathsetmacro{\xbottommid}{\j-1+0.5*\squaresize}
    \else
        \pgfmathsetmacro{\xbottommid}{3*\squaresize+\narrowcolwidth+(\j-5)*\squaresize+0.5*\squaresize}
        \pgfmathsetmacro{\xleft}{3*\squaresize+\narrowcolwidth+(\j-5)*\squaresize}
        \pgfmathsetmacro{\xright}{3*\squaresize+\narrowcolwidth+(\j-5)*\squaresize+\squaresize}
    \fi
    \pgfmathsetmacro{\ytopcell}{5-1+\squaresize}
    \pgfmathsetmacro{\ytopv}{\ytopcell+\yTopVertexOffset}

    % top vertex
    \node[circle, fill=black, inner sep=2pt, scale=0.85] (vT1\j) at (\xbottommid,\ytopv) {};

    % label above
    \node[left=2pt of vT1\j] {$b^{\labelfun{6}{\j}}$};

    % terminals
    \coordinate (t1T1\j) at (\xbottommid+0.1, \ytopv+0.5);
    \coordinate (t2T1\j) at (\xbottommid-0.1, \ytopv+0.5);
    \node[terminal] (vt1T1\j) at (t1T1\j) {};
    \node[terminal] (vt2T1\j) at (t2T1\j) {};
    
    \draw[thick, blue!70!black] (vT1\j) -- (vt1T1\j);
    \draw[thick, blue!70!black] (vT1\j) -- (vt2T1\j);

    % connect vertex to top-left and top-right corners of corresponding gray square
    
    \draw[thick,color=orange!100!black] (vT1\j) -- (\xleft+0.16,\ytopcell);
    \draw[thick,color=orange!100!black] (vT1\j) -- (\xright-0.05,\ytopcell);
    \draw[thick,color=orange!100!black] (vT1\j) -- (\xleft+0.36,\ytopcell);
    \node at (\xleft+0.55,\ytopcell+0.2) {$\ldots$};
}

% 12) Add vertex r (same y as rselect, same x as rcons)
\pgfmathsetmacro{\xr}{\xrcons+0.3}   % same x as rcons
\pgfmathsetmacro{\yr}{\yrselect+0.3}   % same y as rselect

% node itself
\node[circle, fill=black, inner sep=2pt, scale=0.85] (vr) at (\xr,\yr) {};
\node[left=2pt of vr] {$r$};

% connect r to rcons and rselect

\draw[thick, color=blue!70!black,bend left=15] (vr) to (vrcons);
\draw[thick, color=blue!70!black,bend right=10] (vr) to (vrselect);

\end{tikzpicture}

%% file: 2TreesVs4Trees.tex
\begin{figure}[!t]
    \begin{subfigure}[b]{\textwidth}
    \centering
    \begin{minipage}{\textwidth}
    \centering
    \begin{minipage}{0.48\textwidth}
      \centering
      \input{figs/hconnector-instances-2trees-samerow}
    \end{minipage}\hfill
    \begin{minipage}{0.48\textwidth}
      \centering
      \input{figs/hconnector-instances-2trees-2rows}
    \end{minipage}
    \subcaption{Purple edges are of trees, each passing through a connector and covering all terminals. Left: the tree intersects only one row. Right: the tree intersects/switches between two rows.
    Both trees have equal cost.}
  \end{minipage}
    \end{subfigure}
    
  \vspace{1em}

  % ----- Row 2 -----
  \begin{subfigure}[b]{\textwidth}
  \centering
  \begin{minipage}{\textwidth}
    \centering
    \begin{minipage}{0.48\textwidth}
      \centering
      \input{figs/hconnector-instances-4trees-samerow}
    \end{minipage}\hfill
    \begin{minipage}{0.48\textwidth}
      \centering
      \input{figs/hconnector-instances-4trees-2rows}
    \end{minipage}
    \subcaption{Purple edges are shared by pairs of trees with red resp.\ blue private edges, each passing through a connector and covering all terminals. Left: trees intersects one row. Right: the pair of trees switches between two rows and has higher cost due to counting the vertical purple edge twice.}
    % The two pairs have different costs: right includes the length of the vertical purple edge twice, so switching rows increases the cost.}
  \end{minipage}
  \end{subfigure}
  \caption{}
  \label{fig:4TreesVs2Trees}
\end{figure}

%% file: figs/hconnector-instances-2trees-samerow.tex
\begin{tikzpicture}[scale=0.6,every node/.style={scale=0.85}]
\tikzstyle{vertex}=[circle, fill=black, inner sep=1.5pt, scale=0.85]
\tikzset{terminal/.style={fill=gray!30, draw=gray!50, thick, minimum size=0pt, inner sep=2pt}}

% Parameters
\def\rows{6}
\def\cols{8}

% Draw vertices (no corners) and add terminal labels
\foreach \i in {1,...,\rows}{
  \pgfmathtruncatemacro{\halfcols}{\cols-2} % number of inner terminals per row
  \foreach \j in {1,...,\cols}{
    % Skip corners
    \ifnum\i=1
      \ifnum\j=1
      \else\ifnum\j=\cols
      \else
        % Upper terminals (top row)
        \pgfmathtruncatemacro{\q}{\j-1} % left to right 1..6
        \node[terminal] (v\i\j) at (\j,\rows-\i) {};
        % \node[above=2pt of v\i\j] {$h^{i,j}_{\q,\triangleleft}$};
      \fi\fi
    \else
        \ifnum\i=\rows
            \ifnum\j=1
            \else
                \ifnum\j=\cols
                \else
                    % Bottom terminals (bottom row)
                    \pgfmathtruncatemacro{\q}{\j-1} % left to right 1..6
                    \node[terminal] (v\i\j) at (\j,\rows-\i) {};
                    % \node[below=2pt of v\i\j] {$h^{i,j}_{\q,\triangleright}$};
                \fi
            \fi
        \else
            % Internal vertices (same as before)
            \ifnum\j=1
                \ifnum\i=2
                    \node[vertex] (v\i\j) at (\j,\rows-\i) {};
                \else
                    \ifnum\i=3
                        \node[vertex] (v\i\j) at (\j,\rows-\i) {};
                    \else
                        \ifnum\i=4
                            \node[vertex] (v\i\j) at (\j,\rows-\i) {};
                        \else
                            \node[vertex] (v\i\j) at (\j,\rows-\i) {};
                        \fi
                    \fi
                \fi
            \else
                \ifnum\j=8
                    \ifnum\i=2
                        \node[vertex] (v\i\j) at (\j,\rows-\i) {};
                    \else
                        \ifnum\i=3
                            \node[vertex] (v\i\j) at (\j,\rows-\i) {};
                        \else   
                            \ifnum\i=4
                                \node[vertex] (v\i\j) at (\j,\rows-\i) {};
                            \else
                                \node[vertex] (v\i\j) at (\j,\rows-\i) {};
                            \fi
                        \fi
                    \fi
                \else
                    \node[vertex] (v\i\j) at (\j,\rows-\i) {};
                \fi
            \fi
        \fi
    \fi
  }
}

% Draw horizontal edges (skip first/last column)
\foreach \i in {2,...,\numexpr\rows-1\relax}{
  \foreach \j in {1,...,\numexpr\cols-1\relax}{
    \pgfmathtruncatemacro{\jnext}{\j+1}
    \ifnum\i=3
        \draw[very thick,draw=violet!80!white] (v\i\j) -- (v\i\jnext);
    \else
        \draw (v\i\j) -- (v\i\jnext);
    \fi
  }
}

% Draw vertical edges (skip first/last row)
\foreach \i in {1,...,\numexpr\rows-1\relax}{
  \foreach \j in {2,...,\numexpr\cols-1\relax}{
    \pgfmathtruncatemacro{\inext}{\i+1}
    % \ifnum\i<3
    %     \draw[very thick, draw=violet!80!white] (v\i\j) -- (v\inext\j);
    % \else
        \draw[thick,violet!60!white] (v\i\j) -- (v\inext\j);
    % \fi
  }
}

% Add \vdots at specified positions
\node[draw=none, fill=none, inner sep=0pt, yshift=3pt] at (1, \rows-3.5) {$\vdots$};
\node[draw=none, fill=none, inner sep=0pt, yshift=3pt] at (\cols, \rows-3.5) {$\vdots$};

\end{tikzpicture}

%% file: figs/hconnector-instances-2trees-2rows.tex
\begin{tikzpicture}[scale=0.6,every node/.style={scale=0.85}]
\tikzstyle{vertex}=[circle, fill=black, inner sep=1.5pt, scale=0.85]
\tikzset{terminal/.style={fill=gray!30, draw=gray!50, thick, minimum size=0pt, inner sep=2pt}}

% Parameters
\def\rows{6}
\def\cols{8}

% Draw vertices (no corners) and add terminal labels
\foreach \i in {1,...,\rows}{
  \pgfmathtruncatemacro{\halfcols}{\cols-2} % number of inner terminals per row
  \foreach \j in {1,...,\cols}{
    % Skip corners
    \ifnum\i=1
      \ifnum\j=1
      \else\ifnum\j=\cols
      \else
        % Upper terminals (top row)
        \pgfmathtruncatemacro{\q}{\j-1} % left to right 1..6
        \node[terminal] (v\i\j) at (\j,\rows-\i) {};
        % \node[above=2pt of v\i\j] {$h^{i,j}_{\q,\triangleleft}$};
      \fi\fi
    \else
        \ifnum\i=\rows
            \ifnum\j=1
            \else
                \ifnum\j=\cols
                \else
                    % Bottom terminals (bottom row)
                    \pgfmathtruncatemacro{\q}{\j-1} % left to right 1..6
                    \node[terminal] (v\i\j) at (\j,\rows-\i) {};
                    % \node[below=2pt of v\i\j] {$h^{i,j}_{\q,\triangleright}$};
                \fi
            \fi
        \else
            % Internal vertices (same as before)
            \ifnum\j=1
                \ifnum\i=2
                    \node[vertex] (v\i\j) at (\j,\rows-\i) {};
                \else
                    \ifnum\i=3
                        \node[vertex] (v\i\j) at (\j,\rows-\i) {};
                    \else
                        \ifnum\i=4
                            \node[vertex] (v\i\j) at (\j,\rows-\i) {};
                        \else
                            \node[vertex] (v\i\j) at (\j,\rows-\i) {};
                        \fi
                    \fi
                \fi
            \else
                \ifnum\j=8
                    \ifnum\i=2
                        \node[vertex] (v\i\j) at (\j,\rows-\i) {};
                    \else
                        \ifnum\i=3
                            \node[vertex] (v\i\j) at (\j,\rows-\i) {};
                        \else   
                            \ifnum\i=4
                                \node[vertex] (v\i\j) at (\j,\rows-\i) {};
                            \else
                                \node[vertex] (v\i\j) at (\j,\rows-\i) {};
                            \fi
                        \fi
                    \fi
                \else
                    \node[vertex] (v\i\j) at (\j,\rows-\i) {};
                \fi
            \fi
        \fi
    \fi
  }
}

% Draw horizontal edges (skip first/last column)
\foreach \i in {2,...,\numexpr\rows-1\relax}{
  \foreach \j in {1,...,\numexpr\cols-1\relax}{
    \pgfmathtruncatemacro{\jnext}{\j+1}
    \ifnum\i=3
        \ifnum\j<4
            \draw[very thick,draw=violet!80!white] (v\i\j) -- (v\i\jnext);
        \else
            \draw (v\i\j) -- (v\i\jnext);
        \fi
    \else
        \ifnum\i=4
            \ifnum\j>3
                \draw[very thick,draw=violet!80!white] (v\i\j) -- (v\i\jnext);
            \else
                \draw (v\i\j) -- (v\i\jnext);
            \fi
        \else
            \draw (v\i\j) -- (v\i\jnext);
        \fi
    \fi
  }
}

% Draw vertical edges (skip first/last row)
\foreach \i in {1,...,\numexpr\rows-1\relax}{
  \foreach \j in {2,...,\numexpr\cols-1\relax}{
    \pgfmathtruncatemacro{\inext}{\i+1}
    \ifnum\i=3
        \ifnum\j=4
            \draw[very thick, draw=violet!80!white] (v\i\j) -- (v\inext\j);
        \else
            \draw[thick,violet!60!white] (v\i\j) -- (v\inext\j);
        \fi
    \else
        \draw[thick, violet!60!white] (v\i\j) -- (v\inext\j);
    \fi
  }
}

% Add \vdots at specified positions
\node[draw=none, fill=none, inner sep=0pt, yshift=3pt] at (1, \rows-3.5) {$\vdots$};
\node[draw=none, fill=none, inner sep=0pt, yshift=3pt] at (\cols, \rows-3.5) {$\vdots$};

\end{tikzpicture}

%% file: figs/hconnector-instances-4trees-samerow.tex
\begin{tikzpicture}[scale=0.6,every node/.style={scale=0.85}]
\tikzstyle{vertex}=[circle, fill=black, inner sep=1.5pt, scale=0.85]
\tikzset{terminal/.style={fill=gray!30, draw=gray!50, thick, minimum size=0pt, inner sep=2pt}}

% Parameters
\def\rows{6}
\def\cols{8}

% Draw vertices (no corners) and add terminal labels
\foreach \i in {1,...,\rows}{
  \pgfmathtruncatemacro{\halfcols}{\cols-2} % number of inner terminals per row
  \foreach \j in {1,...,\cols}{
    % Skip corners
    \ifnum\i=1
      \ifnum\j=1
      \else\ifnum\j=\cols
      \else
        % Upper terminals (top row)
        \pgfmathtruncatemacro{\q}{\j-1} % left to right 1..6
        \node[terminal] (v\i\j) at (\j,\rows-\i) {};
        % \node[above=2pt of v\i\j] {$h^{i,j}_{\q,\triangleleft}$};
      \fi\fi
    \else
        \ifnum\i=\rows
            \ifnum\j=1
            \else
                \ifnum\j=\cols
                \else
                    % Bottom terminals (bottom row)
                    \pgfmathtruncatemacro{\q}{\j-1} % left to right 1..6
                    \node[terminal] (v\i\j) at (\j,\rows-\i) {};
                    % \node[below=2pt of v\i\j] {$h^{i,j}_{\q,\triangleright}$};
                \fi
            \fi
        \else
            % Internal vertices (same as before)
            \ifnum\j=1
                \ifnum\i=2
                    \node[vertex] (v\i\j) at (\j,\rows-\i) {};
                \else
                    \ifnum\i=3
                        \node[vertex] (v\i\j) at (\j,\rows-\i) {};
                    \else
                        \ifnum\i=4
                            \node[vertex] (v\i\j) at (\j,\rows-\i) {};
                        \else
                            \node[vertex] (v\i\j) at (\j,\rows-\i) {};
                        \fi
                    \fi
                \fi
            \else
                \ifnum\j=8
                    \ifnum\i=2
                        \node[vertex] (v\i\j) at (\j,\rows-\i) {};
                    \else
                        \ifnum\i=3
                            \node[vertex] (v\i\j) at (\j,\rows-\i) {};
                        \else   
                            \ifnum\i=4
                                \node[vertex] (v\i\j) at (\j,\rows-\i) {};
                            \else
                                \node[vertex] (v\i\j) at (\j,\rows-\i) {};
                            \fi
                        \fi
                    \fi
                \else
                    \node[vertex] (v\i\j) at (\j,\rows-\i) {};
                \fi
            \fi
        \fi
    \fi
  }
}

% Draw horizontal edges (skip first/last column)
\foreach \i in {2,...,\numexpr\rows-1\relax}{
  \foreach \j in {1,...,\numexpr\cols-1\relax}{
    \pgfmathtruncatemacro{\jnext}{\j+1}
    \ifnum\i=3
        \draw[very thick,draw=violet!80!white] (v\i\j) -- (v\i\jnext);
    \else
        \draw (v\i\j) -- (v\i\jnext);
    \fi
  }
}

% Draw vertical edges (skip first/last row)
\foreach \i in {1,...,\numexpr\rows-1\relax}{
  \foreach \j in {2,...,\numexpr\cols-1\relax}{
    \pgfmathtruncatemacro{\inext}{\i+1}
    \ifnum\i<3
        \draw[thick, draw=red!40!white] (v\i\j) -- (v\inext\j);
    \else
        \draw[thick, draw=blue!40!white] (v\i\j) -- (v\inext\j);
    \fi
  }
}

% Add \vdots at specified positions
\node[draw=none, fill=none, inner sep=0pt, yshift=3pt] at (1, \rows-3.5) {$\vdots$};
\node[draw=none, fill=none, inner sep=0pt, yshift=3pt] at (\cols, \rows-3.5) {$\vdots$};

\end{tikzpicture}

%% file: figs/hconnector-instances-4trees-2rows.tex
\begin{tikzpicture}[scale=0.6,every node/.style={scale=0.85}]
\tikzstyle{vertex}=[circle, fill=black, inner sep=1.5pt, scale=0.85]
\tikzset{terminal/.style={fill=gray!30, draw=gray!50, thick, minimum size=0pt, inner sep=2pt}}

% Parameters
\def\rows{6}
\def\cols{8}

% Draw vertices (no corners) and add terminal labels
\foreach \i in {1,...,\rows}{
  \pgfmathtruncatemacro{\halfcols}{\cols-2} % number of inner terminals per row
  \foreach \j in {1,...,\cols}{
    % Skip corners
    \ifnum\i=1
      \ifnum\j=1
      \else\ifnum\j=\cols
      \else
        % Upper terminals (top row)
        \pgfmathtruncatemacro{\q}{\j-1} % left to right 1..6
        \node[terminal] (v\i\j) at (\j,\rows-\i) {};
        % \node[above=2pt of v\i\j] {$h^{i,j}_{\q,\triangleleft}$};
      \fi\fi
    \else
        \ifnum\i=\rows
            \ifnum\j=1
            \else
                \ifnum\j=\cols
                \else
                    % Bottom terminals (bottom row)
                    \pgfmathtruncatemacro{\q}{\j-1} % left to right 1..6
                    \node[terminal] (v\i\j) at (\j,\rows-\i) {};
                    % \node[below=2pt of v\i\j] {$h^{i,j}_{\q,\triangleright}$};
                \fi
            \fi
        \else
            % Internal vertices (same as before)
            \ifnum\j=1
                \ifnum\i=2
                    \node[vertex] (v\i\j) at (\j,\rows-\i) {};
                \else
                    \ifnum\i=3
                        \node[vertex] (v\i\j) at (\j,\rows-\i) {};
                    \else
                        \ifnum\i=4
                            \node[vertex] (v\i\j) at (\j,\rows-\i) {};
                        \else
                            \node[vertex] (v\i\j) at (\j,\rows-\i) {};
                        \fi
                    \fi
                \fi
            \else
                \ifnum\j=8
                    \ifnum\i=2
                        \node[vertex] (v\i\j) at (\j,\rows-\i) {};
                    \else
                        \ifnum\i=3
                            \node[vertex] (v\i\j) at (\j,\rows-\i) {};
                        \else   
                            \ifnum\i=4
                                \node[vertex] (v\i\j) at (\j,\rows-\i) {};
                            \else
                                \node[vertex] (v\i\j) at (\j,\rows-\i) {};
                            \fi
                        \fi
                    \fi
                \else
                    \node[vertex] (v\i\j) at (\j,\rows-\i) {};
                \fi
            \fi
        \fi
    \fi
  }
}

% Draw horizontal edges (skip first/last column)
\foreach \i in {2,...,\numexpr\rows-1\relax}{
  \foreach \j in {1,...,\numexpr\cols-1\relax}{
    \pgfmathtruncatemacro{\jnext}{\j+1}
    \ifnum\i=3
        \ifnum\j<4
            \draw[very thick,draw=violet!80!white] (v\i\j) -- (v\i\jnext);
        \else
            \draw (v\i\j) -- (v\i\jnext);
        \fi
    \else
        \ifnum\i=4
            \ifnum\j>3
                \draw[very thick,draw=violet!80!white] (v\i\j) -- (v\i\jnext);
            \else
                \draw (v\i\j) -- (v\i\jnext);
            \fi
        \else
            \draw (v\i\j) -- (v\i\jnext);
        \fi
    \fi
  }
}

% Draw vertical edges (skip first/last row)
\foreach \i in {1,...,\numexpr\rows-1\relax}{
  \foreach \j in {2,...,\numexpr\cols-1\relax}{
    \pgfmathtruncatemacro{\inext}{\i+1}
    \ifnum\j<5
        \ifnum\j=4
            \ifnum\i=3
                \draw[very thick, draw=violet!80!white] (v\i\j) -- (v\inext\j);
            \else
                \ifnum\i<3
                    \draw[thick, draw=red!40!white] (v\i\j) -- (v\inext\j);
                \else
                    \draw[ thick, draw=blue!40!white] (v\i\j) -- (v\inext\j);
                \fi
            \fi
        \else
            \ifnum\i<3
                \draw[ thick, draw=red!40!white] (v\i\j) -- (v\inext\j);
            \else
                \draw[ thick, draw=blue!40!white] (v\i\j) -- (v\inext\j);
            \fi
        \fi
    \else
        \ifnum\i<4
            \draw[very thick, draw=red!40!white] (v\i\j) -- (v\inext\j);
        \else
            \draw[very thick, draw=blue!40!white] (v\i\j) -- (v\inext\j);
        \fi
    \fi
  }
}

% Add \vdots at specified positions
\node[draw=none, fill=none, inner sep=0pt, yshift=3pt] at (1, \rows-3.5) {$\vdots$};
\node[draw=none, fill=none, inner sep=0pt, yshift=3pt] at (\cols, \rows-3.5) {$\vdots$};

\end{tikzpicture}

%% file: Conclusion.tex
\label{sec:Conclusion}
In this paper, we introduced several power network design problems that arose from real-world applications. Despite their similarities to \textsc{Steiner Tree}, we showed that the parameterization by the number of terminals and solution trees is W[1]-hard. We leave the exploration of other parameterizations, such as treewidth, to future work. 

A distinguishing feature of our power network design problems, which are shared by \textsc{Minimum-Cost Shared Arborescence}~\cite{Alvarez-Miranda17}, is that a solution consists of multiple parts and a suitable sharing of resources among the parts leads to a lower overall cost. It would be interesting to introduce this setting to other problems. For example, consider the \textsc{Subset Traveling Salesperson Problem} in which the goal is to find a minimum-cost closed walk that visits all terminals. One could investigate the version in which we look for a set of closed walks that together visit all terminals, with the number of terminals visited by a single walk upper-bounded by a given threshold, while minimizing a weighted sum of the total length of the individual walks and the \emph{visiting costs} of the edges in the union of the walks. Which algorithmic approaches for solving \textsc{Subset TSP} carry over to this setting?